\documentclass[5p,times,twocolumn,authoryear]{elsarticle}

\biboptions{authoryear,round}

\usepackage{amsmath}
\usepackage{amssymb}
\usepackage{amsfonts}
\usepackage{amsthm}
\usepackage{mathtools}
\usepackage{bm}
\usepackage{mathrsfs}
\newcommand{\be}{\mathbf{e}}

\newcommand{\spow}[2]{%
    \operatorname{sig}^{#2}\!\left(#1\right)
}

\usepackage{graphicx}
\usepackage{float}
\usepackage{caption}
\usepackage{subcaption}
\usepackage{booktabs}
\usepackage{array}
\usepackage{multirow}

\usepackage{algorithmic}

\usepackage{url}

\usepackage{textcomp}
\usepackage{xcolor}

\usepackage{titlesec}
\usepackage{etoolbox}

\usepackage[
    colorlinks=true,
    linkcolor=blue,
    citecolor=blue,
    urlcolor=blue
]{hyperref}

\def\BibTeX{%
    {\rm B\kern-.05em%
    {\sc i\kern-.025em b}%
    \kern-.08em%
    T\kern-.1667em%
    \lower.7ex\hbox{E}%
    \kern-.125emX}%
}

\newcommand{\bx}{\bm{x}}
\newcommand{\bu}{\bm{u}}
\newcommand{\bfv}{\bm{f}}
\newcommand{\bg}{\bm{g}}
\newcommand{\bW}{\bm{W}}
\newcommand{\bphi}{\bm{\phi}}
\newcommand{\bsigma}{\bm{\sigma}}
\newcommand{\bR}{\bm{R}}

\newtheoremstyle{compactthm}
{0.4pt}          
{0.4pt}          
{\itshape}       
{0pt}            
{\bfseries}      
{.}              
{0.30em}         
{}               

\newtheoremstyle{compactdef}
{0.4pt}
{0.4pt}
{\normalfont}
{0pt}
{\bfseries}
{.}
{0.30em}
{}

\theoremstyle{compactthm}

\newtheorem{theorem}{Theorem}
\newtheorem{lemma}{Lemma}
\newtheorem{proposition}{Proposition}
\newtheorem{corollary}{Corollary}

\theoremstyle{compactdef}

\newtheorem{assumption}{Assumption}
\newtheorem{definition}{Definition}
\newtheorem{remark}{Remark}

\titlespacing*{\section}
{0pt}
{0.8pt plus 0.2pt minus 0.2pt}
{0.15pt}

\titlespacing*{\subsection}
{0pt}
{0.6pt plus 0.2pt minus 0.2pt}
{0.10pt}

\titlespacing*{\subsubsection}
{0pt}
{0.5pt plus 0.2pt minus 0.2pt}
{0.10pt}

\AtBeginDocument{%
    \setlength{\abovedisplayskip}
    {0.8pt plus 0.2pt minus 0.2pt}%

    \setlength{\belowdisplayskip}
    {0.8pt plus 0.2pt minus 0.2pt}%

    \setlength{\abovedisplayshortskip}
    {0pt plus 0.2pt minus 0.1pt}%

    \setlength{\belowdisplayshortskip}
    {0.4pt plus 0.2pt minus 0.1pt}%

    \setlength{\jot}{0.4pt}%
}

\BeforeBeginEnvironment{equation}{\vspace{-0.6pt}}
\AfterEndEnvironment{equation}{\vspace{-0.6pt}}

\BeforeBeginEnvironment{equation*}{\vspace{-0.6pt}}
\AfterEndEnvironment{equation*}{\vspace{-0.6pt}}

\BeforeBeginEnvironment{align}{\vspace{-0.6pt}}
\AfterEndEnvironment{align}{\vspace{-0.6pt}}

\BeforeBeginEnvironment{align*}{\vspace{-0.6pt}}
\AfterEndEnvironment{align*}{\vspace{-0.6pt}}

\BeforeBeginEnvironment{flalign}{\vspace{-0.6pt}}
\AfterEndEnvironment{flalign}{\vspace{-0.6pt}}

\BeforeBeginEnvironment{flalign*}{\vspace{-0.6pt}}
\AfterEndEnvironment{flalign*}{\vspace{-0.6pt}}

\BeforeBeginEnvironment{gather}{\vspace{-0.6pt}}
\AfterEndEnvironment{gather}{\vspace{-0.6pt}}

\BeforeBeginEnvironment{gather*}{\vspace{-0.6pt}}
\AfterEndEnvironment{gather*}{\vspace{-0.6pt}}

\BeforeBeginEnvironment{theorem}{\vspace{-1.0pt}}
\AfterEndEnvironment{theorem}{\vspace{-1.0pt}}

\BeforeBeginEnvironment{lemma}{\vspace{-1.0pt}}
\AfterEndEnvironment{lemma}{\vspace{-1.0pt}}

\BeforeBeginEnvironment{proposition}{\vspace{-1.0pt}}
\AfterEndEnvironment{proposition}{\vspace{-1.0pt}}

\BeforeBeginEnvironment{corollary}{\vspace{-1.0pt}}
\AfterEndEnvironment{corollary}{\vspace{-1.0pt}}

\BeforeBeginEnvironment{assumption}{\vspace{-1.0pt}}
\AfterEndEnvironment{assumption}{\vspace{-1.0pt}}

\BeforeBeginEnvironment{definition}{\vspace{-1.0pt}}
\AfterEndEnvironment{definition}{\vspace{-1.0pt}}

\BeforeBeginEnvironment{remark}{\vspace{-1.0pt}}
\AfterEndEnvironment{remark}{\vspace{-1.0pt}}

\BeforeBeginEnvironment{proof}{\vspace{-1.2pt}}
\AfterEndEnvironment{proof}{\vspace{-1.2pt}}

\BeforeBeginEnvironment{figure}{\vspace{-0.8pt}}
\AfterEndEnvironment{figure}{\vspace{-1.0pt}}

\BeforeBeginEnvironment{figure*}{\vspace{-0.8pt}}
\AfterEndEnvironment{figure*}{\vspace{-1.0pt}}

\BeforeBeginEnvironment{table}{\vspace{-0.8pt}}
\AfterEndEnvironment{table}{\vspace{-1.0pt}}

\BeforeBeginEnvironment{table*}{\vspace{-0.8pt}}
\AfterEndEnvironment{table*}{\vspace{-1.0pt}}

\makeatletter
\g@addto@macro\@floatboxreset{%
    \setlength{\abovecaptionskip}{1pt}%
    \setlength{\belowcaptionskip}{0pt}%
}
\makeatother

\begin{document}

\begin{frontmatter}


\title{Predefined-Time Integral Reinforcement Learning for Unknown Nonlinear Systems via Inverse-Optimal Design}




\author[inst1]{Tien Dat Vu\corref{cor1}}

\cortext[cor1]{Corresponding author}

\ead{ dat.vuv@hcmut.edu.vn}

\address[inst1]{Faculty of Mechanical Engineering,
Ho Chi Minh City University of Technology (HCMUT),
Vietnam National University Ho Chi Minh City (VNU-HCM),
Ho Chi Minh City, Vietnam}

\begin{abstract}
This paper develops a new predefined-time integral reinforcement learning framework for optimal control of unknown nonlinear systems. The unknown drift is first approximated by a radial basis function (RBF) neural network, together with a data-driven online identification law for updating the corresponding neural weights. After the identifier converges to a sufficiently small neighborhood of the true dynamics, the learned model is incorporated into the integral reinforcement learning (IRL) problem. Unlike conventional reinforcement-learning-based optimal control, the desired convergence time is introduced directly into the control objective: a Lyapunov function and its prescribed decay behavior are specified by the designer, and inverse-optimal control is then used to construct a compatible running cost whose optimal policy inherits the predefined-time stabilization property. The value function is approximated by a second RBF neural network, and a new critic update law is developed to impose predefined-time convergence on the critic weights. Finite informative learning data are stored in a replay buffer and reused during the critic update, thereby avoiding the need for persistent excitation throughout the closed-loop operation. Theoretical analysis proves that the critic-weight error enters a prescribed residual set within the allocated learning horizon, while the closed-loop state reaches a small neighborhood of the origin within the overall designer-specified deadline. Numerical simulations on an unknown nonlinear system verify accurate drift reconstruction, predefined-time critic learning, and closed-loop convergence.
\end{abstract}


\end{frontmatter}

\section{Introduction}
\label{sec:introduction}

Modern cyber--physical and networked control systems
increasingly involve unknown nonlinear dynamics,
limited model information, and demanding performance
requirements, making conventional model-based optimal
control difficult to implement
\citep{Wan2021CPS,Lee2015CPS}.
Although optimal control provides a systematic
performance-oriented design framework, nonlinear
problems generally require solving the
Hamilton--Jacobi--Bellman (HJB) equation, whose
analytical solution is rarely available, especially
when the drift dynamics are unknown
\citep{Doya2000,Vamvoudakis2010}.
The problem becomes more challenging when convergence
must be guaranteed within a designer-specified time
independently of the initial condition
\citep{Polyakov2012,Jimenez2017Predefined,
Kokolakis2025Predefined}.

Reinforcement learning (RL) and adaptive dynamic
programming (ADP) provide data-driven mechanisms for
approximating HJB solutions
\citep{Doya2000,Vrabie2009,Vamvoudakis2010}.
Integral reinforcement learning (IRL) further
formulates Bellman equations from measured trajectory
integrals, reducing explicit dependence on unknown
drift dynamics
\citep{Modares2014,Vamvoudakis2014IntegralRL}.
In parallel, finite-time, fixed-time, and
predefined-time stabilization provide convergence
guarantees with different settling-time properties
\citep{BhatBernstein2000FiniteTime,Polyakov2012,
Jimenez2017Predefined}.
These developments motivate optimal learning schemes
that simultaneously address unknown dynamics,
prescribed convergence requirements, and
closed-loop stability.

A fundamental approach to connecting Lyapunov
stabilization with optimal control is inverse-optimal
design, in which a performance index is reconstructed
from a prescribed stabilizing feedback and Lyapunov
behavior
\citep{FreemanKokotovic1996Inverse,
Haddad2008NonlinearControl,Jimenez2017Predefined}.
However, prescribing a Lyapunov decay does not
automatically guarantee the existence of a feasible
control realizing it or a running cost compatible
with the HJB equation.
Moreover, multiple admissible inputs may satisfy
the same convergence requirement with different
control efforts.
A constructive design must therefore address
feasibility, HJB consistency, and control effort
simultaneously, particularly when actuator limitations
are relevant
\citep{Petersen1987MinimumEffort,
Primbs1999CLFMinNorm}.
For unknown nonlinear systems, this construction
introduces an additional difficulty because the
drift appears explicitly in the Lyapunov derivative
and the resulting running cost
\citep{Jimenez2017Predefined}.

Recent studies have addressed these difficulties
through different learning architectures.
Kokolakis et al.~\citep{Kokolakis2023FixedTime,
Kokolakis2025Predefined} developed fixed- and
predefined-time critic learning using non-Lipschitz
updates and experience replay, while related
extensions considered safe optimal stabilization,
nonlinear differential games, and constrained
predefined-time control
\citep{Kokolakis2025SafePINN,WANG2026112673,
Kharrat2026Constraints}.
Fixed-time reinforcement learning has also been
investigated for uncertain nonlinear multi-agent
formation and secure coordination
\citep{Wang2024FixedTimeFormation,Gong2025}.
Despite their different convergence and safety
guarantees, these formulations do not directly
provide a general constructive procedure that
simultaneously reconstructs a prescribed-time
running cost for unknown drift dynamics and learns
its associated policy from actual integral Bellman
data.
In particular, the pointwise HJB residual employed
in~\citep{Kokolakis2025Predefined} retains explicit
system-dynamics terms, limiting direct application
when the nonlinear drift is unavailable.

Another important issue concerns the consistency
of optimality under model uncertainty.
Identifier-based optimal learning has been developed
to estimate unknown nonlinear dynamics and support
policy optimization
\citep{Bhasin2013ACI}.
More recently, Cai et al.~\citep{Cai2026} combined
integral identification with value iteration for
predefined-time optimal control of unknown nonlinear
systems.
However, when the HJB equation is solved directly
using identified dynamics, the resulting solution
need not satisfy the HJB equation of the actual
system, since identification errors may introduce
a mismatch between the nominal and actual
Hamiltonians.
Consequently, convergence of an optimization
algorithm on the identified model, even together
with closed-loop stability, does not automatically
establish optimality for the physical system.
This distinction concerns the consistency of the
optimal solution rather than merely the size of
the closed-loop residual set.

Integral concurrent learning (ICL) provides a
mechanism for reconstructing unknown drift dynamics
from finite trajectory data without state
differentiation or continuous persistent excitation
\citep{Chowdhary2010CL,Parikh2019ICL}.
Similarly, Modares et al.~\citep{Modares2014}
combined IRL with experience replay to relax the
persistent-excitation requirement, but did not
establish predefined-time learning convergence.
These results suggest separating two mathematical
tasks: using identified dynamics to construct an
inverse-optimal performance index and using actual
system trajectories to learn the corresponding
policy through integral Bellman equations.
Unlike solving the HJB equation exclusively on an
approximate model, the latter formulation avoids
explicitly substituting the identified drift into
the exact integral Bellman relation.
Nevertheless, identification, approximation, and
policy-learning errors must still be accounted for
in the stability analysis.

The remaining challenge is therefore to establish
a unified methodology linking finite-data drift
reconstruction, feasible inverse-optimal cost
construction, and predefined-time critic learning,
while guaranteeing convergence of the coupled
state--critic dynamics within a designer-assigned
deadline.
In particular, predefined-time convergence of the
critic alone does not guarantee the same property
for the closed-loop state, and the total convergence
analysis must explicitly accommodate the remaining
learning and identification errors
\citep{Jimenez2017Predefined,
Kokolakis2025Predefined,Parikh2019ICL}.

To address these issues, the present study employs
integral identification to reconstruct the unknown
drift, which is subsequently used to establish
a predefined-time inverse-optimal running cost.
Once the performance index is fixed, a critic-only
IRL scheme uses current and replayed integral
Bellman data obtained from the actual system.
This separation enables model-assisted cost
construction without treating the identified
dynamics as an exact substitute for the physical
system during policy learning.
A composite Lyapunov analysis then connects the
learning and stabilization horizons while retaining
the effects of identification and approximation
errors.

The main contributions of this paper are threefold.
\begin{enumerate}

\item A unified predefined-time IRL framework is
developed for unknown nonlinear systems, combining
online drift identification, integral concurrent
learning, and inverse-optimal cost construction
with policy learning from actual system trajectories.
A finite-data least-squares estimate initializes
the online identifier without assuming an exact
reconstructed model.

\item A new critic-only IRL update combines current
and replayed finite-interval integral Bellman data.
Unlike pointwise HJB evaluation in
\citep{Kokolakis2025Predefined}, the update does
not explicitly require identified drift dynamics
and guarantees practical predefined-time critic
weight convergence under finite-data informativeness,
without continuous persistent excitation.

 \item Theorem~\ref{thm:inverse_optimal_pt} develops an explicit procedure for
  reconstructing the nominal feedback policy and
  running-cost components from a predefined-time decay law,
  subject to input-feasibility conditions and
  compatibility with the nominal HJB equation.

\end{enumerate}

\textit{Notation:}
Throughout this paper,
\(
\mathbb R
\),
\(
\mathbb R_{\geq0}
\),
and
\(
\mathbb R_{>0}
\)
denote the real, nonnegative real, and positive real numbers, respectively.
For
\(
\bx\in\mathbb R^n
\)
and
\(
\bm A\in\mathbb R^{n\times m}
\),
\(
\|\bx\|
\)
denotes the Euclidean norm,
\(
\|\bm A\|
\)
the induced \(2\)-norm,
and
\(
\bm A^\top
\)
the transpose of
\(
\bm A
\).
For a symmetric matrix
\(
\bm A
\),
\(
\lambda_{\min}(\bm A)
\)
and
\(
\lambda_{\max}(\bm A)
\)
denote its minimum and maximum eigenvalues, respectively, while
\(
\bm A\succ0
\)
and
\(
\bm A\succeq0
\)
denote positive definiteness and positive semidefiniteness.
The symbols
\(
\bm I_n
\),
\(
\operatorname{diag}\{\cdot\}
\),
and
\(
\operatorname{col}\{\cdot\}
\)
denote the
\(
n\times n
\)
identity matrix, diagonal concatenation, and column stacking, respectively.
For a continuously differentiable scalar function
\(
V:\mathbb R^n\to\mathbb R
\),
\(
\nabla V(\bx)
\)
denotes its gradient, and for a vector-valued map
\(
\bphi:\mathbb R^n\to\mathbb R^N
\),
\(
\nabla\bphi(\bx)
\)
denotes its Jacobian.
For
\(
s\in\mathbb R
\)
and
\(
r>0
\),
define
\(
\operatorname{sig}^{r}(s)
:=
|s|^{r}\operatorname{sgn}(s)
\),
with the operation understood componentwise for vectors.
The Gamma function is denoted by
\(
\Gamma_E(z)
:=
\int_{0}^{\infty}
t^{z-1}e^{-t}\,dt
\),
for
\(
z>0
\). The notation \(\mathcal{O}(\cdot)\) denotes the standard Big-O order notation.

\section{Problem Formulation and Data-Driven System Identification}
\label{sec:problem_identification}

\subsection{Problem Formulation}

Consider the nonlinear control-affine system
\begin{equation}
    \dot{\bx}
    =
    \bfv(\bx)
    +
    \bg(\bx)\bu,
    \label{eq:plant}
\end{equation}
where \(\bx\in\mathbb R^n\) is the system state, \(\bu\in\mathbb R^m\) is the control input, \(\bfv:\mathbb R^n\rightarrow\mathbb R^n\) is an unknown nonlinear drift, and \(\bg:\mathbb R^n\rightarrow\mathbb R^{n\times m}\) is a known input matrix. The control objective is to learn the unknown drift from measured state--input data and subsequently construct an optimal feedback policy that drives the state to the origin within a uniform finite settling time.

\begin{definition}[Admissible policy~\citep{AbuKhalaf2005}]
\label{def:admissible_secure_policy}
A continuous policy \(\bu:\Omega\rightarrow\mathcal U\) is admissible
on \(\Omega\), denoted by \(\bu\in\Psi(\Omega)\), if
\(\bu(\bm0)=\bm0\), the resulting closed-loop system is
asymptotically stable on \(\Omega\) in the absence of exogenous signals,
and the associated infinite-horizon cost is finite .
\end{definition}

\begin{assumption}
\label{ass:plant}
Let \(\Omega\subset\mathbb R^n\) be a compact set containing the origin in its interior. The unknown drift \(\bfv(\bx)\) and the known input matrix \(\bg(\bx)\) are locally Lipschitz on an open neighborhood of \(\Omega\), with \(\bfv(\bm 0)=\bm 0\). For every admissible input considered herein, the corresponding solution of \eqref{eq:plant} exists uniquely while it remains in \(\Omega\).
\end{assumption}


\begin{remark}[Compact-domain limitation]
The compact set $\Omega$ is essential for the
subsequent identification, approximation, and
stability analysis, as commonly required in
neural ADP \cite{Vamvoudakis2010,Modares2014}.
However, identifying such a set a priori is
nontrivial, and ensuring that trajectories
remain inside it may require a stabilizing
baseline or backup controller
\cite{Cohen2023SafeExploration,Rabiee2025BackupCBF}.
Even then, the size of the certified invariant
region may be difficult to characterize or
conservative \cite{Rabiee2025BackupCBF}.
Consequently, stability outside this region
is not guaranteed by the local analysis.
Control barrier functions offer a systematic
approach to enforcing forward invariance of
a prescribed admissible set
\cite{Ames2017CBF,Cohen2023SafeExploration}.
For analytical tractability, the present work
retains the compact operating domain as an
assumption.
\end{remark}

\begin{theorem}[Normalized predefined-time comparison~\citep{Kokolakis2025Predefined}]
\label{thm:normalized_predefined_time}
Let \(V:\mathbb R^n\to\mathbb R_{\geq0}\) be positive definite, radially unbounded, and locally absolutely continuous along the trajectories of a forward-complete system. Let \(\alpha,\beta,p,q,r,T_s>0\) satisfy
\begin{equation}
0<p<q,\qquad pr<1<qr,
\label{eq:normalized_exponent_conditions}
\end{equation}
and define
\begin{equation}
\begin{aligned}
\gamma_{p,q,r}
:={}&
\frac{
\Gamma_{\!E}\!\left(\dfrac{1-pr}{q-p}\right)
\Gamma_{\!E}\!\left(\dfrac{qr-1}{q-p}\right)
}{
\alpha^r\Gamma_{\!E}(r)(q-p)
}
\left(\frac{\alpha}{\beta}\right)^{
\frac{1-pr}{q-p}}.
\end{aligned}
\label{eq:gamma_normalization}
\end{equation}
Here, \(\Gamma_{\!E}(z):=\int_0^\infty t^{z-1}e^{-t}dt\), \(z>0\), is the Euler gamma function. If, for almost all \(t\) with \(V(\bx(t))>0\),
\begin{equation}
\dot V(\bx(t))
\leq
-\frac{\gamma_{p,q,r}}{T_s}
\left[
\alpha V^p(\bx(t))
+
\beta V^q(\bx(t))
\right]^r,
\label{eq:normalized_predefined_dissipation}
\end{equation}
then the origin is globally predefined-time stable and
\begin{equation}
T(\bx_0)\leq T_s,
\qquad
\forall\,\bx_0\in\mathbb R^n.
\label{eq:normalized_settling_conclusion}
\end{equation}
\end{theorem}

\subsection{Neural Approximation of the Unknown Dynamics}

Since \(\bfv(\bx)\) is unavailable for control design, it is first reconstructed from measured data. By the universal approximation property~\cite{Kokolakis2025Predefined}, each continuous drift on the compact set \(\Omega\) can be represented as
\begin{equation}
    \bfv(\bx)
    =
    \bW_f^{*\top}\bsigma_f(\bx)
    +
    \bm{\varepsilon}_f(\bx),
    \label{eq:f_approximation}
\end{equation}
where \(\bsigma_f(\bx)\in\mathbb R^{N_f}\) is a prescribed vector of continuously differentiable basis functions, \(\bW_f^*\in\mathbb R^{N_f\times n}\) is the unknown ideal weight matrix, and \(\bm{\varepsilon}_f(\bx)\in\mathbb R^n\) is the approximation error. On \(\Omega\), \(\|\bsigma_f(\bx)\|\le\bar\sigma_f\) and \(\|\bm{\varepsilon}_f(\bx)\|\le\bar\varepsilon_f\) for some finite constants \(\bar\sigma_f>0\) and \(\bar\varepsilon_f>0\).

\begin{assumption}[Local regularity of the identifier basis]
\label{ass:basis_origin}
The basis vector \(\bm{\sigma}_f(\bm{x})\) satisfies
\(\bm{\sigma}_f(\bm{0})=\bm{0}\)
and is locally Lipschitz in a neighborhood of the origin. Hence, there exist
constants \(c_\sigma>0\) and \(\delta_\sigma>0\) such that
\(\|\bm{\sigma}_f(\bm{x})\|
\leq
c_\sigma\|\bm{x}\|\)
for all \(\|\bm{x}\|\leq\delta_\sigma\).
\end{assumption}

The identified drift is subsequently represented by
\(
\hat{\bfv}(\bx)
=
\hat{\bW}_f^\top\bsigma_f(\bx)
\),
where
\(
\hat{\bW}_f
\)
denotes the estimated weight matrix associated with the prescribed basis
vector
\(
\bsigma_f(\bx)
\).

\subsection{Integral Composite Concurrent-Learning Identification}
\label{subsec:integral_identification}

Let
\(
\tau_0<\tau_1<\cdots
\)
denote the sampling instants used to construct integral identification
data. Integrating \eqref{eq:plant} over
\(
[\tau_\ell,\tau_{\ell+1}]
\)
and using \eqref{eq:f_approximation} gives
\begin{equation}
\begin{aligned}
    &\bx(\tau_{\ell+1})-\bx(\tau_\ell)
    -
    \int_{\tau_\ell}^{\tau_{\ell+1}}
    \bg(\bx(\tau))\bu(\tau)\,d\tau
    \\
    &\qquad =
    \bW_f^{*\top}
    \int_{\tau_\ell}^{\tau_{\ell+1}}
    \bsigma_f(\bx(\tau))\,d\tau
    +
    \int_{\tau_\ell}^{\tau_{\ell+1}}
    \bm{\varepsilon}_f(\bx(\tau))\,d\tau .
\end{aligned}
\label{eq:integral_identification}
\end{equation}

For each interval, define
\(
\bm y_{f,\ell}
:=
\bx(\tau_{\ell+1})-\bx(\tau_\ell)
-
\int_{\tau_\ell}^{\tau_{\ell+1}}
\bg(\bx(\tau))\bu(\tau)\,d\tau
\),
\(
\bm\phi_{f,\ell}
:=
\int_{\tau_\ell}^{\tau_{\ell+1}}
\bsigma_f(\bx(\tau))\,d\tau
\),
and
\(
\bm d_{f,\ell}
:=
\int_{\tau_\ell}^{\tau_{\ell+1}}
\bm{\varepsilon}_f(\bx(\tau))\,d\tau
\).
Then
\(
\bm y_{f,\ell}
=
\bW_f^{*\top}\bm\phi_{f,\ell}
+
\bm d_{f,\ell}
\).

The informative data are stored in the history stack
\(
\mathcal H_f
=
\{(\bm y_{f,\ell},\bm\phi_{f,\ell})\}_{\ell=0}^{M_f-1}
\),
with cumulative information matrix
\(
\bm{\mathcal S}_f
:=
\sum_{\ell=0}^{M_f-1}
\bm\phi_{f,\ell}\bm\phi_{f,\ell}^{\top}
\).
Whenever a new data segment is admitted,
\(
\bm{\mathcal S}_f
\)
is updated by
\(
\bm{\mathcal S}_f^{+}
=
\bm{\mathcal S}_f
+
\bm\phi_{f,\mathrm{new}}
\bm\phi_{f,\mathrm{new}}^{\top}
\).

\begin{assumption}[Finite excitation]
\label{ass:finite_excitation}
There exist a finite \(M_f\geq N_f\) and a constant
\(\alpha_f>0\) such that the stored integral data satisfy
\(
\bm{\mathcal S}_f
=
\sum_{\ell=0}^{M_f-1}
\bm\phi_{f,\ell}\bm\phi_{f,\ell}^{\top}
\succeq
\alpha_f\bm I
\).
\end{assumption}

In addition to the stored data, define the current integral data over a
moving window of length \(T_f>0\) by
\(
\bm y_f(t)
:=
\bx(t)-\bx(t-T_f)
-
\int_{t-T_f}^{t}
\bg(\bx(\tau))\bu(\tau)\,d\tau
\)
and
\(
\bm\phi_f(t)
:=
\int_{t-T_f}^{t}
\bsigma_f(\bx(\tau))\,d\tau
\).
From \eqref{eq:plant} and \eqref{eq:f_approximation}, these quantities
satisfy
\(
\bm y_f(t)
=
\bW_f^{*\top}\bm\phi_f(t)
+
\bm d_f(t)
\),
where
\(
\bm d_f(t)
:=
\int_{t-T_f}^{t}
\bm{\varepsilon}_f(\bx(\tau))\,d\tau
\).
Similarly, every stored pair satisfies
\(
\bm y_{f,\ell}
=
\bW_f^{*\top}\bm\phi_{f,\ell}
+
\bm d_{f,\ell}
\).

Using the current estimate \(\hat{\bW}_f(t)\), define the current and
stored integral prediction errors as
\(
\bm e_f(t)
:=
\bm y_f(t)
-
\hat{\bW}_f^\top(t)\bm\phi_f(t)
\)
and
\(
\bm e_{f,\ell}(t)
:=
\bm y_{f,\ell}
-
\hat{\bW}_f^\top(t)\bm\phi_{f,\ell}
\),
respectively. Consider the composite prediction-error objective
\(
J_f
:=
\frac{1}{2}\|\bm e_f(t)\|^2
+
\frac{k_f}{2}
\sum_{\ell=0}^{M_f-1}
\|\bm e_{f,\ell}(t)\|^2
\),
where \(k_f>0\). Since
\(
\nabla_{\hat{\bW}_f}J_f
=
-\bm\phi_f(t)\bm e_f^\top(t)
-
k_f
\sum_{\ell=0}^{M_f-1}
\bm\phi_{f,\ell}\bm e_{f,\ell}^\top(t)
\),
the gradient update
\(
\dot{\hat{\bW}}_f
=
-\bm{\Gamma}_f
\nabla_{\hat{\bW}_f}J_f
\)
gives
\begin{equation}
\dot{\hat{\bW}}_f
=
\bm{\Gamma}_f
\left[
    \bm\phi_f(t)\bm e_f^\top(t)
    +
    k_f
    \sum_{\ell=0}^{M_f-1}
    \bm\phi_{f,\ell}\bm e_{f,\ell}^\top(t)
\right].
\label{eq:icl_identifier}
\end{equation}
where
\(
\bm{\Gamma}_f=\bm{\Gamma}_f^\top\succ0
\).
Hence, the identifier combines the current moving-window information
with the informative historical data retained in the history stack.

Let
\(
\tilde{\bW}_f
:=
\hat{\bW}_f-\bW_f^*
\).
The above regression relations immediately yield
\(
\bm e_f
=
-\tilde{\bW}_f^\top\bm\phi_f+\bm d_f
\)
and
\(
\bm e_{f,\ell}
=
-\tilde{\bW}_f^\top\bm\phi_{f,\ell}
+
\bm d_{f,\ell}
\).
Accordingly, the identification dynamics are analyzed with
\(
V_f
:=
\frac{1}{2}
\operatorname{tr}
\left(
\tilde{\bW}_f^\top
\bm{\Gamma}_f^{-1}
\tilde{\bW}_f
\right)
\).

\begin{lemma}[Composite identification convergence]
\label{lem:identification_error}
Suppose that Assumptions~\ref{ass:plant}--\ref{ass:finite_excitation}
hold and that the integral approximation residuals are bounded.
Then the weight-estimation error generated by
\eqref{eq:icl_identifier} is uniformly ultimately bounded. In
particular, there exist constants \(c_f>0\) and \(\Delta_f\geq0\)
such that
\begin{equation}
    \dot V_f
    \leq
    -c_fV_f+\Delta_f .
    \label{eq:icl_weight_bound}
\end{equation}
Consequently, the reconstructed drift
\(
\hat{\bfv}(\bx,t)
=
\hat{\bW}_f^\top(t)\bsigma_f(\bx)
\)
satisfies
\(
\|\bfv(\bx)-\hat{\bfv}(\bx,t)\|
\leq
\bar\sigma_f\|\tilde{\bW}_f(t)\|_F
+
\bar\varepsilon_f
\)
for every
\(
\bx\in\Omega
\).
If
\(
\bm{\varepsilon}_f\equiv\bm0
\),
then
\(
\Delta_f=0
\)
and
\(
\tilde{\bW}_f(t)\to\bm0
\).
\end{lemma}



\begin{proof}
Differentiating \(V_f\) along \eqref{eq:icl_identifier}, and noting that the ideal weight matrix is constant so that
\(\dot{\tilde{\bW}}_f=\dot{\hat{\bW}}_f\), gives
\begin{equation}
\dot V_f
=
\operatorname{tr}
[
\tilde{\bW}_f^\top
\bm\phi_f\bm e_f^\top
]
+
k_f
\sum_{\ell=0}^{M_f-1}
\operatorname{tr}
[
\tilde{\bW}_f^\top
\bm\phi_{f,\ell}
\bm e_{f,\ell}^{\top}
].
\label{eq:Vf_dot_initial}
\end{equation}
From the current regression relation and its stored counterparts, the corresponding prediction errors can be written as
\(
\bm e_f
=
-\tilde{\bW}_f^\top\bm\phi_f+\bm d_f
\)
and
\(
\bm e_{f,\ell}
=
-\tilde{\bW}_f^\top\bm\phi_{f,\ell}
+
\bm d_{f,\ell}
\),
respectively. Substituting these expressions into \eqref{eq:Vf_dot_initial} yields
\begin{equation}
\begin{aligned}
\dot V_f
={}&
-\operatorname{tr}
[
\tilde{\bW}_f^\top
\bm\phi_f
\bm\phi_f^\top
\tilde{\bW}_f
]
-
k_f
\sum_{\ell=0}^{M_f-1}
\operatorname{tr}
[
\tilde{\bW}_f^\top
\bm\phi_{f,\ell}
\bm\phi_{f,\ell}^\top
\tilde{\bW}_f
]
\\
&+
\operatorname{tr}
[
\tilde{\bW}_f^\top
\bm\phi_f
\bm d_f^\top
]
+
k_f
\sum_{\ell=0}^{M_f-1}
\operatorname{tr}
[
\tilde{\bW}_f^\top
\bm\phi_{f,\ell}
\bm d_{f,\ell}^\top
].
\end{aligned}
\label{eq:Vf_dot_expanded}
\end{equation}
Using
\(
\operatorname{tr}
[
\tilde{\bW}_f^\top
\bm\phi_f
\bm\phi_f^\top
\tilde{\bW}_f
]
=
\|\tilde{\bW}_f^\top\bm\phi_f\|^2
\)
and, similarly,
\(
\operatorname{tr}
[
\tilde{\bW}_f^\top
\bm\phi_{f,\ell}
\bm\phi_{f,\ell}^\top
\tilde{\bW}_f
]
=
\|\tilde{\bW}_f^\top\bm\phi_{f,\ell}\|^2
\),
one obtains
\begin{equation}
\begin{aligned}
\dot V_f
={}&
-\|\tilde{\bW}_f^\top\bm\phi_f\|^2
-
k_f
\sum_{\ell=0}^{M_f-1}
\|\tilde{\bW}_f^\top\bm\phi_{f,\ell}\|^2
\\
&+
\operatorname{tr}
[
\tilde{\bW}_f^\top
\bm\phi_f
\bm d_f^\top
]
+
k_f
\sum_{\ell=0}^{M_f-1}
\operatorname{tr}
[
\tilde{\bW}_f^\top
\bm\phi_{f,\ell}
\bm d_{f,\ell}^\top
].
\end{aligned}
\label{eq:Vf_dot_residual}
\end{equation}

The last two terms in \eqref{eq:Vf_dot_residual} arise from the integral approximation residuals and therefore must be bounded explicitly. By the Cauchy--Schwarz inequality for the Frobenius inner product,
\(
\left|
\operatorname{tr}
[
\tilde{\bW}_f^\top
\bm\phi_f
\bm d_f^\top
]
\right|
\leq
\|\tilde{\bW}_f^\top\bm\phi_f\|
\|\bm d_f\|
\)
and
\(
\left|
\operatorname{tr}
[
\tilde{\bW}_f^\top
\bm\phi_{f,\ell}
\bm d_{f,\ell}^\top
]
\right|
\leq
\|\tilde{\bW}_f^\top\bm\phi_{f,\ell}\|
\|\bm d_{f,\ell}\|
\).
For arbitrary constants
\(
\epsilon_{f,0}\in(0,2)
\)
and
\(
\epsilon_{f,\ell}\in(0,2)
\),
Young's inequality gives
\begin{equation}
\|\tilde{\bW}_f^\top\bm\phi_f\|
\|\bm d_f\|
\leq
\frac{\epsilon_{f,0}}{2}
\|\tilde{\bW}_f^\top\bm\phi_f\|^2
+
\frac{1}{2\epsilon_{f,0}}
\|\bm d_f\|^2,
\label{eq:young_current}
\end{equation}
and
\begin{equation}
\|\tilde{\bW}_f^\top\bm\phi_{f,\ell}\|
\|\bm d_{f,\ell}\|
\leq
\frac{\epsilon_{f,\ell}}{2}
\|\tilde{\bW}_f^\top\bm\phi_{f,\ell}\|^2
+
\frac{1}{2\epsilon_{f,\ell}}
\|\bm d_{f,\ell}\|^2.
\label{eq:young_stored}
\end{equation}
Substituting \eqref{eq:young_current} and \eqref{eq:young_stored} into \eqref{eq:Vf_dot_residual} gives
\begin{equation}
\begin{aligned}
\dot V_f
\leq{}&
-
\left(
1-\frac{\epsilon_{f,0}}{2}
\right)
\|\tilde{\bW}_f^\top\bm\phi_f\|^2
-
k_f
\sum_{\ell=0}^{M_f-1}
\left(
1-\frac{\epsilon_{f,\ell}}{2}
\right)
\|\tilde{\bW}_f^\top\bm\phi_{f,\ell}\|^2
\\
&+
\frac{1}{2\epsilon_{f,0}}
\|\bm d_f\|^2
+
\frac{k_f}{2}
\sum_{\ell=0}^{M_f-1}
\frac{1}{\epsilon_{f,\ell}}
\|\bm d_{f,\ell}\|^2.
\end{aligned}
\label{eq:Vf_dot_young}
\end{equation}

Since the integral approximation residuals are bounded, there exist finite constants
\(
\bar d_f\geq0
\)
and
\(
\bar d_{f,\ell}\geq0
\)
such that
\(
\|\bm d_f\|\leq\bar d_f
\)
and
\(
\|\bm d_{f,\ell}\|\leq\bar d_{f,\ell}
\).
Define
\(
\underline\epsilon_f
:=
\min_{\ell=0,\ldots,M_f-1}
\left(
1-\frac{\epsilon_{f,\ell}}{2}
\right)>0
\).
Then \eqref{eq:Vf_dot_young} implies
\begin{equation}
\begin{aligned}
\dot V_f
\leq{}&
-
\left(
1-\frac{\epsilon_{f,0}}{2}
\right)
\|\tilde{\bW}_f^\top\bm\phi_f\|^2
-
k_f\underline\epsilon_f
\sum_{\ell=0}^{M_f-1}
\|\tilde{\bW}_f^\top\bm\phi_{f,\ell}\|^2
+
\Delta_{f,r},
\end{aligned}
\label{eq:Vf_dot_bounded_residual}
\end{equation}
where
\(
\Delta_{f,r}
:=
\frac{\bar d_f^2}{2\epsilon_{f,0}}
+
\frac{k_f}{2}
\sum_{\ell=0}^{M_f-1}
\frac{\bar d_{f,\ell}^2}{\epsilon_{f,\ell}}
\geq0
\).

By Assumption~\ref{ass:finite_excitation},
\(
\bm{\mathcal S}_f
=
\sum_{\ell=0}^{M_f-1}
\bm\phi_{f,\ell}\bm\phi_{f,\ell}^\top
\succeq
\alpha_f I
\).
Hence,
\begin{equation}
\begin{aligned}
\sum_{\ell=0}^{M_f-1}
\|\tilde{\bW}_f^\top\bm\phi_{f,\ell}\|^2
&=
\sum_{\ell=0}^{M_f-1}
\operatorname{tr}
[
\tilde{\bW}_f^\top
\bm\phi_{f,\ell}
\bm\phi_{f,\ell}^\top
\tilde{\bW}_f
]
\\
&=
\operatorname{tr}
[
\tilde{\bW}_f^\top
\bm{\mathcal S}_f
\tilde{\bW}_f
]
\geq
\alpha_f
\operatorname{tr}
[
\tilde{\bW}_f^\top
\tilde{\bW}_f
]
\\
&=
\alpha_f
\|\tilde{\bW}_f\|_F^2.
\end{aligned}
\label{eq:stored_data_lower_bound}
\end{equation}
Using \eqref{eq:stored_data_lower_bound} in \eqref{eq:Vf_dot_bounded_residual}, and discarding the additional nonpositive current-data term, gives
\begin{equation}
\dot V_f
\leq
-
k_f\underline\epsilon_f\alpha_f
\|\tilde{\bW}_f\|_F^2
+
\Delta_{f,r}.
\label{eq:Vf_dot_weight_norm}
\end{equation}
Thus, defining
\(
\bar c_f
:=
k_f\underline\epsilon_f\alpha_f>0
\)
and
\(
\Delta_f:=\Delta_{f,r}\geq0
\),
one obtains
\begin{equation}
\dot V_f
\leq
-\bar c_f
\|\tilde{\bW}_f\|_F^2
+
\Delta_f.
\label{eq:Vf_dot_intermediate}
\end{equation}

Finally, from the definition
\(
V_f
=
\frac{1}{2}
\operatorname{tr}
[
\tilde{\bW}_f^\top
\bm{\Gamma}_f^{-1}
\tilde{\bW}_f
]
\)
and
\(
\bm{\Gamma}_f\succ0
\),
it follows that
\(
V_f
\leq
[2\lambda_{\min}(\bm{\Gamma}_f)]^{-1}
\|\tilde{\bW}_f\|_F^2
\),
or equivalently
\(
\|\tilde{\bW}_f\|_F^2
\geq
2\lambda_{\min}(\bm{\Gamma}_f)V_f
\).
Substituting this inequality into \eqref{eq:Vf_dot_intermediate} yields
\begin{equation}
\dot V_f
\leq
-
2\bar c_f
\lambda_{\min}(\bm{\Gamma}_f)
V_f
+
\Delta_f.
\label{eq:Vf_dot_final_step}
\end{equation}
Therefore, by defining
\(
c_f
:=
2\bar c_f
\lambda_{\min}(\bm{\Gamma}_f)>0
\),
one obtains
\(
\dot V_f
\leq
-c_fV_f+\Delta_f
\),
which proves \eqref{eq:icl_weight_bound}.
\end{proof}

The online reconstructed model
\(
\hat{\bfv}(\bx,t)
=
\hat{\bW}_f^\top(t)\bsigma_f(\bx)
\)
is subsequently used in the inverse-optimal construction and is
updated concurrently with the closed-loop evolution.

\begin{lemma}[Bounded identification error]
\label{lem:identification_error}
Under Assumptions~\ref{ass:plant}--\ref{ass:finite_excitation},
the identification error is uniformly bounded on \(\Omega\). In particular,
there exists a finite constant \(\bar\Delta_f>0\) such that
\begin{equation}
    \left\|
        \bfv(\bx)-\hat{\bfv}(\bx,t)
    \right\|
    \leq
    \bar\Delta_f,
    \qquad
    \forall\,\bx\in\Omega .
    \label{eq:identification_error_bound}
\end{equation}
\end{lemma}

\begin{proof}
From
\(
\bfv(\bx)
=
\bW_f^{*\top}\bsigma_f(\bx)
+
\bm{\varepsilon}_f(\bx)
\)
and
\(
\hat{\bfv}(\bx,t)
=
\hat{\bW}_f^\top(t)\bsigma_f(\bx)
\),
one has
\(
\bfv(\bx)-\hat{\bfv}(\bx,t)
=
-\tilde{\bW}_f^\top(t)\bsigma_f(\bx)
+
\bm{\varepsilon}_f(\bx)
\).
Hence,
\(
\|\bfv(\bx)-\hat{\bfv}(\bx,t)\|
\leq
\bar\sigma_f\|\tilde{\bW}_f(t)\|_F
+
\bar\varepsilon_f
\).
By Lemma~\ref{lem:identification_error},
\(
\tilde{\bW}_f(t)
\)
is uniformly ultimately bounded. Therefore, there exists a finite
constant \(\bar W_f>0\) such that
\(
\|\tilde{\bW}_f(t)\|_F\leq\bar W_f
\),
and thus \eqref{eq:identification_error_bound} holds with
\(
\bar\Delta_f
:=
\bar\sigma_f\bar W_f
+
\bar\varepsilon_f
\).
\end{proof}


\begin{remark}
Before activating the online update law \eqref{eq:icl_identifier}, the stored integral regressions satisfy \(\bm y_{f,\ell}\approx\hat{\bW}_{f}^{\top}\bm\phi_{f,\ell}\), \(\ell=0,\ldots,M_f-1\). Collecting \(M_f\) data segments gives \(\bm Y_f\approx\hat{\bW}_{f}^{\top}\bm\Phi_f\), where \(\bm Y_f=[\bm y_{f,0},\ldots,\bm y_{f,M_f-1}]\) and \(\bm\Phi_f=[\bm\phi_{f,0},\ldots,\bm\phi_{f,M_f-1}]\). Under Assumption~\ref{ass:finite_excitation}, \(\bm\Phi_f\bm\Phi_f^\top\succ0\), and therefore the initial weight estimate can be obtained from the least-squares solution \(\hat{\bW}_{f,\mathrm{LS}}^\top=\bm Y_f\bm\Phi_f^\top(\bm\Phi_f\bm\Phi_f^\top)^{-1}\), equivalently, \(\hat{\bW}_{f,\mathrm{LS}}^\top=\left(\sum_{\ell=0}^{M_f-1}\bm y_{f,\ell}\bm\phi_{f,\ell}^\top\right)\left(\sum_{\ell=0}^{M_f-1}\bm\phi_{f,\ell}\bm\phi_{f,\ell}^\top\right)^{-1}\). This batch solution is then used as the informed initialization \(\hat{\bW}_f(0)=\hat{\bW}_{f,\mathrm{LS}}\) for the subsequent online update. Nevertheless, because the integral regression is only approximate, the batch estimate should not be interpreted as the exact ideal weight \(\bW_f^*\); rather, it provides a data-informed starting point for the online identifier. This distinction is important because an arbitrary trial-and-error initialization may produce a large initial parameter mismatch \(\tilde{\bW}_f(0)=\hat{\bW}_f(0)-\bW_f^*\), whereas the least-squares initialization exploits the informative history stack before online adaptation begins. The subsequent update \eqref{eq:icl_identifier} then combines the current moving-window regression with the stored informative data to refine the weight estimate online, following the concurrent/integral-concurrent-learning philosophy in \citep{Chowdhary2010CL,Parikh2019ICL}.
\end{remark}


\section{Predefined-Time Reinforcement Learning Optimal Control}
\label{sec:predefined_time_rl}


\subsection{Predefined-Time Optimal-Control Formulation}

The normalized predefined-time comparison result in Theorem~\ref{thm:normalized_predefined_time} provides a direct mechanism for assigning a desired convergence-time bound in advance. In particular, rather than selecting the decay gains first and subsequently obtaining an upper bound on the settling time, a designer may prescribe an arbitrary constant \(T_s>0\) and construct a positive-definite function whose derivative satisfies the normalized predefined-time decay condition. Motivated by this observation, we first select a continuously differentiable positive-definite function \(V_d:\Omega\rightarrow\mathbb R_{\geq0}\) with \(V_d(\bm 0)=0\), and prescribe its desired closed-loop evolution as
\begin{equation}
\dot V_d(\bx)
=
-
\frac{\gamma_{p,q,r}}{T_s}
\left[
    \alpha V_d^p(\bx)
    +
    \beta V_d^q(\bx)
\right]^r,
\label{eq:desired_predefined_time_decay}
\end{equation}
where \(0<p<q\), \(pr<1<qr\), and \(\alpha,\beta,r,T_s>0\). The normalization coefficient is defined as
\begin{equation}
\gamma_{p,q,r}
:=
\frac{
    \Gamma_E\!\left(
        \dfrac{1-pr}{q-p}
    \right)
    \Gamma_E\!\left(
        \dfrac{qr-1}{q-p}
    \right)
}{
    \alpha^r
    \Gamma_E(r)
    (q-p)
}
\left(
    \frac{\alpha}{\beta}
\right)^{
    \frac{1-pr}{q-p}
},
\label{eq:predefined_gamma}
\end{equation}
where \(\Gamma_E(z)=\int_0^\infty t^{z-1}e^{-t}dt\) denotes the Euler gamma function.

The function \(V_d\) is freely selected by the designer to encode the desired closed-loop behavior. According to Theorem~\ref{thm:normalized_predefined_time}, any feedback law satisfying \eqref{eq:desired_predefined_time_decay} renders the origin predefined-time stable and guarantees \(T(\bx_0)\leq T_s\), \(\forall\,\bx_0\in\mathbb R^n\). Hence, \(T_s\) is specified directly by the designer before the controller is constructed. The remaining development therefore seeks an optimal feedback law whose closed-loop derivative of \(V_d\) exactly reproduces \eqref{eq:desired_predefined_time_decay}.

\begin{remark}
It is important to distinguish \(T_s\) from the prescribed convergence
deadline of the overall learning process. In particular, \(T_s\) denotes
only the state-decay horizon associated with the desired value function,
whereas the actual prescribed deadline is \(T_p\) (See Eq.~\eqref{eq:deadline_allocation}), which accounts for the
different stages of the learning and control procedure. Since the optimal
value function is not yet available during the learning transient,
\(T_s\) cannot be selected independently. Accordingly,
Corollary~\ref{cor:Ts_from_Tp} provides an explicit condition for choosing
\(T_s\) consistently with a given overall deadline \(T_p\).
\end{remark}

Consider the performance index ~\cite{Kokolakis2025Predefined}
\begin{equation}
\begin{aligned}
J(\bx_0,\bu)
={}&
\int_0^\infty
\left[
    \bm L_1(\bx)\bu
    +
    L_2(\bx)
    +
    \bu^\top\bR\bu
\right]dt,
\qquad
\bR=\bR^\top>0.
\end{aligned}
\label{eq:cost_function}
\end{equation}
Here, \(\bm L_1:\Omega\to\mathbb R^{1\times m}\)
and \(L_2:\Omega\to\mathbb R\) are to be
constructed such that \(V_d\) is the optimal
value function and the corresponding policy
satisfies the prescribed Predefined-time decay
\eqref{eq:desired_predefined_time_decay}.
The term \(\bu^\top\bR\bu\) penalizes control
effort, while \(\bm L_1(\bx)\bu\) introduces
state-control coupling and \(L_2(\bx)\)
penalizes state deviations when nonnegative
\cite{Kokolakis2025Predefined}.
The running cost need not be positive definite
for every state-input pair.
As discussed in
\cite[Sec.~8.4, pp.~524--525]{Haddad2008NonlinearControl},
the state-control cross-weighting term may
render the running cost negative for certain
control inputs without necessarily
invalidating inverse optimality over the
class of admissible stabilizing policies.
In particular, inverse optimality can be
established through the Hamilton--Jacobi--Bellman
equation and the corresponding admissibility
and transversality conditions, even when
the running cost is not pointwise nonnegative.
The constructed coefficients satisfy
\(\bm L_1(\bm0)=\bm0\) and \(L_2(\bm0)=0\),
as established subsequently (See Lemma~\ref{lem:continuous_extension}).

For the identified model \(\dot{\bx}=\hat{\bfv}(\bx)+\bg(\bx)\bu\), the Hamiltonian associated with \eqref{eq:cost_function} is

\begin{equation}
\begin{aligned}
H(\bx,\bu,\nabla V_d)
={}&
\bm L_1(\bx)\bu
+
L_2(\bx)
+
\bu^\top\bR\bu
+
\nabla V_d^\top(\bx)
\hat{\bfv}(\bx)
\\
&+
\nabla V_d^\top(\bx)
\bg(\bx)\bu .
\end{aligned}
\label{eq:hamiltonian_predefined_time}
\end{equation}

Equivalently, \(H=L_2+\nabla V_d^\top\hat{\bfv}+\bu^\top\bR\bu+\left(\bm L_1+\nabla V_d^\top\bg\right)\bu\). The corresponding HJB equation is \(\min_{\bu}H(\bx,\bu,\nabla V_d)=0\). Since \(\bR>0\), one has \(\partial^2H/\partial\bu^2=2\bR>0\), and therefore the unique minimizing policy satisfies \(\partial H/\partial\bu=\bm L_1^\top(\bx)+2\bR\bu+\bg^\top(\bx)\nabla V_d(\bx)=\bm0\). Hence, \(2\bR\bu^*=-\left[\bm L_1^\top(\bx)+\bg^\top(\bx)\nabla V_d(\bx)\right]\), which gives
\begin{equation}
\bu^*(\bx)
=
-\frac{1}{2}
\bR^{-1}
\left[
    \bm L_1^\top(\bx)
    +
    \bg^\top(\bx)\nabla V_d(\bx)
\right].
\label{eq:optimal_control_general}
\end{equation}

It should be noted that the function $V_d(\bx)$
in \eqref{eq:optimal_control_general}
is the prescribed value function of the true
system. Our objective is to enforce the value
function associated with the identified system
to exhibit the same prescribed predefined-time
behavior as that of the true system, thereby
allowing us to derive $\bm L_1(\bx)$ and
$L_2(\bx)$, as will be presented in
Theorem~\ref{thm:inverse_optimal_pt}.
This approach remains valid; however, the
trade-off is the introduction of an additional
residual term in the subsequent convergence
analysis.


\subsection{Reverse Construction of the Predefined-Time Optimal Cost}
\label{subsec:reverse_predefined_time_cost}



The prescribed predefined-time decay is embedded into an inverse-optimal
minimum-control-effort design. This construction is consistent with the
inverse-optimal control framework developed by Haddad and Chellaboina
~\cite[Sec.~8.4, pp.~522--524]{Haddad2008NonlinearControl}, in which
a stabilizing feedback law is prescribed and a corresponding performance
functional is constructed to establish its optimality.
Additional feasibility and small-control
conditions are introduced ~\cite{Sontag1989Universal} to ensure that the resulting feedback is well
defined and admits a continuous extension at the origin. Such feasibility
and continuous-extension issues are not explicitly addressed in the
predefined-time inverse-optimal formulation of
Kokolakis \emph{et al.}~\cite{Kokolakis2025Predefined}, which motivates
their explicit treatment here.


\begin{theorem}[Inverse-optimal predefined-time stabilization]
\label{thm:inverse_optimal_pt}
Suppose that Assumptions \ref{ass:plant}--\ref{ass:inverse_optimal_feasibility} hold. For the minimum-control-effort
construction, choose
\(\bm L_1(\bx)
=
\left[
\frac{
2\left(
\nabla V_d^\top(\bx)\hat{\bfv}(\bx)
+
\frac{\gamma_{p,q,r}}{T_s}
\left[
\alpha V_d^p(\bx)
+
\beta V_d^q(\bx)
\right]^r
\right)
}{
\nabla V_d^\top(\bx)
\bg(\bx)
\bR^{-1}
\bg^\top(\bx)
\nabla V_d(\bx)
}
-1
\right]
\nabla V_d^\top(\bx)\bg(\bx)\)
and
\(L_2(\bx)
=
-\nabla V_d^\top(\bx)\hat{\bfv}(\bx)
+
\frac{
\left(
\nabla V_d^\top(\bx)\hat{\bfv}(\bx)
+
\frac{\gamma_{p,q,r}}{T_s}
\left[
\alpha V_d^p(\bx)
+
\beta V_d^q(\bx)
\right]^r
\right)^2
}{
\nabla V_d^\top(\bx)
\bg(\bx)
\bR^{-1}
\bg^\top(\bx)
\nabla V_d(\bx)
}\).


Then, for either construction, the running cost
\(L(\bx,\bu)
=
\bm L_1(\bx)\bu
+
L_2(\bx)
+
\bu^\top\bR\bu\)
admits the corresponding inverse-optimal policy
\(\bu_c(\bx)\).
Moreover, the closed-loop identified system
\(\dot{\bx}
=
\hat{\bfv}(\bx)
+
\bg(\bx)\bu_c(\bx)\)
satisfies
\(\dot V_d(\bx)
\leq
-\frac{\gamma_{p,q,r}}{T_s}
\left[
\alpha V_d^p(\bx)
+
\beta V_d^q(\bx)
\right]^r\),
and is therefore predefined-time stable with
\(T(\bx_0)\leq T_s\).
\end{theorem}

\begin{proof}





Along the identified dynamics, the derivative of \(V_d\) is \(\dot V_d=\nabla V_d^\top(\bx)\hat{\bfv}(\bx)+\nabla V_d^\top(\bx)\bg(\bx)\bu\). Imposing the prescribed predefined-time evolution \eqref{eq:desired_predefined_time_decay} yields
\[
\nabla V_d^\top(\bx)\hat{\bfv}(\bx)
+
\nabla V_d^\top(\bx)\bg(\bx)\bu
=
-
\frac{\gamma_{p,q,r}}{T_s}
\left[
    \alpha V_d^p(\bx)
    +
    \beta V_d^q(\bx)
\right]^r.
\]
Equivalently, \(\nabla V_d^\top(\bx)\bg(\bx)\bu=-\sigma_d(\bx)\), where
\begin{equation}
\begin{aligned}
\sigma_d(\bx)
:={}&
\nabla V_d^\top(\bx)\hat{\bfv}(\bx)
+
\frac{\gamma_{p,q,r}}{T_s}
\left[
    \alpha V_d^p(\bx)
    +
    \beta V_d^q(\bx)
\right]^r.
\end{aligned}
\label{eq:sigma_d_predefined}
\end{equation}
Define also
\begin{equation}
D_d(\bx)
:=
\nabla V_d^\top(\bx)
\bg(\bx)
\bR^{-1}
\bg^\top(\bx)
\nabla V_d(\bx).
\label{eq:Dd_predefined}
\end{equation}

Since infinitely many control inputs may satisfy \(\nabla V_d^\top(\bx)\bg(\bx)\bu+\sigma_d(\bx)=0\), and  inspired by the minimum-effort stabilization
problem of Petersen and Barmish
~\cite{Petersen1987MinimumEffort}
and the pointwise min-norm CLF approach
presented in Section III of Primbs and
Giannelli~\cite{Primbs1999CLFMinNorm}, we select the one minimizing the weighted control effort \(\bu^\top\bR\bu\). Accordingly, consider \(\min_{\bu}\bu^\top\bR\bu\) subject to \(\nabla V_d^\top(\bx)\bg(\bx)\bu+\sigma_d(\bx)=0\). The corresponding Lagrangian is \(\mathcal L_c(\bu,\lambda_c)=\bu^\top\bR\bu+\lambda_c\left[\nabla V_d^\top(\bx)\bg(\bx)\bu+\sigma_d(\bx)\right]\). The stationarity condition with respect to \(\bu\) gives \(\partial\mathcal L_c/\partial\bu=2\bR\bu+\lambda_c\bg^\top(\bx)\nabla V_d(\bx)=\bm0\), and therefore \(\bu=-(\lambda_c/2)\bR^{-1}\bg^\top(\bx)\nabla V_d(\bx)\). Substituting this expression into the constraint yields \(-(\lambda_c/2)D_d(\bx)+\sigma_d(\bx)=0\), so that \(\lambda_c=2\sigma_d(\bx)/D_d(\bx)\). Hence, the unique \(\bR\)-weighted minimum-effort feedback satisfying the prescribed predefined-time decay is

\begin{equation}
\begin{aligned}
\bu_c(\bx)
={}&
-
\frac{
\nabla V_d^\top(\bx)\hat{\bfv}(\bx)
+
\dfrac{\gamma_{p,q,r}}{T_s}
\left[
\alpha V_d^p(\bx)
+
\beta V_d^q(\bx)
\right]^r
}{
\nabla V_d^\top(\bx)
\bg(\bx)\bR^{-1}
\bg^\top(\bx)
\nabla V_d(\bx)
}
\\[-1mm]
&\qquad {}\times
\bR^{-1}
\bg^\top(\bx)
\nabla V_d(\bx).
\end{aligned}
\label{eq:continuous_predefined_time_controller}
\end{equation}

The following condition ensures that \eqref{eq:continuous_predefined_time_controller} is well defined.

\begin{assumption}[Inverse-optimal feasibility]
\label{ass:inverse_optimal_feasibility}
For every \(\bx\in\Omega\setminus\{\bm0\}\), \(\nabla V_d^\top(\bx)\bg(\bx)\bR^{-1}\bg^\top(\bx)\nabla V_d(\bx)>0\).
\end{assumption}

\begin{remark}
Assumption~\ref{ass:inverse_optimal_feasibility} excludes nonzero states
at which \(D_d(\bm{x})=0\), equivalently
\(\bm{g}^\top(\bm{x})\nabla V_d(\bm{x})=\bm{0}\).
For example, if
\(\bm{g}=[1\;0]^\top\) and
\(V_d(\bm{x})=\frac{1}{2}(x_1^2+x_2^2)\),
then \(D_d(\bm{x})=0\) whenever \(x_1=0\), even if \(x_2\neq0\).
\end{remark}

Indeed, using \eqref{eq:continuous_predefined_time_controller}, one obtains \(\nabla V_d^\top\bg\bu_c=-(\sigma_d/D_d)\nabla V_d^\top\bg\bR^{-1}\bg^\top\nabla V_d=-\sigma_d\), and therefore
\[
\dot V_d
=
\nabla V_d^\top\hat{\bfv}
+
\nabla V_d^\top\bg\bu_c
=
\nabla V_d^\top\hat{\bfv}
-
\sigma_d
=
-
\frac{\gamma_{p,q,r}}{T_s}
\left[
    \alpha V_d^p
    +
    \beta V_d^q
\right]^r.
\]

We now construct \(\bm L_1(\bx)\) such that the Hamiltonian-minimizing policy \(\bu^*\) coincides with \(\bu_c\). From \eqref{eq:optimal_control_general}, \(\bu^*=-(1/2)\bR^{-1}\left[\bm L_1^\top+\bg^\top\nabla V_d\right]\), whereas \(\bu_c=-(\sigma_d/D_d)\bR^{-1}\bg^\top\nabla V_d\). Imposing \(\bu^*=\bu_c\) gives \(\bm L_1^\top+\bg^\top\nabla V_d=(2\sigma_d/D_d)\bg^\top\nabla V_d\), so that
\begin{equation}
\bm L_1^\top(\bx)
=
\left[
    \frac{2\sigma_d(\bx)}{D_d(\bx)}
    -
    1
\right]
\bg^\top(\bx)\nabla V_d(\bx).
\label{eq:L1_design_compact_predefined}
\end{equation}
Substituting \eqref{eq:sigma_d_predefined} and \eqref{eq:Dd_predefined} gives the explicit expression

\begin{equation}
\begin{aligned}
\bm L_1^\top(\bx)
={}&
\left[
\frac{
    2
    \left[
        \nabla V_d^\top(\bx)\hat{\bfv}(\bx)
        +
        \dfrac{\gamma_{p,q,r}}{T_s}
        \left(
            \alpha V_d^p(\bx)
            +
            \beta V_d^q(\bx)
        \right)^r
    \right]
}{
    \nabla V_d^\top(\bx)
    \bg(\bx)
    \bR^{-1}
    \bg^\top(\bx)
    \nabla V_d(\bx)
}
-
1
\right]
\\
&\qquad\times
\bg^\top(\bx)\nabla V_d(\bx).
\end{aligned}
\label{eq:L1_design}
\end{equation}

It remains to determine \(L_2(\bx)\). Substitution of \eqref{eq:optimal_control_general} into the Hamiltonian gives \(H^*=L_2+\nabla V_d^\top\hat{\bfv}-\frac{1}{4}\left[\bm L_1^\top+\bg^\top\nabla V_d\right]^\top\bR^{-1}\left[\bm L_1^\top+\bg^\top\nabla V_d\right]\). Using \(\bm L_1^\top+\bg^\top\nabla V_d=(2\sigma_d/D_d)\bg^\top\nabla V_d\), the quadratic term becomes
\[
\frac{1}{4}
\left[
    \bm L_1^\top+\bg^\top\nabla V_d
\right]^\top
\bR^{-1}
\left[
    \bm L_1^\top+\bg^\top\nabla V_d
\right]
=
\frac{\sigma_d^2(\bx)}{D_d(\bx)}.
\]
Hence, \(H^*=L_2+\nabla V_d^\top\hat{\bfv}-\sigma_d^2/D_d\). The HJB condition \(H^*=0\) therefore gives
\begin{equation}
L_2(\bx)
=
-
\nabla V_d^\top(\bx)\hat{\bfv}(\bx)
+
\frac{\sigma_d^2(\bx)}{D_d(\bx)}.
\label{eq:L2_design_compact_predefined}
\end{equation}
Equivalently,

\begin{equation}
\begin{aligned}
L_2(\bx)
={}&
-\nabla V_d^\top(\bx)\hat{\bfv}(\bx)
\\
&+
\frac{
\left[
    \nabla V_d^\top(\bx)\hat{\bfv}(\bx)
    +
    \dfrac{\gamma_{p,q,r}}{T_s}
    \left(
        \alpha V_d^p(\bx)
        +
        \beta V_d^q(\bx)
    \right)^r
\right]^2
}{
    \nabla V_d^\top(\bx)
    \bg(\bx)
    \bR^{-1}
    \bg^\top(\bx)
    \nabla V_d(\bx)
}.
\end{aligned}
\label{eq:L2_design}
\end{equation}

Hence the Proof complete.
\end{proof}

\begin{lemma}[Sufficient small-control condition]
\label{lem:small_control_condition}
Suppose that, in a neighborhood of the origin,
\(
c_1\|\bm{x}\|^m
\leq
V_d(\bm{x})
\leq
c_2\|\bm{x}\|^m
\)
and
\(
\|\nabla V_d(\bm{x})\|
\leq
c_3\|\bm{x}\|^{m-1}
\),
for some \(m>1\).
Under Assumption~\ref{ass:basis_origin}, together with the boundedness
of the identifier weights, there exists a constant \(c_f>0\) such that
\(
\|\hat{\bm f}(\bm{x})\|
\leq
c_f\|\bm{x}\|
\)
in a neighborhood of the origin.

Assume further that there exists a constant \(c_b>0\) such that
\(
\|\bm g^\top(\bm{x})\nabla V_d(\bm{x})\|
\geq
c_b\|\bm{x}\|^{m-1}
\)
for all sufficiently small \(\bm{x}\neq\bm{0}\).
Since
\(
D_d(\bm{x})
=
\bigl[\bm g^\top(\bm{x})\nabla V_d(\bm{x})\bigr]^\top
\bm R^{-1}
\bigl[\bm g^\top(\bm{x})\nabla V_d(\bm{x})\bigr]
\),
there exists \(c_g>0\), specifically
\(
c_g
=
\lambda_{\min}(\bm R^{-1})c_b^2
\),
such that
\(
D_d(\bm{x})
\geq
c_g\|\bm{x}\|^{2m-2}
\)
for all sufficiently small \(\bm{x}\neq\bm{0}\).

If
\(
mpr>m-1
\),
then
\(
\lim_{\bm{x}\to\bm{0}}
\frac{|\sigma_d(\bm{x})|}
{\sqrt{D_d(\bm{x})}}
=0
\).
\end{lemma}


\begin{proof}
From Assumption~\ref{ass:basis_origin} and the boundedness of the
identifier weights, there exists a constant \(c_f>0\) such that
\(\|\hat{\bm f}(\bm{x})\|\leq c_f\|\bm{x}\|\)
in a neighborhood of the origin. Hence, using
\(\|\nabla V_d(\bm{x})\|\leq c_3\|\bm{x}\|^{m-1}\),
one obtains
\(
\left|
\nabla V_d^\top(\bm{x})\hat{\bm f}(\bm{x})
\right|
\leq
c_3c_f\|\bm{x}\|^m
=
\mathcal O(\|\bm{x}\|^m)
\).

Moreover, since \(p<q\) and
\(V_d(\bm{x})=\mathcal O(\|\bm{x}\|^m)\) near the origin,
one has
\(
\left[
\alpha V_d^p(\bm{x})
+
\beta V_d^q(\bm{x})
\right]^r
=
\mathcal O\!\left(
\|\bm{x}\|^{mpr}
\right)
\)
as \(\bm{x}\to\bm{0}\). Therefore,
\(
|\sigma_d(\bm{x})|
=
\mathcal O(\|\bm{x}\|^m)
+
\mathcal O\!\left(
\|\bm{x}\|^{mpr}
\right)
\).

Since
\(
\sqrt{D_d(\bm{x})}
\geq
\sqrt{c_g}\|\bm{x}\|^{m-1}
\),
it follows that
\(
\frac{|\sigma_d(\bm{x})|}
{\sqrt{D_d(\bm{x})}}
=
\mathcal O(\|\bm{x}\|)
+
\mathcal O\!\left(
\|\bm{x}\|^{mpr-(m-1)}
\right)
\).

Because \(m>1\) and \(mpr>m-1\), both terms vanish as
\(\bm{x}\to\bm{0}\). Hence,
\(
\lim_{\bm{x}\to\bm{0}}
\frac{|\sigma_d(\bm{x})|}
{\sqrt{D_d(\bm{x})}}
=0
\).
\end{proof}

\begin{lemma}[Continuous extension at the origin]
\label{lem:continuous_extension}
Under Assumption~\ref{ass:inverse_optimal_feasibility}
and the preceding regularity conditions, let the conditions of
Lemma~\ref{lem:small_control_condition} hold. Then the control law
\(\bm{u}_c\) and the inverse-optimal cost coefficients
\(\bm{L}_1\) and \(L_2\) admit continuous extensions at the origin given by
\(\bm{u}_c(\bm{0})=\bm{0}\),
\(\bm{L}_1(\bm{0})=\bm{0}\), and
\(L_2(\bm{0})=0\).
\end{lemma}

\begin{proof}
Assumption~\ref{ass:inverse_optimal_feasibility} ensures that
\(D_d(\bm{x})>0\) for
\(\bm{x}\in\Omega\setminus\{\bm{0}\}\), so the expressions
for \(\bm{u}_c(\bm{x})\), \(\bm{L}_1(\bm{x})\), and
\(L_2(\bm{x})\) are well defined away from the origin.
Their continuity there follows from the preceding regularity
conditions. It therefore remains to establish their limits at the origin.

By Lemma~\ref{lem:small_control_condition},
\(\lim_{\bm{x}\to\bm{0}}
|\sigma_d(\bm{x})|/\sqrt{D_d(\bm{x})}=0\).

First, substituting
\(\bm{u}_c(\bm{x})
=-(\sigma_d(\bm{x})/D_d(\bm{x}))
\bm{R}^{-1}\bm{g}^\top(\bm{x})
\nabla V_d(\bm{x})\)
into its weighted quadratic norm gives
\(\bm{u}_c^\top(\bm{x})\bm{R}\bm{u}_c(\bm{x})
=
\frac{\sigma_d^2(\bm{x})}{D_d^2(\bm{x})}
\nabla V_d^\top(\bm{x})\bm{g}(\bm{x})
\bm{R}^{-1}\bm{g}^\top(\bm{x})
\nabla V_d(\bm{x})
=
\frac{\sigma_d^2(\bm{x})}{D_d(\bm{x})}\).
Since \(\bm{R}\succ0\),
\(\lambda_{\min}(\bm{R})\|\bm{u}_c(\bm{x})\|^2
\leq
\bm{u}_c^\top(\bm{x})\bm{R}\bm{u}_c(\bm{x})\),
and consequently
\(\|\bm{u}_c(\bm{x})\|
\leq
\frac{|\sigma_d(\bm{x})|}
{\sqrt{\lambda_{\min}(\bm{R})D_d(\bm{x})}}\).
Hence, Lemma~\ref{lem:small_control_condition} implies
\(\lim_{\bm{x}\to\bm{0}}\bm{u}_c(\bm{x})=\bm{0}\),
which proves continuity of the extension
\(\bm{u}_c(\bm{0})=\bm{0}\).

Next, multiplying the control law by \(\bm{R}\) yields
\(\bm{R}\bm{u}_c(\bm{x})
=
-\frac{\sigma_d(\bm{x})}{D_d(\bm{x})}
\bm{g}^\top(\bm{x})\nabla V_d(\bm{x})\).
Hence, the expression for \(\bm{L}_1\) can be rewritten as
\(\bm{L}_1^\top(\bm{x})
=
\left[
\frac{2\sigma_d(\bm{x})}{D_d(\bm{x})}-1
\right]
\bm{g}^\top(\bm{x})\nabla V_d(\bm{x})
=
-2\bm{R}\bm{u}_c(\bm{x})
-\bm{g}^\top(\bm{x})\nabla V_d(\bm{x})\).
Because \(V_d\) is continuously differentiable and has a minimum
at the origin, \(\nabla V_d(\bm{0})=\bm{0}\) and
\(\nabla V_d(\bm{x})\to\bm{0}\) as \(\bm{x}\to\bm{0}\).
Continuity of \(\bm{g}\) implies that it is locally bounded. Therefore,
\(\|\bm{L}_1(\bm{x})\|
\leq
2\|\bm{R}\|\|\bm{u}_c(\bm{x})\|
+
\|\bm{g}(\bm{x})\|\|\nabla V_d(\bm{x})\|
\to0\).
Thus, \(\bm{L}_1(\bm{0})=\bm{0}\) defines a continuous extension.

Finally, using the quadratic identity above,
\(L_2(\bm{x})
=
-\nabla V_d^\top(\bm{x})\hat{\bm{f}}(\bm{x})
+
\bm{u}_c^\top(\bm{x})\bm{R}\bm{u}_c(\bm{x})\).
By Assumption~\ref{ass:basis_origin} and the boundedness of the
identifier weights, \(\hat{\bm{f}}(\bm{x})\) is locally bounded
near the origin. Hence,
\(|L_2(\bm{x})|
\leq
\|\nabla V_d(\bm{x})\|
\|\hat{\bm{f}}(\bm{x})\|
+
\|\bm{R}\|\|\bm{u}_c(\bm{x})\|^2
\to0\).
Therefore, \(L_2(\bm{0})=0\) also defines a continuous extension,
which completes the proof.
\end{proof}

\begin{corollary}[Design rule for the desired Lyapunov function]
\label{cor:vd_design}
Under the conditions of Lemma~\ref{lem:small_control_condition},
a sufficient design condition ensuring the small-control property is
\(mpr>m-1\), where \(m>1\) denotes the local growth order of
\(V_d\) around the origin. Combining this condition with the
predefined-time requirement \(pr<1<qr\) yields the admissible design
range \((m-1)/m<pr<1<qr\). In particular, if
\(V_d(\bm{x})=\frac{1}{2}\bm{x}^\top\bm{P}\bm{x}\), with
\(\bm{P}=\bm{P}^\top\succ0\), then \(m=2\), and the design condition
reduces to \(\frac{1}{2}<pr<1<qr\).
\end{corollary}

\subsection{Critic-Only Integral Reinforcement Learning}
\label{subsec:critic_only_irl}

\begin{remark}
Although the unknown drift \(\bfv(\bx)\) has been reconstructed in the
preceding identification stage, the estimate \(\hat{\bfv}(\bx)\) is not
treated as an exact model of the plant. It is used only as an auxiliary
quantity for constructing the inverse-optimal terms \(\bm L_1(\bx)\) and
\(L_2(\bx)\) in the preceding section. In the present IRL formulation,
\(\bfv(\bx)\) is still regarded as unknown, and the optimal value function
and control policy are learned directly from integral data without
requiring the exact drift dynamics. Another point to note is that, since the function $V_d$ and its desired behavior are defined with respect to \(\hat{\bfv}(\bx)\) rather than  \(\bfv(\bx)\), the resulting dynamics on  \(\bfv(\bx)\) will only achieve practical predefined-time convergence (see Theorem~\ref{thm:closed_loop}). The discrepancy between \(\bfv(\bx)\) and \(\hat{\bfv}(\bx)\) represented by the term \(\bar\Delta_f\) contributes to the radius of the convergence region (as shown in the proof of Theorem~\ref{thm:closed_loop}).
\end{remark}

Follow by~\citep{Kokolakis2025Predefined}, after constructing \(\bm L_1(\bx)\) and \(L_2(\bx)\), define
\(\ell(\bx,\bu):=\bm L_1(\bx)\bu+L_2(\bx)+\bu^\top\bR\bu\).
The optimal value function is
\begin{equation}
    V^*(\bx_0)
    =
    \min_{\bu}
    \int_0^\infty
    \ell(\bx(t),\bu(t))\,dt .
    \label{eq:optimal_value_rl}
\end{equation}
Note that \(V_d\) and \(V^*\) are the same.
For \(\dot{\bx}=\bfv(\bx)+\bg(\bx)\bu\), the Hamiltonian is
\(H(\bx,\bu,\nabla V^*)=\ell(\bx,\bu)+
\nabla V^{*\top}(\bx)[\bfv(\bx)+\bg(\bx)\bu]\).
The HJB equation satisfies
\(0=\min_{\bu}H(\bx,\bu,\nabla V^*)\).
Since \(\partial^2H/\partial\bu^2=2\bR>0\), the minimizing policy follows
from
\(\partial H/\partial\bu=
\bm L_1^\top(\bx)+2\bR\bu+
\bg^\top(\bx)\nabla V^*(\bx)=\bm0\).
Hence,
\begin{equation}
    \bu^*(\bx)
    =
    -\frac{1}{2}
    \bR^{-1}
    \left[
        \bm L_1^\top(\bx)
        +
        \bg^\top(\bx)\nabla V^*(\bx)
    \right].
    \label{eq:critic_only_optimal_policy}
\end{equation}


On the compact set \(\Omega\), by the
higher-order Weierstrass approximation theorem~\citep{Vamvoudakis2010}, let
\begin{equation}
    V^*(\bx)
    =
    \bW^{*\top}\bphi(\bx)
    +
    \varepsilon(\bx),
    \label{eq:critic_value_approximation}
\end{equation}

where \(\bphi(\bx)\in\mathbb R^{N_c}\),
\(\bW^*\in\mathbb R^{N_c}\), and
\(\varepsilon(\bx)\) denotes the approximation error. Then
\(\nabla V^*(\bx)=
(\nabla\bphi(\bx))^\top\bW^*
+
\nabla\varepsilon(\bx)\).
The critic estimate is
\begin{equation}
    \hat V(\bx)
    =
    \hat{\bW}^{\top}\bphi(\bx),
    \qquad
    \nabla\hat V(\bx)
    =
    (\nabla\bphi(\bx))^\top\hat{\bW}.
    \label{eq:critic_estimate}
\end{equation}
Thus, the critic-generated policy is
\begin{equation}
    \hat{\bu}(\bx,\hat{\bW})
    =
    -\frac{1}{2}
    \bR^{-1}
    \left[
        \bm L_1^\top(\bx)
        +
        \bg^\top(\bx)
        (\nabla\bphi(\bx))^\top
        \hat{\bW}
    \right].
    \label{eq:critic_only_policy}
\end{equation}

\begin{proposition}[Critic approximation residual~\cite{Vamvoudakis2010,Kokolakis2025Predefined}]
\label{prop:critic_approximation_residual}
On the compact set
\(
\Omega
\),
suppose that the critic approximation converges
uniformly in the
\(
C^1
\)
sense as
\(
N_C\to\infty
\).
Then,
\[
\lim_{N_C\to\infty}
\sup_{\bx\in\Omega}
\left(
|\varepsilon(\bx)|
+
\|\nabla\varepsilon(\bx)\|
\right)
=0.
\]
\end{proposition}

\begin{remark}[Approximation accuracy and computational trade-off]
Proposition~\ref{prop:critic_approximation_residual}
provides one possible means of reducing the
ultimate bounds established in
Theorems~\ref{thm:predefined_time_critic}
and~\ref{thm:closed_loop}
by increasing the number of critic neurons
\(
N_C
\).
However, approximation residuals are only one
of several factors determining the UUB radii.
Increasing
\(
N_C
\)
may improve approximation accuracy at the
expense of computational complexity, whereas
a smaller network may sacrifice approximation
capability~\cite{ParkSandberg1991}.
This trade-off suggests a possible extension
using an adaptive network architecture, in which
\(
N_C
\)
is adjusted online through neuron allocation
or growing-and-pruning mechanisms, as explored
in~\cite{Platt1991RAN,Han2010SORBF}.
\end{remark}

We have $\hat{\bu}(\bx,\hat{\bW})$ is  the current
learned policy and let $\bu_b(t)$ the behavior input
used to generate the measured trajectory.
During learning, the behavior input satisfies
\begin{equation}
\bu_b(t)
=
\hat{\bu}(\bx(t),\hat{\bW}(t))
+
\be_p(t),
\label{eq:learned_behavior_policy}
\end{equation}
where $\be_p(t)\in\mathbb R^m$ is a bounded
probing signal introduced to provide sufficient
excitation during learning.

The deviation of the behavior input from the
optimal policy can be decomposed as
$\bu_b-\bu^*
=
(\bu_b-\hat{\bu})
+
(\hat{\bu}-\bu^*)
=
\be_p+
(\hat{\bu}-\bu^*)$.

Along the measured dynamics
$\dot{\bx}=\bfv(\bx)+\bg(\bx)\bu_b$,
the optimal Bellman equation is
$0
=
\ell(\bx,\bu^*)
+
\nabla V^{*\top}(\bx)
\left[
\bfv(\bx)+\bg(\bx)\bu^*
\right]$.

Since the measured trajectory is generated
under $\bu_b$ rather than $\bu^*$,
define the policy-suboptimality residual
associated with the current learned policy as
$\delta_{\hat u}(\bx)
:=
\ell(\bx,\hat{\bu})
-
\ell(\bx,\bu^*)
+
\nabla V^{*\top}(\bx)\bg(\bx)
\left[
\hat{\bu}-\bu^*
\right]$.

Consequently,
$\nabla V^{*\top}
\left[
\bfv+\bg\hat{\bu}
\right]
=
-\ell(\bx,\hat{\bu})
+
\delta_{\hat u}(\bx)$.

Using
$\bfv+\bg\bu_b
=
\bfv+\bg\hat{\bu}
+
\bg(\bu_b-\hat{\bu})$,
one obtains
$\dot V^*
=
\nabla V^{*\top}
\left[
\bfv+\bg\bu_b
\right]
=
-\ell(\bx,\hat{\bu})
+
\delta_{\hat u}(\bx)
+
\nabla V^{*\top}\bg
\left(
\bu_b-\hat{\bu}
\right)$.

Thus, the probing signal enters the Bellman
relation through
$\nabla V^{*\top}\bg\be_p$,
whereas $\delta_{\hat u}$ accounts for the
suboptimality of the learned policy.
After the data-excitation phase
\cite{Guo2020CompositeLearning,Parikh2019IntegralCL},
the probing signal is removed, yielding
$\be_p=\mathbf0$.
In contrast, $\delta_{\hat u}=0$ when
$\hat{\bu}=\bu^*$.

Integrating over $[t-T,t]$ gives
\begin{equation}
\begin{aligned}
0
={}&
V^*(\bx(t))
-
V^*(\bx(t-T))
+
\int_{t-T}^{t}
\ell\!\left(
    \bx(\tau),
    \hat{\bu}(\tau)
\right)d\tau
\\
&-
\int_{t-T}^{t}
\nabla V^{*\top}(\bx(\tau))
\bg(\bx(\tau))
\left[
    \bu_b(\tau)
    -
    \hat{\bu}(\tau)
\right]d\tau
-
\int_{t-T}^{t}
\delta_{\hat u}(\bx(\tau))\,d\tau.
\end{aligned}
\label{eq:integral_bellman_equation}
\end{equation}

From Eq.~\eqref{eq:critic_value_approximation} we have
$V^*(\bx(t))-V^*(\bx(t-T))
=
\bW^{*\top}
[\bphi(\bx(t))-\bphi(\bx(t-T))]
+
\varepsilon(\bx(t))
-
\varepsilon(\bx(t-T))$.
Define
$\Delta\bphi(t):=
\bphi(\bx(t))-\bphi(\bx(t-T))$.
Since
$\nabla V^{*\top}(\bx)
=
\bW^{*\top}\nabla\bphi(\bx)
+
\nabla\varepsilon^\top(\bx)$,
substitution into
\eqref{eq:integral_bellman_equation}
gives
\[
\begin{aligned}
0
={}&
\bW^{*\top}\Delta\bphi(t)
+
\varepsilon(\bx(t))
-
\varepsilon(\bx(t-T))
+
\int_{t-T}^{t}
\ell(\bx(\tau),\hat{\bu}(\tau))d\tau
\\
&-
\bW^{*\top}
\int_{t-T}^{t}
\nabla\bphi(\bx(\tau))
\bg(\bx(\tau))
[
\bu_b(\tau)-\hat{\bu}(\tau)
]d\tau
\\
&-
\int_{t-T}^{t}
\nabla\varepsilon^\top(\bx(\tau))
\bg(\bx(\tau))
[
\bu_b(\tau)-\hat{\bu}(\tau)
]d\tau
-
\int_{t-T}^{t}
\delta_{\hat u}(\bx(\tau))\,d\tau.
\end{aligned}
\]

Define
\begin{equation}
\begin{aligned}
\boldsymbol\Omega(t)
:={}&
\Delta\bphi(t)
-
\int_{t-T}^{t}
\nabla\bphi(\bx(\tau))
\bg(\bx(\tau))
\left[
\bu_b(\tau)-\hat{\bu}(\tau)
\right]d\tau.
\end{aligned}
\label{eq:integral_regressor}
\end{equation}
and
\begin{equation}
\rho(t)
:=
\int_{t-T}^{t}
\ell\!\left(
\bx(\tau),
\hat{\bu}(\tau)
\right)d\tau.
\label{eq:integral_cost_term}
\end{equation}

Define the aggregate Bellman residual as
\begin{equation}
\begin{aligned}
\varepsilon_B(t)
:={}&
\varepsilon(\bx(t))
-
\varepsilon(\bx(t-T))
\\
&-
\int_{t-T}^{t}
\nabla\varepsilon^\top(\bx(\tau))
\bg(\bx(\tau))
[
\bu_b(\tau)-\hat{\bu}(\tau)
]d\tau
\\
&-
\int_{t-T}^{t}
\delta_{\hat u}(\bx(\tau))\,d\tau.
\end{aligned}
\label{Residual}
\end{equation}

Using
\eqref{eq:integral_regressor}
and
\eqref{eq:integral_cost_term},
the integral Bellman equation becomes
\begin{equation}
    \bW^{*\top}\boldsymbol\Omega(t)
    +
    \rho(t)
    +
    \varepsilon_B(t)
    =
    0.
    \label{eq:integral_bellman_with_approximation}
\end{equation}


Equation~\eqref{eq:integral_bellman_with_approximation} is a significant equation in stability analysis; however, it is not the equation used for learning.
For the critic estimate \(\hat{\bW}\), define the integral Bellman
residual
\begin{equation}
    \xi(t)
    :=
    \hat{\bW}^{\top}(t)
    \boldsymbol\Omega(t)
    +
    \rho(t).
    \label{eq:critic_integral_residual}
\end{equation}

Motivated by the analysis in Proposition 1
of~\cite{Gong2025}, we have:

\begin{proposition}[Bounded Bellman approximation residual]
\label{pr2}
Under Proposition~\ref{prop:critic_approximation_residual},
suppose that the probing signal satisfies
$\|\be_p(t)\|\le e_M$.
If the learned policy is uniformly close
to the optimal policy over the considered
learning interval, then, for sufficiently
large $N_C$, there exists
$\mathscr K_B>0$ such that
$|\bar{\varepsilon}_B(t)|
\le\mathscr K_B$ and
$|\bar{\varepsilon}_{B,j}|
\le\mathscr K_B$,
$j=1,\ldots,M$.
\end{proposition}

\begin{proof}
By Proposition~\ref{prop:critic_approximation_residual},
for any $\eta>0$ and sufficiently large
$N_C$,
$|\varepsilon(\bx)|\le\eta$ and
$\|\nabla\varepsilon(\bx)\|
\le\bar\varepsilon_g\le\eta$.

From the optimal policy expression,
\begin{equation}
\begin{aligned}
\bu(\bx,\bW^*)-\bu^*(\bx)
=
\frac12\bR^{-1}\bg^\top(\bx)
\nabla\varepsilon(\bx),
\end{aligned}
\label{eq:ideal_policy_approximation_error}
\end{equation}
giving
$\|\bu(\bx,\bW^*)-\bu^*(\bx)\|
\le\frac12\|\bR^{-1}\|
G\bar\varepsilon_g$.
Hence, Proposition~\ref{prop:critic_approximation_residual}
ensures that the ideal-weight policy
approximation error can be made arbitrarily
small by increasing $N_C$, since the
gradient approximation error bound
$\bar\varepsilon_g$ can be reduced
accordingly as $N_C$ increases.
For the learned policy, under the stated
uniform policy-accuracy condition, let
$\|\hat{\bu}-\bu^*\|\le\hat\epsilon$.
Using the HJB stationary condition
$\partial H/\partial\bu
|_{\bu=\bu^*}
=\mathbf0$,
one has
$2\bR\bu^*
+\bm L_1^\top(\bx)
+\bg^\top(\bx)\nabla V^*(\bx)
=\mathbf0$.

Let
$\Delta\bu=\hat{\bu}-\bu^*$.
Then,
$\hat{\bu}=\bu^*+\Delta\bu$,
yielding
\begin{equation*}
\begin{aligned}
&\ell(\bx,\hat{\bu})
-\ell(\bx,\bu^*)
=
\bm L_1(\bx)\Delta\bu
+2\bu^{*\top}\bR\Delta\bu
+\Delta\bu^\top\bR\Delta\bu.
\end{aligned}
\end{equation*}

Consequently,
\begin{equation*}
\begin{aligned}
\delta_{\hat u}
={}&
\ell(\bx,\hat{\bu})
-\ell(\bx,\bu^*)
+\nabla V^{*\top}(\bx)\bg(\bx)
\Delta\bu
\\
={}&
\big(
\bm L_1^\top(\bx)
+2\bR\bu^*
+\bg^\top(\bx)\nabla V^*(\bx)
\big)^\top
\Delta\bu
\\
&+
\Delta\bu^\top\bR\Delta\bu.
\end{aligned}
\end{equation*}

By the HJB stationary condition and
$\bR\succ0$, it follows that
\begin{equation}
\begin{aligned}
0\le\delta_{\hat u}
&=
\Delta\bu^\top\bR\Delta\bu
\le
\lambda_{\max}(\bR)
\|\Delta\bu\|^2
\le
\lambda_{\max}(\bR)
\hat\epsilon^2.
\end{aligned}
\label{eq:policy_suboptimality_residual_bound}
\end{equation}

We have
$\bu_b-\hat{\bu}=\be_p$.
By the triangle and Cauchy--Schwarz
inequalities,
\begin{equation}
\begin{aligned}
|\varepsilon_B(t)|
&\le
2\eta+
\int_{t-T}^{t}
\|\nabla\varepsilon\|
\|\bg\|
\|\be_p\|\,d\tau
\\
&\quad+
\int_{t-T}^{t}
|\delta_{\hat u}|\,d\tau
\le
2\eta+
TG e_M\bar\varepsilon_g
+
T\lambda_{\max}(\bR)
\hat\epsilon^2
=:\mathscr K_B<\infty.
\end{aligned}
\label{eq:aggregate_residual_bound}
\end{equation}

The same bound applies to the historical
samples collected over the compact operating
set. Since the normalization denominators
are not smaller than one, the normalized
approximation residuals are also uniformly
bounded. Consequently, there exists a
finite constant $\mathscr K_B>0$
such that
$|\bar{\varepsilon}_B(t)|
\le\mathscr K_B$ and
$|\bar{\varepsilon}_{B,j}|
\le\mathscr K_B$,
for $j=1,\ldots,M$.
Moreover, increasing $N_C$ reduces
the $C^1$ approximation residual,
whereas the online policy error
additionally depends on critic learning.
\end{proof}

Let
\(\tilde{\bW}:=
\hat{\bW}-\bW^*\).
Under exact approximation,
\(\bW^{*\top}\boldsymbol\Omega+\rho=0\),
and hence
\(\xi=
\hat{\bW}^{\top}\boldsymbol\Omega
-
\bW^{*\top}\boldsymbol\Omega
=
\tilde{\bW}^{\top}\boldsymbol\Omega\).
More generally, when
\(\varepsilon\neq0\),
\eqref{eq:integral_bellman_with_approximation} gives
\(\rho=
-\bW^{*\top}\boldsymbol\Omega
-
\varepsilon_B\),
and therefore
\(\xi=
\hat{\bW}^{\top}\boldsymbol\Omega
-
\bW^{*\top}\boldsymbol\Omega
-
\varepsilon_B
=
\tilde{\bW}^{\top}\boldsymbol\Omega
-
\varepsilon_B\).

To normalize the regressor, define
\begin{equation}
    m(t)
    :=
    1+
    \boldsymbol\Omega^\top(t)
    \boldsymbol\Omega(t),
    \qquad
    \bar{\boldsymbol\Omega}(t)
    :=
    \frac{\boldsymbol\Omega(t)}
    {m(t)} .
    \label{eq:normalized_integral_regressor}
\end{equation}
The normalized Bellman residual is
\begin{equation}
    e(t)
    :=
    \frac{\xi(t)}{m(t)}
    =
    \frac{
        \hat{\bW}^{\top}(t)\boldsymbol\Omega(t)
        +
        \rho(t)
    }{
        m(t)
    } .
    \label{eq:normalized_bellman_error}
\end{equation}
Under exact approximation,
\(e(t)=
\tilde{\bW}^{\top}
\bar{\boldsymbol\Omega}(t)\).
For the general approximation case, define
\(\bar\varepsilon_B(t):=
\varepsilon_B(t)/m(t)\).
Then
\(e(t)=
\tilde{\bW}^{\top}
\bar{\boldsymbol\Omega}(t)
-
\bar\varepsilon_B(t)\),
or equivalently,
\(\tilde{\bW}^{\top}
\bar{\boldsymbol\Omega}(t)
=
e(t)
+
\bar\varepsilon_B(t)\).

Let \(M\) previously recorded integral samples be denoted by
\(\{\boldsymbol\Omega_j,\rho_j\}_{j=1}^{M}\).
For each stored sample, define
\(m_j:=
1+\boldsymbol\Omega_j^\top\boldsymbol\Omega_j\),
\(\bar{\boldsymbol\Omega}_j:=
\boldsymbol\Omega_j/m_j\),
and
\begin{equation}
    e_j
    :=
    \frac{
        \hat{\bW}^{\top}\boldsymbol\Omega_j
        +
        \rho_j
    }{
        m_j
    } .
    \label{eq:recorded_bellman_error}
\end{equation}
For each recorded sample, define
\(\varepsilon_{B,j}\)
according to
\eqref{Residual}
evaluated on the corresponding stored window and
\(\bar\varepsilon_{B,j}:=
\varepsilon_{B,j}/m_j\).
Then
\(e_j=
\tilde{\bW}^{\top}
\bar{\boldsymbol\Omega}_j
-
\bar\varepsilon_{B,j}\),
or equivalently,
\(\tilde{\bW}^{\top}
\bar{\boldsymbol\Omega}_j
=
e_j
+
\bar\varepsilon_{B,j}\).

The recorded regressors are assumed to satisfy
\begin{equation}
    \sum_{j=1}^{M}
    \bar{\boldsymbol\Omega}_j
    \bar{\boldsymbol\Omega}_j^\top
    \succeq
    \alpha_{\mathrm{IE}}\bm I,
    \qquad
    \alpha_{\mathrm{IE}}>0 .
    \label{eq:critic_ie_condition}
\end{equation}
Equivalently,
\(\lambda_{\min}
(
\sum_{j=1}^{M}
\bar{\boldsymbol\Omega}_j
\bar{\boldsymbol\Omega}_j^\top
)
\geq
\alpha_{\mathrm{IE}}>0\).

Let
\(0<\gamma_1<1<\gamma_2\).
The exponents
\(
\gamma_1
\)
and
\(
\gamma_2
\)
govern the convergence behavior of the critic update; however, they are
not selected independently. To ensure consistency with the subsequent
closed-loop predefined-time analysis (Section~\ref{sec:closed_loop}), they are chosen as
\(
\gamma_1=pr
\)
and
\(
\gamma_2=2qr-1
\),
so that
\(
0<\gamma_1<1<\gamma_2
\).

Define
\(\operatorname{sig}^{\gamma}(s)
:=
|s|^\gamma\operatorname{sgn}(s)\).
Consider the scalar two-power penalty
\(\Psi(e)
:=
\frac{k_1}{\gamma_1+1}|e|^{\gamma_1+1}
+
\frac{k_2}{\gamma_2+1}|e|^{\gamma_2+1}\),
where \(k_1,k_2>0\). Its derivative is
\(\frac{\partial\Psi(e)}{\partial e}
=
k_1\operatorname{sig}^{\gamma_1}(e)
+
k_2\operatorname{sig}^{\gamma_2}(e)\).

Since
\(e(t)=
[
\hat{\bW}^{\top}\boldsymbol\Omega(t)+\rho(t)
]/m(t)\)
and by chain rule,
one has
\(\partial e(t)/\partial\hat{\bW}
=
\boldsymbol\Omega(t)/m(t)
=
\bar{\boldsymbol\Omega}(t)\).
Therefore,
\(\nabla_{\hat{\bW}}\Psi(e(t))
=
\bar{\boldsymbol\Omega}(t)
[
k_1\operatorname{sig}^{\gamma_1}(e(t))
+
k_2\operatorname{sig}^{\gamma_2}(e(t))
]\).

Similarly,
\(\partial e_j/\partial\hat{\bW}
=
\bar{\boldsymbol\Omega}_j\),
and hence
\(\nabla_{\hat{\bW}}\Psi(e_j)
=
\bar{\boldsymbol\Omega}_j
[
k_1\operatorname{sig}^{\gamma_1}(e_j)
+
k_2\operatorname{sig}^{\gamma_2}(e_j)
]\).

Define the composite critic objective
\begin{equation}
    \mathcal J_W
    :=
    \Psi(e(t))
    +
    \sum_{j=1}^{M}
    \Psi(e_j).
    \label{eq:critic_composite_objective}
\end{equation}
Then
\[
\begin{aligned}
\nabla_{\hat{\bW}}\mathcal J_W
={}&
\bar{\boldsymbol\Omega}(t)
\left[
k_1\operatorname{sig}^{\gamma_1}(e(t))
+
k_2\operatorname{sig}^{\gamma_2}(e(t))
\right]
\\
&+
\sum_{j=1}^{M}
\bar{\boldsymbol\Omega}_j
\left[
k_1\operatorname{sig}^{\gamma_1}(e_j)
+
k_2\operatorname{sig}^{\gamma_2}(e_j)
\right].
\end{aligned}
\]

To introduce an explicit scaling parameter that will subsequently be
selected to enforce the desired predefined-time critic-learning horizon,
consider the scaled gradient flow
\(
\dot{\hat{\bW}}
=
-\alpha_W
\bm{\Gamma}
\nabla_{\hat{\bW}}\mathcal J_W
\),
where
\(\alpha_W>0\) and
\(\bm{\Gamma}
=
\bm{\Gamma}^{\top}
\succ0\).
This gives the predefined-time critic update law
\begin{equation}
\begin{aligned}
    \dot{\hat{\bW}}
    ={}&
    -\alpha_W
    \bm{\Gamma}
    \bar{\boldsymbol\Omega}(t)
    \left[
        k_1
        \operatorname{sig}^{\gamma_1}(e(t))
        +
        k_2
        \operatorname{sig}^{\gamma_2}(e(t))
    \right]
    \\
    &-
    \alpha_W
    \bm{\Gamma}
    \sum_{j=1}^{M}
    \bar{\boldsymbol\Omega}_j
    \left[
        k_1
        \operatorname{sig}^{\gamma_1}(e_j)
        +
        k_2
        \operatorname{sig}^{\gamma_2}(e_j)
    \right].
\end{aligned}
\label{eq:predefined_time_critic_update}
\end{equation}

\begin{remark}[Distinction from existing learning schemes]
~\cite{Wang2024FixedTime} investigated
fixed-time reinforcement learning but did not employ
experience replay. In contrast,
~\cite{Gong2025} and
~\cite{Kokolakis2025} incorporated experience-replay
mechanisms. However, ~\cite{Gong2025} did not explicitly
construct a running cost designed to induce
fixed-time stability, and their learning law was
developed under the assumption of known system
dynamics. Similarly, ~\cite{Kokolakis2025} did not
explicitly clarify the construction of the running
cost and its connection to the derivation of the
critic update law, while also considering known
system dynamics. ~\cite{Wang2026SafeFixedTime}
further investigated safe fixed-time reinforcement
learning with experience replay in a zero-sum
game setting. The fundamental distinction of the
present work lies in the integral Bellman residual
$e(t)$, defined in~\eqref{eq:normalized_integral_regressor},
which is constructed from off-policy data over
finite trajectory windows. Consequently, the
proposed critic update
\eqref{eq:predefined_time_critic_update}
is fundamentally different in its underlying
regression formulation. Nevertheless, the
two-power non-Lipschitz structure follows
a similar principle of
~\cite{Gong2025,Kokolakis2025Predefined,Wang2026SafeFixedTime}
for achieving fixed-time or predefined-time
convergence. Furthermore, the present work
explicitly constructs an inverse-optimal running
cost consistent with the prescribed Lyapunov decay
while identifying the unknown drift through ICL.
These distinctions are further clarified by the
stability proofs of
Theorems~\ref{thm:predefined_time_critic}
and~\ref{thm:closed_loop}.
\end{remark}

\begin{lemma}[Replay-stack coercivity]
\label{lem:replay_stack_coercivity}
Suppose that the recorded normalized regressors satisfy the finite-data
informativity condition \eqref{eq:critic_ie_condition}, namely,
\(
\sum_{j=1}^{M}
\bar{\boldsymbol\Omega}_j
\bar{\boldsymbol\Omega}_j^\top
\succeq
\alpha_{\mathrm{IE}}\bm I
\),
with
\(
\alpha_{\mathrm{IE}}>0
\).
Then, for any
\(
\tilde{\bW}\in\mathbb R^{N_c}
\),
the following inequalities hold:
\begin{equation}
\sum_{j=1}^{M}
\left|
\tilde{\bW}^{\top}
\bar{\boldsymbol\Omega}_j
\right|^{\gamma_1+1}
\geq
\alpha_{\mathrm{IE}}^{\frac{\gamma_1+1}{2}}
\|\tilde{\bW}\|^{\gamma_1+1},
\label{eq:replay_lower_bound_gamma1}
\end{equation}
and
\begin{equation}
\sum_{j=1}^{M}
\left|
\tilde{\bW}^{\top}
\bar{\boldsymbol\Omega}_j
\right|^{\gamma_2+1}
\geq
M^{\frac{1-\gamma_2}{2}}
\alpha_{\mathrm{IE}}^{\frac{\gamma_2+1}{2}}
\|\tilde{\bW}\|^{\gamma_2+1}.
\label{eq:replay_lower_bound_gamma2}
\end{equation}
In particular, both lower-bound coefficients are strictly positive.
\end{lemma}

\begin{proof}
Define the replay-regressor matrix
\(
\bar{\bm\Omega}
:=
[
\bar{\boldsymbol\Omega}_1,\ldots,
\bar{\boldsymbol\Omega}_M
]
\).
Then
\(
\bar{\bm\Omega}^{\top}\tilde{\bW}
=
\operatorname{col}
\{
\tilde{\bW}^{\top}\bar{\boldsymbol\Omega}_1,
\ldots,
\tilde{\bW}^{\top}\bar{\boldsymbol\Omega}_M
\}
\).
From \eqref{eq:critic_ie_condition},
\(
\bar{\bm\Omega}\bar{\bm\Omega}^{\top}
\succeq
\alpha_{\mathrm{IE}}\bm I
\),
and therefore
\(
\sigma_{\min}(\bar{\bm\Omega})
\geq
\sqrt{\alpha_{\mathrm{IE}}}
\).
Hence,
\(
\|\bar{\bm\Omega}^{\top}\tilde{\bW}\|_2
\geq
\sqrt{\alpha_{\mathrm{IE}}}\,
\|\tilde{\bW}\|
\).

Since
\(
0<\gamma_1<1
\),
one has
\(
1<\gamma_1+1<2
\).
For any finite-dimensional vector, the
\(
\ell_p
\)-norm is no smaller than the
\(
\ell_2
\)-norm when
\(
1<p<2
\).
Thus,
\[
\begin{aligned}
\sum_{j=1}^{M}
\left|
\tilde{\bW}^{\top}
\bar{\boldsymbol\Omega}_j
\right|^{\gamma_1+1}
&=
\left\|
\bar{\bm\Omega}^{\top}\tilde{\bW}
\right\|_{\gamma_1+1}^{\gamma_1+1}
\geq
\left\|
\bar{\bm\Omega}^{\top}\tilde{\bW}
\right\|_2^{\gamma_1+1}
\geq
\alpha_{\mathrm{IE}}^{\frac{\gamma_1+1}{2}}
\|\tilde{\bW}\|^{\gamma_1+1},
\end{aligned}
\]
which proves \eqref{eq:replay_lower_bound_gamma1}.

Since
\(
\gamma_2>1
\),
one has
\(
\gamma_2+1>2
\).
Using the finite-dimensional norm relation
\(
\|z\|_p
\geq
M^{\frac{1}{p}-\frac{1}{2}}
\|z\|_2
\)
for
\(
p>2
\),
with
\(
p=\gamma_2+1
\),
gives
\[
\left\|
\bar{\bm\Omega}^{\top}\tilde{\bW}
\right\|_{\gamma_2+1}
\geq
M^{
\frac{1}{\gamma_2+1}
-
\frac{1}{2}
}
\left\|
\bar{\bm\Omega}^{\top}\tilde{\bW}
\right\|_2.
\]
Raising both sides to the power
\(
\gamma_2+1
\)
and using the finite-data lower bound yields
\[
\begin{aligned}
\sum_{j=1}^{M}
\left|
\tilde{\bW}^{\top}
\bar{\boldsymbol\Omega}_j
\right|^{\gamma_2+1}
&=
\left\|
\bar{\bm\Omega}^{\top}\tilde{\bW}
\right\|_{\gamma_2+1}^{\gamma_2+1}
\geq
M^{\frac{1-\gamma_2}{2}}
\left\|
\bar{\bm\Omega}^{\top}\tilde{\bW}
\right\|_2^{\gamma_2+1}
\\
&\geq
M^{\frac{1-\gamma_2}{2}}
\alpha_{\mathrm{IE}}^{\frac{\gamma_2+1}{2}}
\|\tilde{\bW}\|^{\gamma_2+1},
\end{aligned}
\]
which proves \eqref{eq:replay_lower_bound_gamma2}. Since
\(
M\geq1
\)
and
\(
\alpha_{\mathrm{IE}}>0
\),
both coefficients are strictly positive.
\end{proof}

\begin{lemma}[Perturbed signed-power inequality~\cite{Gong2025}]
\label{lem:perturbed_power}
For every \(\gamma>0\) and all \(z,\epsilon\in\mathbb R\),
\begin{equation}
z\spow{-z+\epsilon}{\gamma}
\le
-2^{-(\gamma+1)}
|z|^{\gamma+1}
+
\left(
1+2\,3^\gamma
\right)
|\epsilon|^{\gamma+1}.
\label{eq:perturbed_power}
\end{equation}
If \(\epsilon=0\), the sharper identity
\begin{equation}
z\spow{-z}{\gamma}
=
-|z|^{\gamma+1}
\label{eq:exact_signed_power_identity}
\end{equation}
holds.
\end{lemma}

\begin{theorem}[Predefined-time integral critic]
\label{thm:predefined_time_critic}
Suppose that the recorded regressors
satisfy the finite-data informativity condition
\eqref{eq:critic_ie_condition} at some data-informativity instant 
\(T_E\), and the informative replay stack is retained for all
\(t\geq T_E\). Let \(T_W>0\) be an arbitrarily assigned critic-learning
horizon, independent of the initial critic-weight error. Define
\[
V_W
:=
\frac{1}{2}
\tilde{\bW}^{\top}
\bm{\Gamma}^{-1}
\tilde{\bW},
\quad
\mu_W
:=
\frac{\gamma_1+1}{2},
\quad
\nu_W
:=
\frac{\gamma_2+1}{2},
\]
so that \(0<\mu_W<1<\nu_W\), and define
\[
a_W
:=
k_1
\left(
    2\lambda_{\min}(\bm{\Gamma})
    \alpha_{\mathrm{IE}}
\right)^{\mu_W},
\]
\[
b_W
:=
k_2
M^{\frac{1-\gamma_2}{2}}
\left(
    2\lambda_{\min}(\bm{\Gamma})
    \alpha_{\mathrm{IE}}
\right)^{\nu_W}.
\]
Furthermore, define the normalized predefined-time coefficient
\[
\gamma_W
:=
\frac{
    \Gamma_E\!\left(
        \dfrac{1-\mu_W}{\nu_W-\mu_W}
    \right)
    \Gamma_E\!\left(
        \dfrac{\nu_W-1}{\nu_W-\mu_W}
    \right)
}{
    a_W
    \Gamma_E(1)
    (\nu_W-\mu_W)
}
\left(
    \frac{a_W}{b_W}
\right)^{
    \frac{1-\mu_W}{\nu_W-\mu_W}
}.
\]

Then the critic update law
\eqref{eq:predefined_time_critic_update}
has the following properties.

\emph{(i) Exact approximation case:}
If
\(
\bar\varepsilon_B(t)=0
\)
and
\(
\bar\varepsilon_{B,j}=0
\),
\(j=1,\ldots,M\),
and the critic scaling parameter satisfies
\begin{equation}
    \alpha_W
    \geq
    \frac{\gamma_W}{T_W},
    \label{eq:alphaW_exact_predefined}
\end{equation}
then, independently of
\(
\tilde{\bW}(T_E)
\),
the critic-weight error satisfies
\begin{equation}
    \tilde{\bW}(t)
    =
    \bm0,
    \qquad
    t
    \geq
    T_E+T_W.
    \label{eq:critic_exact_predefined_convergence}
\end{equation}


\emph{(ii) Nonzero residual case:}
Following Proposition~\ref{pr2}, we define
\begin{equation}
    a_{W\Delta}
    :=
    2^{-(\gamma_1+1)}a_W,
    \qquad
    b_{W\Delta}
    :=
    2^{-(\gamma_2+1)}b_W.
    \label{eq:perturbed_critic_coefficients}
\end{equation}
Let
\(
\Delta_W>0
\)
denote the aggregate critic perturbation induced by the bounded
approximation residuals, and define the corresponding normalized
predefined-time coefficient as

\begin{equation}
\gamma_{W\Delta}
:=
\frac{
\Gamma_E\!\left(
\dfrac{1-\mu_W}{\nu_W-\mu_W}
\right)
\Gamma_E\!\left(
\dfrac{\nu_W-1}{\nu_W-\mu_W}
\right)
}{
a_{W\Delta}
\Gamma_E(1)
(\nu_W-\mu_W)
}
\left(
\frac{a_{W\Delta}}{b_{W\Delta}}
\right)^{
\dfrac{1-\mu_W}{\nu_W-\mu_W}
}.
\label{eq:gammaW_delta}
\end{equation}

Select any
\(
\chi_W\in(0,1)
\)
and choose the critic scaling parameter such that
\begin{equation}
    \alpha_W
    \geq
    \frac{\gamma_{W\Delta}}
    {\chi_W T_W}.
    \label{eq:alphaW_practical_predefined}
\end{equation}
Then \(V_W(t)\) enters, no later than
\(T_E+T_W\), the forward-invariant set
\begin{equation}
    \mathcal B_W
    :=
    \left\{
        V_W\in\mathbb R_{\geq0}:
        V_W
        \leq
        \bar V_W
    \right\},
    \label{eq:critic_predefined_invariant_set}
\end{equation}
where
\(
\bar V_W
\)
is the unique positive solution of
\begin{equation}
    \alpha_W(1-\chi_W)
    \left[
        a_{W\Delta}\bar V_W^{\mu_W}
        +
        b_{W\Delta}\bar V_W^{\nu_W}
    \right]
    =
    \Delta_W.
    \label{eq:critic_predefined_residual_radius}
\end{equation}
Consequently,
\begin{equation}
    \|\tilde{\bW}(t)\|
    \leq
    \sqrt{
        2\lambda_{\max}(\bm{\Gamma})
        \bar V_W
    },
    \qquad
    t
    \geq
    T_E+T_W.
    \label{eq:critic_predefined_weight_bound}
\end{equation}

\end{theorem}

\begin{proof}
Since \(\bW^*\) is constant and
\(\tilde{\bW}=\hat{\bW}-\bW^*\), one has
\(\dot{\tilde{\bW}}=\dot{\hat{\bW}}\).
Consider
\(
V_W
=
\frac{1}{2}
\tilde{\bW}^{\top}
\bm{\Gamma}^{-1}
\tilde{\bW}
\).
Differentiating \(V_W\) and substituting
\eqref{eq:predefined_time_critic_update} yield
\begin{equation}
\begin{aligned}
\dot V_W
={}&
-\alpha_W
\tilde{\bW}^{\top}
\bar{\boldsymbol\Omega}(t)
\left[
k_1\spow{e(t)}{\gamma_1}
+
k_2\spow{e(t)}{\gamma_2}
\right]
\\
&-
\alpha_W
\sum_{j=1}^{M}
\tilde{\bW}^{\top}
\bar{\boldsymbol\Omega}_j
\left[
k_1\spow{e_j}{\gamma_1}
+
k_2\spow{e_j}{\gamma_2}
\right].
\end{aligned}
\label{eq:VW_dot_raw}
\end{equation}

For the exact-approximation case,
\(\bar\varepsilon_B(t)=0\) and
\(\bar\varepsilon_{B,j}=0\), and hence
\(e(t)=\tilde{\bW}^{\top}\bar{\boldsymbol\Omega}(t)\) and
\(e_j=\tilde{\bW}^{\top}\bar{\boldsymbol\Omega}_j\).
Using the identity
\(z\spow{z}{\gamma}=|z|^{\gamma+1}\)
in \eqref{eq:VW_dot_raw} gives
\begin{equation}
\begin{aligned}
\dot V_W
={}&
-\alpha_W k_1
\left|
\tilde{\bW}^{\top}
\bar{\boldsymbol\Omega}(t)
\right|^{\gamma_1+1}
-
\alpha_W k_2
\left|
\tilde{\bW}^{\top}
\bar{\boldsymbol\Omega}(t)
\right|^{\gamma_2+1}
\\
&-
\alpha_W k_1
\sum_{j=1}^{M}
\left|
\tilde{\bW}^{\top}
\bar{\boldsymbol\Omega}_j
\right|^{\gamma_1+1}
-
\alpha_W k_2
\sum_{j=1}^{M}
\left|
\tilde{\bW}^{\top}
\bar{\boldsymbol\Omega}_j
\right|^{\gamma_2+1}.
\end{aligned}
\label{eq:VW_exact_raw}
\end{equation}
Discarding the nonpositive current-data terms and applying
Lemma~\ref{lem:replay_stack_coercivity} yield
\begin{equation}
\begin{aligned}
\dot V_W
\leq{}&
-\alpha_W k_1
\alpha_{\mathrm{IE}}^{\frac{\gamma_1+1}{2}}
\|\tilde{\bW}\|^{\gamma_1+1}
\\
&-
\alpha_W k_2
M^{\frac{1-\gamma_2}{2}}
\alpha_{\mathrm{IE}}^{\frac{\gamma_2+1}{2}}
\|\tilde{\bW}\|^{\gamma_2+1}.
\end{aligned}
\label{eq:VW_exact_weight}
\end{equation}

From the definition of \(V_W\),
\(
V_W
\leq
\|\tilde{\bW}\|^2/
[2\lambda_{\min}(\bm{\Gamma})]
\),
and therefore
\(
\|\tilde{\bW}\|^2
\geq
2\lambda_{\min}(\bm{\Gamma})V_W
\).
Using this relation in
\eqref{eq:VW_exact_weight} gives
\begin{equation}
\dot V_W
\leq
-\alpha_W
\left[
a_WV_W^{\mu_W}
+
b_WV_W^{\nu_W}
\right],
\label{eq:VW_exact_comparison}
\end{equation}
\[
a_W
=
k_1
\left(
2\lambda_{\min}(\bm{\Gamma})
\alpha_{\mathrm{IE}}
\right)^{\mu_W},
\quad
b_W
=
k_2
M^{\frac{1-\gamma_2}{2}}
\left(
2\lambda_{\min}(\bm{\Gamma})
\alpha_{\mathrm{IE}}
\right)^{\nu_W}.
\]

if
\(
\alpha_W\geq\gamma_W/T_W
\),
then
\[
\dot V_W
\leq
-
\frac{\gamma_W}{T_W}
\left[
a_WV_W^{\mu_W}
+
b_WV_W^{\nu_W}
\right].
\]
Theorem~\ref{thm:normalized_predefined_time} therefore gives
\(T_W^{\mathrm{act}}\leq T_W\), independently of
\(V_W(T_E)\), where \(T_W^{\mathrm{act}}\) is the actual convergence time. Consequently,
\[
\tilde{\bW}(t)=\bm0,
\qquad
t\geq T_E+T_W,
\]
which proves the exact-approximation claim.

For the nonzero-residual case,
\[
e(t)
=
\tilde{\bW}^{\top}
\bar{\boldsymbol\Omega}(t)
-
\bar\varepsilon_B(t),
\qquad
e_j
=
\tilde{\bW}^{\top}
\bar{\boldsymbol\Omega}_j
-
\bar\varepsilon_{B,j}.
\]
Applying Lemma~\ref{lem:perturbed_power} with
\(\gamma=\gamma_1\) and \(\gamma=\gamma_2\) to every term in
\eqref{eq:VW_dot_raw}, using Proposition
~\ref{pr2}, and discarding the
remaining nonpositive current-data dissipation give
\begin{equation}
\begin{aligned}
\dot V_W
\leq{}&
-\alpha_W k_1
2^{-(\gamma_1+1)}
\sum_{j=1}^{M}
\left|
\tilde{\bW}^{\top}
\bar{\boldsymbol\Omega}_j
\right|^{\gamma_1+1}
\\
&-
\alpha_W k_2
2^{-(\gamma_2+1)}
\sum_{j=1}^{M}
\left|
\tilde{\bW}^{\top}
\bar{\boldsymbol\Omega}_j
\right|^{\gamma_2+1}
+
\Delta_W,
\end{aligned}
\label{eq:VW_perturbed_before_coercivity}
\end{equation}
where
\begin{equation}
\begin{aligned}
\Delta_W
:={}&
\alpha_W(M+1)
\Big[
k_1
\left(
1+2\,3^{\gamma_1}
\right)
\mathscr K_B^{\gamma_1+1}
\\
&\qquad\qquad
+
k_2
\left(
1+2\,3^{\gamma_2}
\right)
\mathscr K_B^{\gamma_2+1}
\Big].
\end{aligned}
\label{eq:critic_residual_perturbation}
\end{equation}

Lemma~\ref{lem:replay_stack_coercivity} and
\(
\|\tilde{\bW}\|^2
\geq
2\lambda_{\min}(\bm{\Gamma})V_W
\)
now imply
\begin{equation}
\dot V_W
\leq
-\alpha_W
\left[
a_{W\Delta}V_W^{\mu_W}
+
b_{W\Delta}V_W^{\nu_W}
\right]
+
\Delta_W.
\label{eq:VW_practical}
\end{equation}

By the definition of \(\bar V_W\) in
\eqref{eq:critic_predefined_residual_radius}, for every
\(V_W>\bar V_W\), the strict monotonicity of the two power terms gives
\[
\Delta_W
<
\alpha_W(1-\chi_W)
\left[
a_{W\Delta}V_W^{\mu_W}
+
b_{W\Delta}V_W^{\nu_W}
\right].
\]
Substituting this inequality into \eqref{eq:VW_practical} yields,
for every \(V_W\notin\mathcal B_W\),
\begin{equation}
\dot V_W
\leq
-\alpha_W\chi_W
\left[
a_{W\Delta}V_W^{\mu_W}
+
b_{W\Delta}V_W^{\nu_W}
\right].
\label{eq:VW_outside_BW}
\end{equation}

Using the gain condition
\eqref{eq:alphaW_practical_predefined}, namely
\[
\alpha_W
\geq
\frac{\gamma_{W\Delta}}
{\chi_WT_W},
\]
one further obtains
\begin{equation}
\dot V_W
\leq
-
\frac{\gamma_{W\Delta}}{T_W}
\left[
a_{W\Delta}V_W^{\mu_W}
+
b_{W\Delta}V_W^{\nu_W}
\right],
\qquad
V_W>\bar V_W.
\label{eq:VW_practical_comparison}
\end{equation}
Hence, by Theorem~\ref{thm:normalized_predefined_time}, the entrance
time into \(\mathcal B_W\) satisfies
\[
T_{W,\mathcal B_W}
\leq
T_W,
\]
where \(T_{W,\mathcal B_W}\) is the time of entry into the sphere \(\mathcal B_W\).
Since the informative replay stack is established at \(T_E\), it follows that
\begin{equation}
V_W(t)
\leq
\bar V_W,
\qquad
t\geq T_E+T_W.
\label{eq:VW_practical_entry}
\end{equation}

It remains to verify forward invariance. At the boundary
\(V_W=\bar V_W\), equations
\eqref{eq:VW_practical} and
\eqref{eq:critic_predefined_residual_radius} yield
\begin{equation}
\begin{aligned}
\dot V_W
&\leq
-\alpha_W
\left[
a_{W\Delta}\bar V_W^{\mu_W}
+
b_{W\Delta}\bar V_W^{\nu_W}
\right]
+
\Delta_W
\\
&=
-\alpha_W\chi_W
\left[
a_{W\Delta}\bar V_W^{\mu_W}
+
b_{W\Delta}\bar V_W^{\nu_W}
\right]
\leq 0.
\end{aligned}
\label{eq:VW_boundary_invariance}
\end{equation}
Therefore, \(\mathcal B_W\) is forward invariant.

Finally, from
\[
V_W
=
\frac{1}{2}
\tilde{\bW}^{\top}
\bm{\Gamma}^{-1}
\tilde{\bW},
\]
one has
\[
V_W
\geq
\frac{
\|\tilde{\bW}\|^2
}{
2\lambda_{\max}(\bm{\Gamma})
}.
\]
Combining this relation with \eqref{eq:VW_practical_entry} gives
\[
\|\tilde{\bW}(t)\|
\leq
\sqrt{
2\lambda_{\max}(\bm{\Gamma})
\bar V_W
},
\qquad
t\geq T_E+T_W.
\]
This proves the practical predefined-time critic-convergence claim.
\end{proof}

\section{Predefined-Time Closed-Loop Stability Analysis}
\label{sec:closed_loop}



Recall
\(
\mu
=
\mu_W
=
(pr+1)/2
=
(\gamma_1+1)/2
\)
and
\(
\nu
=
\nu_W
=
qr
=
(\gamma_2+1)/2
\),
with
\(
\frac12<\mu<1<\nu
\).

\begin{equation}
\begin{aligned}
\mathcal J(t)
&:=
\mathcal \int_{t-T}^{t}
V^*(\bx(\tau))\,d\tau+V_W(t),
\\
\mathcal I_x(t)
&:=
\int_{t-T}^{t}
V^*(\bx(\tau))\,d\tau.
\end{aligned}
\label{eq:closed_loop_J}
\end{equation}

Since the closed-loop analysis is carried out on the compact set
\(\Omega\), and \(V^*\), \(\bg\), \(\bphi\), and \(\varepsilon\)
are continuously differentiable on \(\Omega\), the corresponding
quantities are uniformly bounded. Define
\(
r_\Omega
:=
\sup_{\bm\xi\in\Omega}\|\bm\xi\|
\),
\(
\bar p_V
:=
\sup_{\bm\xi\in\Omega}
\|\nabla V^*(\bm\xi)\|
\),
\(
\bar g
:=
\sup_{\bm\xi\in\Omega}
\|\bg(\bm\xi)\|
\),
\(
\bar\phi_g
:=
\sup_{\bm\xi\in\Omega}
\|\nabla\bphi(\bm\xi)\|
\),
and
\(
\bar\varepsilon_g
:=
\sup_{\bm\xi\in\Omega}
\|\nabla\varepsilon(\bm\xi)\|
\).
Hence, for all \(\bx\in\Omega\), one has
\(
\|\nabla V^*(\bx)\|
\leq
\bar p_V
\),
\(
\|\bg(\bx)\|
\leq
\bar g
\),
\(
\|\nabla\bphi(\bx)\|
\leq
\bar\phi_g
\),
and
\(
\|\nabla\varepsilon(\bx)\|
\leq
\bar\varepsilon_g
\).

From
\(
V^*
=
\bW^{*\top}\bphi+\varepsilon
\)
and
\(
\hat V
=
\hat{\bW}^{\top}\bphi
\),
one has
\(
\nabla V^*-\nabla\hat V
=
-(\nabla\bphi)^\top\tilde{\bW}
+
\nabla\varepsilon
\),
and therefore
\(
\|\nabla V^*-\nabla\hat V\|
\leq
\bar\phi_g\|\tilde{\bW}\|
+
\bar\varepsilon_g
\).
Moreover,
\(
\hat{\bu}-\bu^*
=
-\frac12
\bR^{-1}\bg^\top
[
\nabla\hat V-\nabla V^*
]
\),
which gives
\(
\|\hat{\bu}-\bu^*\|
\leq
\frac12
\|\bR^{-1}\|
\bar g
[
\bar\phi_g\|\tilde{\bW}\|
+
\bar\varepsilon_g
]
\).
Consequently,
\(
|
\nabla V^{*\top}
\bg
(\hat{\bu}-\bu^*)
|
\leq
c_W\|\tilde{\bW}\|+c_0
\),
where
\begin{equation}
c_W
:=
\frac12
\bar p_V
\|\bR^{-1}\|
\bar g^2
\bar\phi_g,
\qquad
c_0
:=
\frac12
\bar p_V
\|\bR^{-1}\|
\bar g^2
\bar\varepsilon_g.
\label{eq:closed_loop_mismatch_constants}
\end{equation}

The ideal inverse-optimal policy satisfies
\(
\nabla V^{*\top}
[
\hat{\bfv}
+
\bg\bu^*
]
=
-
\frac{\gamma_{p,q,r}}{T_s}
\left[
\alpha(V^*)^p+\beta(V^*)^q
\right]^r
\)
along the identified dynamics.
Since
\(
(a+b)^r
\geq
c_r(a^r+b^r)
\),
with
\(
c_r:=\min\{1,2^{r-1}\}
\),
one obtains
\(
[
\alpha(V^*)^p+\beta(V^*)^q
]^r
\geq
c_r
[
\alpha^r(V^*)^{pr}
+
\beta^r(V^*)^{qr}
]
\).
Define
\begin{equation}
\kappa_1
:=
\frac{\gamma_{p,q,r}}{T_s}
c_r\alpha^r,
\qquad
\kappa_2
:=
\frac{\gamma_{p,q,r}}{T_s}
c_r\beta^r.
\label{eq:kappa_state_predefined}
\end{equation}

For the actual closed-loop dynamics
\(
\dot{\bx}
=
\bfv(\bx)
+
\bg(\bx)\hat{\bu}
\),
write
\(
\bfv(\bx)
=
\hat{\bfv}(\bx,t)
+
\tilde{\bfv}(\bx,t)
\),
where
\(
\tilde{\bfv}(\bx,t)
:=
\bfv(\bx)-\hat{\bfv}(\bx,t)
\).
Then
\(
\dot V^*
=
\nabla V^{*\top}
[
\hat{\bfv}
+
\bg\bu^*
]
+
\nabla V^{*\top}\tilde{\bfv}
+
\nabla V^{*\top}\bg
(\hat{\bu}-\bu^*)
\).

By Lemma~\ref{lem:identification_error},
\(
\|\tilde{\bfv}(\bx,t)\|
\leq
\bar{\Delta}_f
\),
and hence
\(
\nabla V^{*\top}\tilde{\bfv}
\leq
\|\nabla V^*\|
\bar{\Delta}_f
\leq
\bar p_V\bar{\Delta}_f
\).
Define
\(
c_f:=\bar p_V\bar{\Delta}_f
\).
Moreover, from the policy-error bound,
\(
\nabla V^{*\top}\bg
(\hat{\bu}-\bu^*)
\leq
c_W\|\tilde{\bW}\|
+
c_0
\).
Therefore,
\begin{equation}
\dot V^*
\leq
-\kappa_1(V^*)^{pr}
-\kappa_2(V^*)^{qr}
+
c_W\|\tilde{\bW}\|
+
c_0
+
c_f.
\label{eq:Vstar_learned}
\end{equation}

\begin{lemma}[Optimized Young inequality~\citep{Hardy1988Inequalities}]
\label{lem:optimized_young}
For
\(
s\geq0
\),
\(
c\geq0
\),
\(
b>0
\),
and
\(
p>1
\),
\begin{equation}
cs
\leq
bs^p
+
(p-1)
p^{-\frac{p}{p-1}}
b^{-\frac{1}{p-1}}
c^{\frac{p}{p-1}}.
\label{eq:optimized_Young}
\end{equation}
\end{lemma}

For arbitrary
\(
b_1,b_2>0
\),
write
\(
c_W\|\tilde{\bW}\|
=
\frac{c_W}{2}\|\tilde{\bW}\|
+
\frac{c_W}{2}\|\tilde{\bW}\|
\).
Applying Lemma~\ref{lem:optimized_young} to the first term with
\(
s=\|\tilde{\bW}\|
\),
\(
c=c_W/2
\),
\(
b=b_1
\),
and
\(
p=\gamma_1+1
\),
and to the second term with
\(
s=\|\tilde{\bW}\|
\),
\(
c=c_W/2
\),
\(
b=b_2
\),
and
\(
p=\gamma_2+1
\),
yields
\(
c_W\|\tilde{\bW}\|+c_0+c_f
\leq
b_1\|\tilde{\bW}\|^{\gamma_1+1}
+
b_2\|\tilde{\bW}\|^{\gamma_2+1}
+
\Delta_x
\),
where
\begin{equation}
\begin{aligned}
\Delta_x
:={}&
c_0+c_f
+
\gamma_1
(\gamma_1+1)^{
-\frac{\gamma_1+1}{\gamma_1}}
b_1^{-\frac1{\gamma_1}}
\left(
\frac{c_W}{2}
\right)^{
\frac{\gamma_1+1}{\gamma_1}}
\\
&+
\gamma_2
(\gamma_2+1)^{
-\frac{\gamma_2+1}{\gamma_2}}
b_2^{-\frac1{\gamma_2}}
\left(
\frac{c_W}{2}
\right)^{
\frac{\gamma_2+1}{\gamma_2}} .
\end{aligned}
\label{eq:Delta_x}
\end{equation}

Thus,
\begin{equation}
\dot V^*
\leq
-\kappa_1(V^*)^{pr}
-\kappa_2(V^*)^{qr}
+
b_1\|\tilde{\bW}\|^{\gamma_1+1}
+
b_2\|\tilde{\bW}\|^{\gamma_2+1}
+
\Delta_x.
\label{eq:Vstar_learned_bound}
\end{equation}

Since \(V^*\) is continuous on the compact set
\(\Omega\), it is uniformly bounded on \(\Omega\). Define
\(
\bar V^*
:=
\sup_{\bm\xi\in\Omega}
V^*(\bm\xi)
\)
and
\(
\bar I_x
:=
T\bar V^*
\).
Hence,
\(
0\leq\mathcal I_x(t)\leq\bar I_x
\).
Moreover,
\(
\dot V^*
=
\nabla V^{*\top}\dot{\bx}
\)
is continuous on compact
\(\Omega\)
under the preceding regularity and boundedness conditions;
therefore
\(
|\dot V^*(\bx(t))|
\leq
L_{\dot V}
\)
for some finite
\(
L_{\dot V}>0
\).

\begin{lemma}[Derivative-limited low-order window bound]
\label{lem:low_window}
Let
\(
y:[t-T,t]\rightarrow\mathbb R_{\geq0}
\)
be absolutely continuous,
\(
|\dot y(\tau)|\leq L_{\dot y}
\),
and
\(
\int_{t-T}^{t}y(\tau)d\tau\leq\bar I
\).
For
\(
0<\gamma_1<1
\)
and
\(
\mu=(\gamma_1+1)/2
\),
\begin{equation}
\int_{t-T}^{t}
y^{\gamma_1}(\tau)d\tau
\geq
\rho_{\gamma_1,T}
\left(
\int_{t-T}^{t}
y(\tau)d\tau
\right)^\mu,
\label{eq:low_integral_bound}
\end{equation}
where
\begin{equation}
\rho_{\gamma_1,T}
:=
\min
\left\{
2^{\gamma_1-1}
L_{\dot y}^{-\frac{1-\gamma_1}{2}},
\;
4^{\gamma_1-1}
T^{1-\gamma_1}
\bar I^{\frac{\gamma_1-1}{2}}
\right\}.
\label{eq:rho_low}
\end{equation}
\end{lemma}

\begin{proof}
    See~\ref{app:proof_derivative_limited_window}.
\end{proof}

For
\(
y(\tau)=V^*(\bx(\tau))
\),
\begin{equation}
\rho_{\gamma_1,T}
=
\min
\left\{
2^{\gamma_1-1}
L_{\dot V}^{-\frac{1-\gamma_1}{2}},
\;
4^{\gamma_1-1}
T^{1-\gamma_1}
\bar I_x^{\frac{\gamma_1-1}{2}}
\right\}.
\label{eq:rho_state_window}
\end{equation}

The certified deadline is
\(
T_p=T_E+T_W+T+T_X
\):
\(
T
\)
flushes pre-learning information from the moving window, and
\(
T_X
\)
bounds the remaining composite convergence time.

Let
\(
T_p>T_E+T
\)
be assigned, choose
\(
\chi\in(0,1)
\),
and set
\begin{equation}
T_W
:=
\chi(T_p-T_E-T),
\qquad
T_X
:=
(1-\chi)(T_p-T_E-T),
\label{eq:deadline_allocation}
\end{equation}
so that
\(
T_p=T_E+T_W+T+T_X
\).

From Theorem~\ref{thm:predefined_time_critic},
\(
\|\tilde{\bW}(t)\|
\leq
\bar W_c
\)
for
\(
t\geq T_E+T_W
\),
where
\begin{equation}
\bar W_c
:=
\sqrt{
2\lambda_{\max}(\bm\Gamma)\bar V_W
}.
\label{eq:Wbar_closed_loop}
\end{equation}
We have
\(
\int_{t-T}^{t}
\|\tilde{\bW}(\tau)\|^{\gamma_1+1}d\tau
\leq
T\bar W_c^{\gamma_1+1}
\)
and
\(
\int_{t-T}^{t}
\|\tilde{\bW}(\tau)\|^{\gamma_2+1}d\tau
\leq
T\bar W_c^{\gamma_2+1}
\).


Moreover,
\(
\dot{\mathcal I}_x
=
V^*(\bx(t))
-
V^*(\bx(t-T))
=
\int_{t-T}^{t}
\dot V^*(\bx(\tau))d\tau
\).
Integrating
\eqref{eq:Vstar_learned_bound}
over
\(
[t-T,t]
\)
gives
\[
\begin{aligned}
\dot{\mathcal I}_x
\leq{}&
-\kappa_1
\int_{t-T}^{t}
(V^*)^{pr} d\tau
-
\kappa_2
\int_{t-T}^{t}
(V^*)^{qr} d\tau
\\
&+
b_1
\int_{t-T}^{t}
\|\tilde{\bW}\|^{\gamma_1+1}d\tau
+
b_2
\int_{t-T}^{t}
\|\tilde{\bW}\|^{\gamma_2+1}d\tau
+
T\Delta_x .
\end{aligned}
\]

Lemma~\ref{lem:low_window} gives
\(
\int_{t-T}^{t}
(V^*)^{\gamma_1}d\tau
\geq
\rho_{\gamma_1,T}
\mathcal I_x^\mu
\),
whereas
\(
\nu>1
\)
and Jensen's inequality yield
\(
\int_{t-T}^{t}
(V^*)^\nu d\tau
\geq
T^{1-\nu}
\mathcal I_x^\nu
\).
Therefore, defining
\(
a_x
:=
\kappa_1\rho_{\gamma_1,T}
\)
and
\(
b_x
:=
\kappa_2T^{1-\nu}
\),
and
\(
\Delta_I
:=
T
\left[
b_1\bar W_c^{\gamma_1+1}
+
b_2\bar W_c^{\gamma_2+1}
+
\Delta_x
\right]
\),
one obtains
\(
\dot{\mathcal I}_x
\leq
-a_x\mathcal I_x^\mu
-b_x\mathcal I_x^\nu
+\Delta_I
\)
for
\(
t\geq t_0
\).

From Theorem~\ref{thm:predefined_time_critic} and \(\mu
=
\mu_W, \nu
=
\nu_W\),
\(
\dot V_W
\leq
-h_1V_W^\mu
-h_2V_W^\nu
+\Delta_W
\),
where
\(
h_1:=\alpha_Wa_{W\Delta}
\)
and
\(
h_2:=\alpha_Wb_{W\Delta}
\).
Hence,
\(
\dot{\mathcal J}
\leq
-a_x\mathcal I_x^\mu
-b_x\mathcal I_x^\nu
-h_1V_W^\mu
-h_2V_W^\nu
+
\Delta_I+\Delta_W
\).

Using
\(
a^\mu+b^\mu\geq(a+b)^\mu
\)
for
\(
0<\mu<1
\)
and
\(
a^\nu+b^\nu
\geq
2^{1-\nu}(a+b)^\nu
\)
for
\(
\nu>1
\),
define
\(
C_1
:=
\min\{a_x,h_1\}
\),
\(
C_2
:=
2^{1-\nu}
\min\{b_x,h_2\}
\),
and
\(
\Delta_J
:=
\Delta_I+\Delta_W
\).
Then
\begin{equation}
\dot{\mathcal J}
\leq
-C_1\mathcal J^\mu
-C_2\mathcal J^\nu
+\Delta_J,
\qquad
t\geq t_0.
\label{eq:J_comparison}
\end{equation}

Define
\(
\mathfrak C_{\mu,\nu}
:=
\frac{
\Gamma_E\!\left(
\dfrac{\nu-1}{\nu-\mu}
\right)
\Gamma_E\!\left(
\dfrac{1-\mu}{\nu-\mu}
\right)
}{
\nu-\mu
}
\)
and
\(
\Lambda_X
:=
\frac{
\mathfrak C_{\mu,\nu}
}{
(1-\theta_J)T_X
}
\),
with
\(
\theta_J\in(0,1)
\).



Since
\(
a_x
=
\frac{\gamma_{p,q,r}}{T_s}
c_r\alpha^r
\rho_{\gamma_1,T}
\)
and
\(
b_x
=
\frac{\gamma_{p,q,r}}{T_s}
c_r\beta^r
T^{1-\nu}
\),
the conditions
\(
a_x\geq\Lambda_X
\)
and
\(
b_x\geq2^{\nu-1}\Lambda_X
\)
are respectively equivalent to
\(
\frac{\gamma_{p,q,r}}{T_s}
c_r\alpha^r
\rho_{\gamma_1,T}
\geq
\Lambda_X
\)
and
\(
\frac{\gamma_{p,q,r}}{T_s}
c_r\beta^r
T^{1-\nu}
\geq
2^{\nu-1}\Lambda_X
\).
Since all coefficients are positive, these inequalities give
\(
T_s
\leq
\gamma_{p,q,r}
c_r\alpha^r
\rho_{\gamma_1,T}/\Lambda_X
\)
and
\(
T_s
\leq
\gamma_{p,q,r}
c_r\beta^r
T^{1-\nu}/
(2^{\nu-1}\Lambda_X)
\),
respectively.
Therefore, both conditions hold whenever
\begin{equation}
T_s
\leq
\min
\left\{
\frac{
\gamma_{p,q,r}
c_r\alpha^r
\rho_{\gamma_1,T}
}{
\Lambda_X
},
\;
\frac{
\gamma_{p,q,r}
c_r\beta^r
T^{1-\nu}
}{
2^{\nu-1}\Lambda_X
}
\right\}.
\label{eq:Ts_closed_loop_condition}
\end{equation}

Likewise,
\(
h_1=\alpha_Wa_{W\Delta}
\)
and
\(
h_2=\alpha_Wb_{W\Delta}
\).
To guarantee
\(
h_1\geq\Lambda_X
\)
and
\(
h_2\geq2^{\nu-1}\Lambda_X
\),
it is sufficient that
\(
\alpha_Wa_{W\Delta}
\geq
\Lambda_X
\)
and
\(
\alpha_Wb_{W\Delta}
\geq
2^{\nu-1}\Lambda_X
\),
which give
\(
\alpha_W
\geq
\Lambda_X/a_{W\Delta}
\)
and
\(
\alpha_W
\geq
2^{\nu-1}\Lambda_X/b_{W\Delta}
\),
respectively.
In addition, the critic predefined-time condition requires
\(
\alpha_W
\geq
\gamma_{W\Delta}/(\chi_WT_W)
\).
Hence, all critic-side requirements are simultaneously satisfied by selecting
\begin{equation}
\alpha_W
\geq
\max
\left\{
\frac{\gamma_{W\Delta}}{\chi_WT_W},
\;
\frac{\Lambda_X}{a_{W\Delta}},
\;
\frac{2^{\nu-1}\Lambda_X}{b_{W\Delta}}
\right\}.
\label{eq:alphaW_closed_loop_condition}
\end{equation}

\begin{theorem}[Practical predefined-time closed-loop stability]
\label{thm:closed_loop}
Suppose that the preceding compact-set bounds hold, the hypotheses of
Theorem~\ref{thm:predefined_time_critic} are satisfied, and
\(
\mu
=
\mu_W
=
(pr+1)/2
=
(\gamma_1+1)/2
\),
\(
\nu
=
\nu_W
=
qr
=
(\gamma_2+1)/2
\),
with
\(
\gamma_1=pr
\)
and
\(
\gamma_2=2qr-1
\).
Let
\(
T_p>T_E+T
\),
choose
\(
T_W
\)
and
\(
T_X
\)
according to
\eqref{eq:deadline_allocation}, and let
\(
T_s
\)
and
\(
\alpha_W
\)
satisfy
\eqref{eq:Ts_closed_loop_condition}
and
\eqref{eq:alphaW_closed_loop_condition}.
Then
\(
\mathcal J(t)
\)
enters no later than
\(
T_p
\)
the forward-invariant set
\begin{equation}
\mathcal B_J
:=
\left\{
\mathcal J\geq0:
\mathcal J\leq\bar{\mathcal J}
\right\},
\qquad
C_1\bar{\mathcal J}^{\mu}
+
C_2\bar{\mathcal J}^{\nu}
=
\frac{\Delta_J}{\theta_J}.
\label{eq:B_J}
\end{equation}
Moreover,
\begin{equation}
\|\bx(t)\|
\leq
r_x
:=
\sqrt{
\frac1{\underline v}
\max
\left\{
2L_{\dot V}T,
\frac{2\bar{\mathcal J}}{T}
\right\}
},
\qquad
t\geq T_p,
\label{eq:pointwise_x}
\end{equation}
and
\begin{equation}
\|\tilde{\bW}(t)\|
\leq
r_W
:=
\min
\left\{
\bar W_c,
\sqrt{
2\lambda_{\max}(\bm\Gamma)
\bar{\mathcal J}
}
\right\},
\qquad
t\geq T_p.
\label{eq:pointwise_W}
\end{equation}
Hence
\(
\operatorname{col}\{\bx,\tilde{\bW}\}
\)
is practically predefined-time stable.
If
\(
\bar\varepsilon_B(t)\equiv0
\),
\(
\bar\varepsilon_{B,j}=0
\)
for
\(
j=1,\ldots,M
\),
and
\(
\bar\varepsilon_g=0
\),
then
\(
\bx(t)=\bm0
\)
and
\(
\tilde{\bW}(t)=\bm0
\)
for all
\(
t\geq T_p
\).
\end{theorem}

\begin{proof}
From
\eqref{eq:alphaW_closed_loop_condition},
\(
\alpha_W
\geq
\gamma_{W\Delta}/(\chi_WT_W)
\),
hence
\(
\|\tilde{\bW}(t)\|
\leq
\bar W_c
\)
for
\(
t\geq T_E+T_W
\).
Moreover,
\(
h_1
=
\alpha_Wa_{W\Delta}
\geq
\Lambda_X
\)
and
\(
h_2
=
\alpha_Wb_{W\Delta}
\geq
2^{\nu-1}\Lambda_X
\).
From
\eqref{eq:Ts_closed_loop_condition},
\(
a_x
=
\frac{\gamma_{p,q,r}}{T_s}
c_r\alpha^r
\rho_{pr,T}
\geq
\Lambda_X
\)
and
\(
b_x
=
\frac{\gamma_{p,q,r}}{T_s}
c_r\beta^r
T^{1-\nu}
\geq
2^{\nu-1}\Lambda_X
\).
Thus
\(
C_1
=
\min\{a_x,h_1\}
\geq
\Lambda_X
\)
and
\(
C_2
=
2^{1-\nu}\min\{b_x,h_2\}
\geq
\Lambda_X
\).

Define
\(
F(s)
:=
C_1s^\mu+C_2s^\nu
\).
Then
\(
F(0)=0
\),
\(
F'(s)
=
\mu C_1s^{\mu-1}
+
\nu C_2s^{\nu-1}
>0
\)
for
\(
s>0
\),
and
\(
F(s)\rightarrow\infty
\)
as
\(
s\rightarrow\infty
\).
Hence the equation
\(
F(\bar{\mathcal J})
=
\Delta_J/\theta_J
\)
has a unique
\(
\bar{\mathcal J}\geq0
\).

For
\(
\mathcal J>\bar{\mathcal J}
\),
\(
F(\mathcal J)
>
\Delta_J/\theta_J
\),
so
\(
\Delta_J
<
\theta_J
[
C_1\mathcal J^\mu
+
C_2\mathcal J^\nu
]
\).
Substitution into
\eqref{eq:J_comparison}
gives
\begin{equation}
\dot{\mathcal J}
\leq
-(1-\theta_J)
\left[
C_1\mathcal J^\mu
+
C_2\mathcal J^\nu
\right].
\label{eq:J_outside}
\end{equation}

Let
\(
\eta_1
:=
(\nu-1)/(\nu-\mu)
\)
and
\(
\eta_2
:=
(1-\mu)/(\nu-\mu)
\),
so that
\(
\eta_1,\eta_2>0
\)
and
\(
\eta_1+\eta_2=1
\).
From
\eqref{eq:J_outside},
\[
dt
\leq
-\frac{
d\mathcal J
}{
(1-\theta_J)
[
C_1\mathcal J^\mu
+
C_2\mathcal J^\nu
]
}.
\]
Hence, with
\(
t_0=T_E+T_W+T
\),
\[
\begin{aligned}
T_{\mathcal J}
&\leq
\frac1{1-\theta_J}
\int_{\bar{\mathcal J}}^{\mathcal J(t_0)}
\frac{ds}{
C_1s^\mu+C_2s^\nu
}
\leq
\frac1{1-\theta_J}
\int_0^\infty
\frac{ds}{
C_1s^\mu+C_2s^\nu
}
\\
&=
\frac{
\mathfrak C_{\mu,\nu}
}{
(1-\theta_J)
C_1^{\eta_1}
C_2^{\eta_2}
}.
\end{aligned}
\]
Since
\(
C_1,C_2\geq\Lambda_X
\),
\(
C_1^{\eta_1}
C_2^{\eta_2}
\geq
\Lambda_X^{\eta_1+\eta_2}
=
\Lambda_X
\),
and therefore
\(
T_{\mathcal J}
\leq
\mathfrak C_{\mu,\nu}/[(1-\theta_J)\Lambda_X]
=
T_X
\).
Thus
\(
t_0+T_{\mathcal J}
\leq
T_E+T_W+T+T_X
=
T_p
\).

At
\(
\mathcal J=\bar{\mathcal J}
\),
\[
\dot{\mathcal J}
\leq
-C_1\bar{\mathcal J}^{\mu}
-C_2\bar{\mathcal J}^{\nu}
+\Delta_J
=
-\frac{\Delta_J}{\theta_J}
+\Delta_J
=
-\frac{1-\theta_J}{\theta_J}\Delta_J
\leq0,
\]
so
\(
\mathcal B_J
\)
is forward invariant.

For
\(
t\geq T_p
\),
\(
\mathcal I_x(t)
\leq
\mathcal J(t)
\leq
\bar{\mathcal J}
\).
If
\(
V^*(\bx(t))
\leq
2L_{\dot V}T
\),
the first state bound follows.
Otherwise,
\(
V^*(\bx(t))
>
2L_{\dot V}T
\),
and for every
\(
\tau\in[t-T,t]
\),
\[
V^*(\bx(\tau))
\geq
V^*(\bx(t))
-
L_{\dot V}T
>
\frac12V^*(\bx(t)).
\]

Therefore,
\[
\mathcal I_x(t)
=
\int_{t-T}^{t}
V^*(\bx(\tau))d\tau
>
\frac{T}{2}
V^*(\bx(t)),
\]
so
\(
V^*(\bx(t))
<
2\bar{\mathcal J}/T
\).
Hence
\(
V^*(\bx(t))
\leq
\max\{
2L_{\dot V}T,
2\bar{\mathcal J}/T
\}
\),
and
\(
V^*(\bx)
\geq
\underline v\|\bx\|^2
\)
gives
\eqref{eq:pointwise_x}.

Similarly,
\(
V_W(t)
\leq
\mathcal J(t)
\leq
\bar{\mathcal J}
\),
while
\[
V_W
=
\frac12
\tilde{\bW}^{\top}
\bm\Gamma^{-1}
\tilde{\bW}
\geq
\frac{
\|\tilde{\bW}\|^2
}{
2\lambda_{\max}(\bm\Gamma)
}.
\]
Thus
\(
\|\tilde{\bW}(t)\|
\leq
\sqrt{
2\lambda_{\max}(\bm\Gamma)
\bar{\mathcal J}
}
\).
Together with
\(
\|\tilde{\bW}(t)\|
\leq
\bar W_c
\),
this yields
\eqref{eq:pointwise_W}.

Finally, if
\(
\bar\varepsilon_B(t)\equiv0
\),
\(
\bar\varepsilon_{B,j}=0
\),
\(
\bar\varepsilon_g=0
\),
and
\(
\bar\Delta_f=0
\),
Theorem~\ref{thm:predefined_time_critic} gives
\(
\tilde{\bW}(t)=\bm0
\)
for
\(
t\geq T_E+T_W
\),
hence
\(
\nabla\hat V=\nabla V^*
\)
and
\(
\hat{\bu}=\bu^*
\).
Moreover,
\(
c_f=0
\),
so that for
\(
t\geq t_0
\),
\[
\dot{\mathcal I}_x
\leq
-a_x\mathcal I_x^\mu
-b_x\mathcal I_x^\nu.
\]
Therefore,
\[
T_x
\leq
\frac{
\mathfrak C_{\mu,\nu}
}{
a_x^{\eta_1}
b_x^{\eta_2}
}
\leq
\frac{
\mathfrak C_{\mu,\nu}
}{
\Lambda_X
}
=
(1-\theta_J)T_X
<
T_X.
\]
Hence
\(
\mathcal I_x(t)=0
\)
no later than
\(
T_E+T_W+T+T_X=T_p
\).
Since
\(
\mathcal I_x
=
\int_{t-T}^{t}
V^*(\bx(\tau))d\tau
\)
with
\(
V^*(\bx)\geq0
\),
continuity gives
\(
V^*(\bx(\tau))=0
\)
throughout the window, and positive definiteness yields
\(
\bx(t)=\bm0
\).
Thus
\(
\bx(t)=\bm0
\)
and
\(
\tilde{\bW}(t)=\bm0
\)
for
\(
t\geq T_p
\).
\end{proof}


\begin{corollary}[Overall-deadline-based selection of \(T_s\)]
\label{cor:Ts_from_Tp}
Let
\(
T_p>T_E+T
\)
be a designer-assigned overall convergence deadline, and choose
\(
\chi\in(0,1)
\).
With the allocation
\(
T_W=\chi(T_p-T_E-T)
\)
and
\(
T_X=(1-\chi)(T_p-T_E-T)
\),
a sufficient state-side condition in
Theorem~\ref{thm:closed_loop} is
\begin{equation}
\begin{aligned}
T_s
\leq{}&
T_s^{\max}=(1-\theta_J)(1-\chi)
(T_p-T_E-T)
\\
&\times
\min
\left\{
\frac{
\gamma_{p,q,r}
c_r\alpha^r
\rho_{pr,T}
}{
\mathfrak C_{\mu,\nu}
},
\;
\frac{
\gamma_{p,q,r}
c_r\beta^r
T^{1-\nu}
}{
2^{\nu-1}
\mathfrak C_{\mu,\nu}
}
\right\}.
\end{aligned}
\label{eq:Ts_from_Tp}
\end{equation}
\end{corollary}

\section{Simulation}
\subsection{Simulation setup}

All simulations are implemented in MATLAB R2026a. 

Consider the nonlinear control-affine plant
\(
\dot{\bx}
=
\bfv(\bx)
+
\bg(\bx)\bu
\),
where
\(
\bx=[x_1,x_2,x_3]^\top\in\mathbb R^3
\),
\(
\bu=[u_1,u_2,u_3]^\top\in\mathbb R^3
\),
and
\(
\bg(\bx)=\bm I_3
\).
The unknown drift is
\(
\bfv(\bx)
=
[f_1(\bx),f_2(\bx),f_3(\bx)]^\top
\),
with
\(
f_1(\bx)
=
-\tanh(x_1+x_2)-0.4x_1^3
\),
\(
f_2(\bx)
=
-\sin(x_2+x_3)-0.3x_1x_2
\),
and
\(
f_3(\bx)
=
-\frac{x_3}{1+x_3^2}
-0.2\tanh(x_1x_3)
\).

For the identifier, the drift dictionary is selected as
\(
\bsigma_f(\bx)
=
[
\tanh(x_1+x_2),
x_1^3,
\sin(x_2+x_3),
x_1x_2,
\frac{x_3}{1+x_3^2},
\tanh(x_1x_3)
]^\top
\),
such that
\(
\hat{\bfv}(\bx)
=
\hat{\bW}_f^\top
\bsigma_f(\bx)
\),
where
\(
\hat{\bW}_f\in\mathbb R^{6\times3}
\)
is obtained from
\eqref{eq:icl_identifier}.

The desired value function is selected as
\(
V_d(\bx)
=
\frac12\bx^\top\bx
=
\frac12(x_1^2+x_2^2+x_3^2)
\),
with
\(
r^2(\bx)
=
x_1^2+x_2^2+x_3^2
\).
The critic basis is chosen as
\(
\bphi(\bx)
=
[
x_1^2,
x_2^2,
x_3^2,
x_1x_2,
x_1x_3,
x_2x_3
]^\top
\),
and the corresponding ideal critic weight is
\(
\bW^*
=
[0.5,0.5,0.5,0,0,0]^\top
\).


Using the identified drift, define
\(
\hat F_V(\bx)
=
\nabla V_d^\top(\bx)\hat{\bfv}(\bx)
\)
and
\(
\hat A(\bx)
=
\hat F_V(\bx)
+
\kappa_V
[
2V_d^{3/4}(\bx)
+
V_d^{3/2}(\bx)
]
\),
where
\(
\kappa_V
=
\gamma_{p,q,r}/T_s
\).
The inverse-optimal running cost is
\(
L(\bx,\bu)
=
\bm L_1(\bx)\bu
+
L_2(\bx)
+
\bu^\top\bR\bu
\),
where, for
\(
\bx\neq\bm0
\),
\(
\bm L_1(\bx)
=
(
\hat A(\bx)/r^2(\bx)-1
)
\bx^\top
\)
and
\(
L_2(\bx)
=
-\hat F_V(\bx)
+
\hat A^2(\bx)/(2r^2(\bx))
\).
The associated inverse-optimal policy is
\(
\bu^*(\bx)
=
-\hat A(\bx)\bx/r^2(\bx)
\)
for
\(
\bx\neq\bm0
\),
with
\(
\bu^*(\bm0)=\bm0
\).

Two prescribed overall deadlines,
\(T_p=5~\mathrm{s}\) and
\(T_p=50~\mathrm{s}\),
are considered to evaluate the critic-learning and closed-loop
performance under different time specifications.
For each deadline, five distinct pairs of initial states and
critic weights are examined to assess the influence of initial
conditions on the learning and convergence behavior.
The numerical parameters used in both cases are summarized in
Table~\ref{tab:simulation_parameters}.

\begin{table}[!t]
\centering
\caption{Simulation parameters for the two prescribed overall deadlines.}
\label{tab:simulation_parameters}
\renewcommand{\arraystretch}{1.12}
\setlength{\tabcolsep}{3.5pt}
\resizebox{\columnwidth}{!}{%
\begin{tabular}{lcc}
\hline
\textbf{Parameter}
&
\textbf{\(T_p=5~\mathrm{s}\)}
&
\textbf{\(T_p=50~\mathrm{s}\)}
\\
\hline

\(\bR\)
&
\(0.5\bm I_3\)
&
\(0.5\bm I_3\)
\\

\(\alpha,\beta,p,q,r\)
&
\(2,\;1,\;\frac34,\;\frac32,\;1\)
&
\(2,\;1,\;\frac34,\;\frac32,\;1\)
\\

\(\gamma_{p,q,r}\)
&
\(3.046992\)
&
\(3.046992\)
\\

\(T_E\)
&
\(0.5~\mathrm{s}\)
&
\(5~\mathrm{s}\)
\\

\(T\)
&
\(0.125~\mathrm{s}\)
&
\(0.15~\mathrm{s}\)
\\

\(\chi\)
&
\(0.30\)
&
\(0.30\)
\\

\(T_W\)
&
\(1.3125~\mathrm{s}\)
&
\(13.455~\mathrm{s}\)
\\

\(T_X\)
&
\(3.0625~\mathrm{s}\)
&
\(31.395~\mathrm{s}\)
\\

\(T_s^{\max}\)
&
\(1.091180~\mathrm{s}\)
&
\(10.211511~\mathrm{s}\)
\\

\(T_s\)
&
\(0.982062~\mathrm{s}\)
&
\(9.190360~\mathrm{s}\)
\\

\(\hat{\bW}^{(i)}(0)\)
&
\(\bW_0^{(i)},\ i=1,\ldots,5\)
&
\(\bW_0^{(i)},\ i=1,\ldots,5\)
\\

\(\bm\Gamma\)
&
\(0.05\bm I_6\)
&
\(0.05\bm I_6\)
\\

\(k_1,k_2\)
&
\(0.8,\;0.15\)
&
\(0.8,\;0.15\)
\\

\(\gamma_1,\gamma_2\)
&
\(\frac34,\;2\)
&
\(\frac34,\;2\)
\\

Replay/current gains
&
\(1,\;0.1\)
&
\(1,\;0.1\)
\\

\(\chi_W\)
&
\(0.9\)
&
\(0.9\)
\\

Pre-\(T_E\) critic gain
&
\(15\)
&
\(15\)
\\

Post-\(T_E\) critic gain
&
Eq.~\eqref{eq:alphaW_closed_loop_condition}
&
Eq.~\eqref{eq:alphaW_closed_loop_condition}
\\

Replay-stack capacity
&
\(10000\)
&
\(10000\)
\\

Stored replay samples \(M\)
&
\(7501\)
&
\(9701\)
\\

Acquisition step \(\Delta t_E\)
&
\(5\times10^{-5}~\mathrm{s}\)
&
\(5\times10^{-4}~\mathrm{s}\)
\\

Post-acquisition step \(\Delta t\)
&
\(0.005~\mathrm{s}\)
&
\(0.005~\mathrm{s}\)
\\

\hline
\end{tabular}%
}
\end{table}

Here the pre- and post-$T_E$ critic gains specify the
learning gains before and after data acquisition, respectively.
The replay-stack capacity denotes the maximum number of
stored samples, whereas $M$ is the actual number of valid
sliding-window samples, given by
$M=(T_E-T)/\Delta t_E+1$.
Finally, $\Delta t_E$ and $\Delta t$ denote the numerical
integration steps during and after data acquisition, respectively.

Five distinct pairs of initial states and critic weights are
considered to evaluate the closed-loop and learning performance
under different initial conditions:
$\bx_0^{(1)}=[0.30,-0.30,0.2646]^\top$,
$\bx_0^{(2)}=[0.70,-0.50,0.30]^\top$,
$\bx_0^{(3)}=[1.00,0.60,-0.40]^\top$,
$\bx_0^{(4)}=[-1.20,0.90,0.70]^\top$, and
$\bx_0^{(5)}=[1.80,-1.80,1.60]^\top$.

The associated initial critic weights are
\begin{equation}
\begin{aligned}
\hat{\bW}^{(1)}(0)
&=[0.002,0.002,0.002,0.002,0.002,0.002]^\top,\\
\hat{\bW}^{(2)}(0)
&=[0.15,0.10,0.20,0.05,-0.05,0.10]^\top,\\
\hat{\bW}^{(3)}(0)
&=[-0.20,0.10,-0.10,0.15,-0.10,0.20]^\top,\\
\hat{\bW}^{(4)}(0)
&=[0.80,-0.30,0.70,-0.20,0.25,-0.15]^\top,\\
\hat{\bW}^{(5)}(0)
&=[-0.50,-0.40,-0.30,0.40,-0.35,0.30]^\top.
\end{aligned}
\label{eq:simulation_initial_critic_weights}
\end{equation}

Weights are initialized here to facilitate a faster simulation process. In some cases, however, this initialization can affect the stability of the closed-loop system—a point addressed in several notable studies (e.g.,~\citep{VuDoan2026WarmStart,Xu2024TrajectoryADP}). In the present work, we do not focus on a formal methodology for weight initialization; instead, the weights are pre-set based on design experience.

All five pairs are used for the comparative state-response,
critic-estimation-error, and value-function-approximation-error
simulations. Case~5 is selected to illustrate the individual
critic-weight trajectories.

\subsection{Simulation Results}


\begin{figure}[H]
    \centering
    \scalebox{1}[0.6]{%
        \includegraphics[width=\columnwidth]{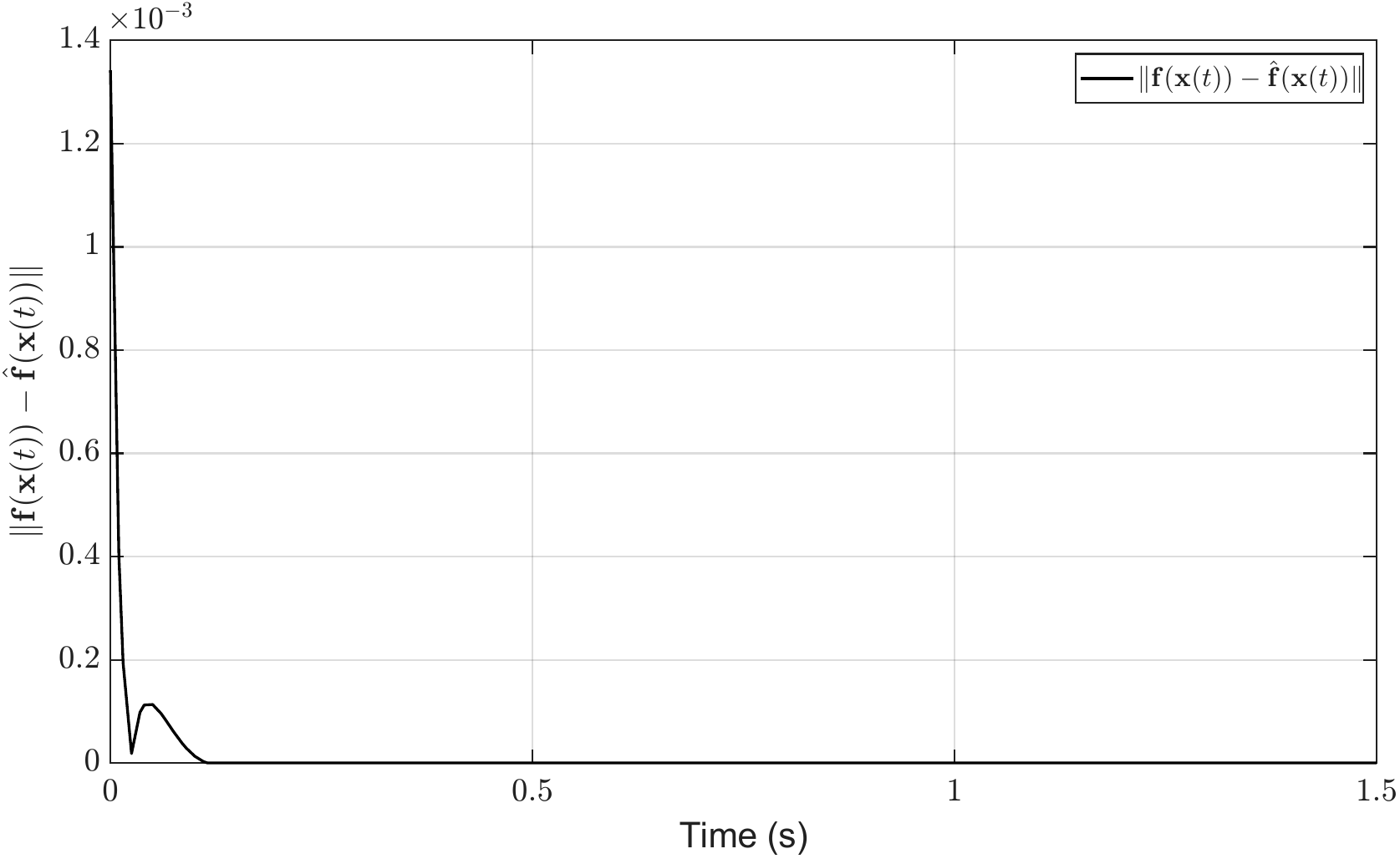}%
    }
    \caption{Identification error.}
    \label{fig:identification_error}
\end{figure}

\begin{figure}[H]
    \centering
    \scalebox{1}[0.7]{%
        \includegraphics[width=\columnwidth]{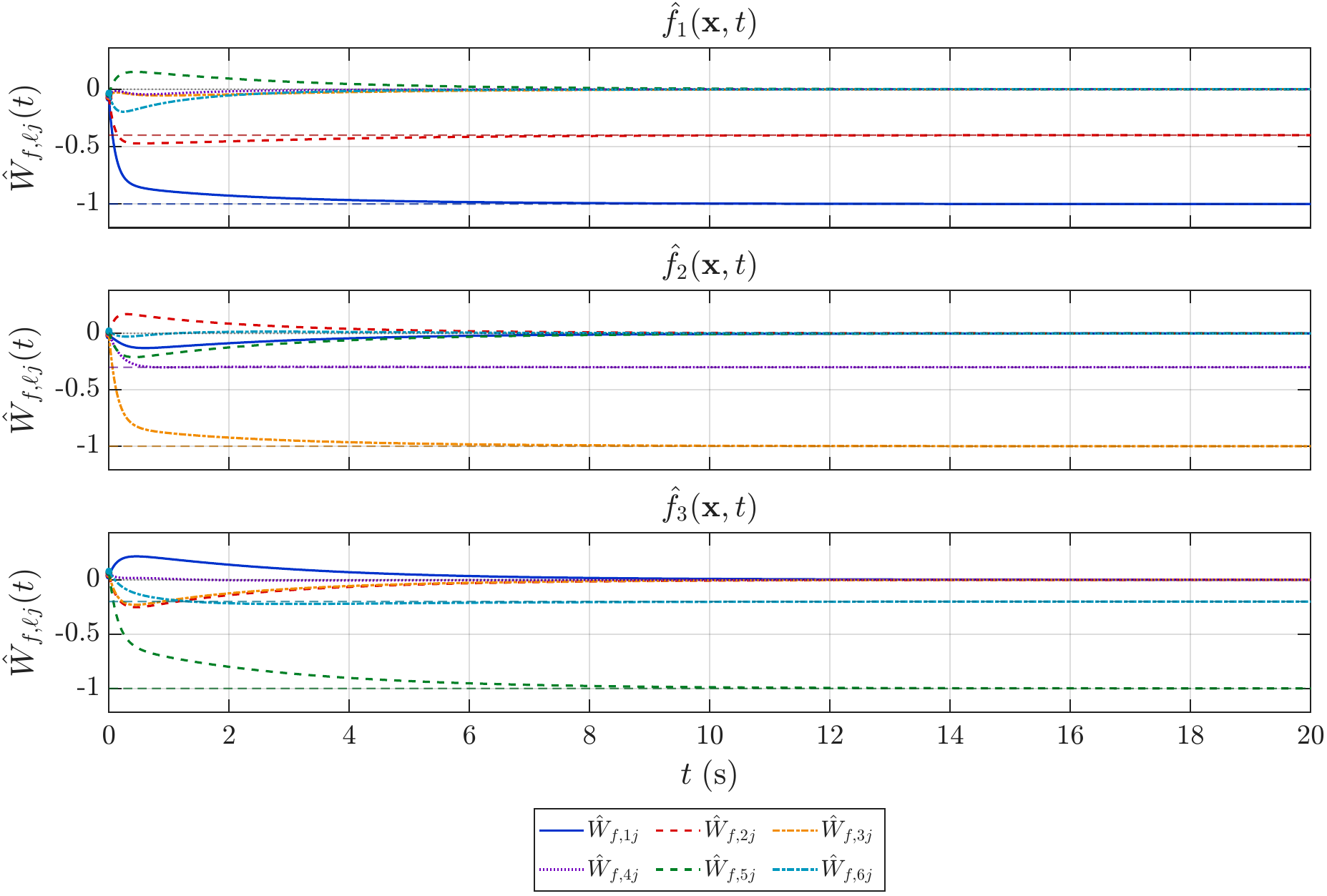}%
    }
    \caption{Online ICL identifier-weight trajectories.}
    \label{fig:identifier_weights}
\end{figure}


\subsection{Case A: \(T_p=5\,\mathrm{s}\)}
\label{subsec:Tp5_results}

\begin{figure}[H]
    \centering
    \scalebox{1}[0.7]{%
        \includegraphics[width=\columnwidth]{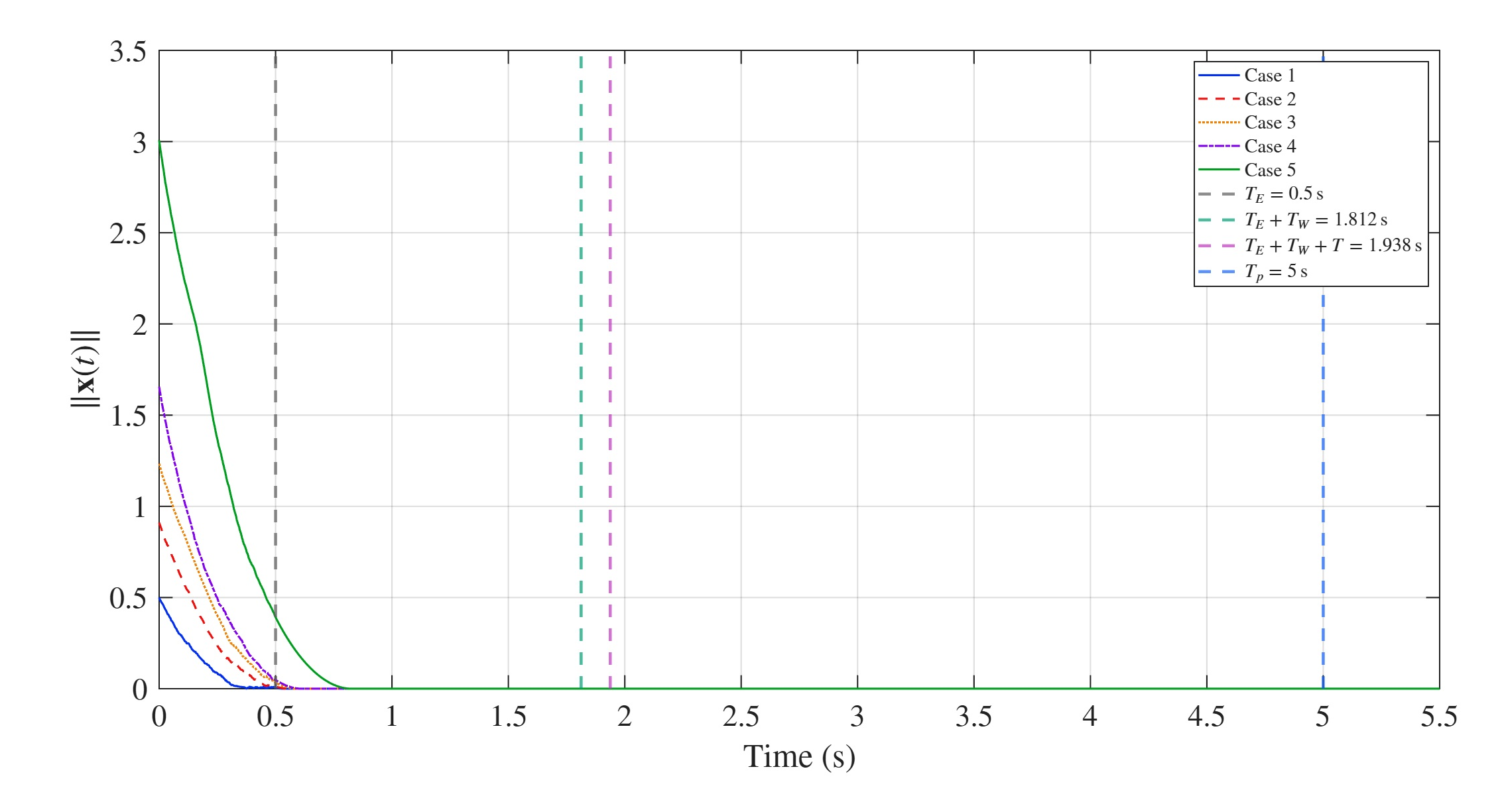}%
    }
    \caption{Predefined-time closed-loop state responses for
    \(T_p=5\,\mathrm{s}\) under five distinct initial conditions.}
    \label{fig:Tp5_state}
\end{figure}

\begin{figure}[H]
    \centering
    \scalebox{1}[1.0]{%
        \includegraphics[width=\columnwidth]{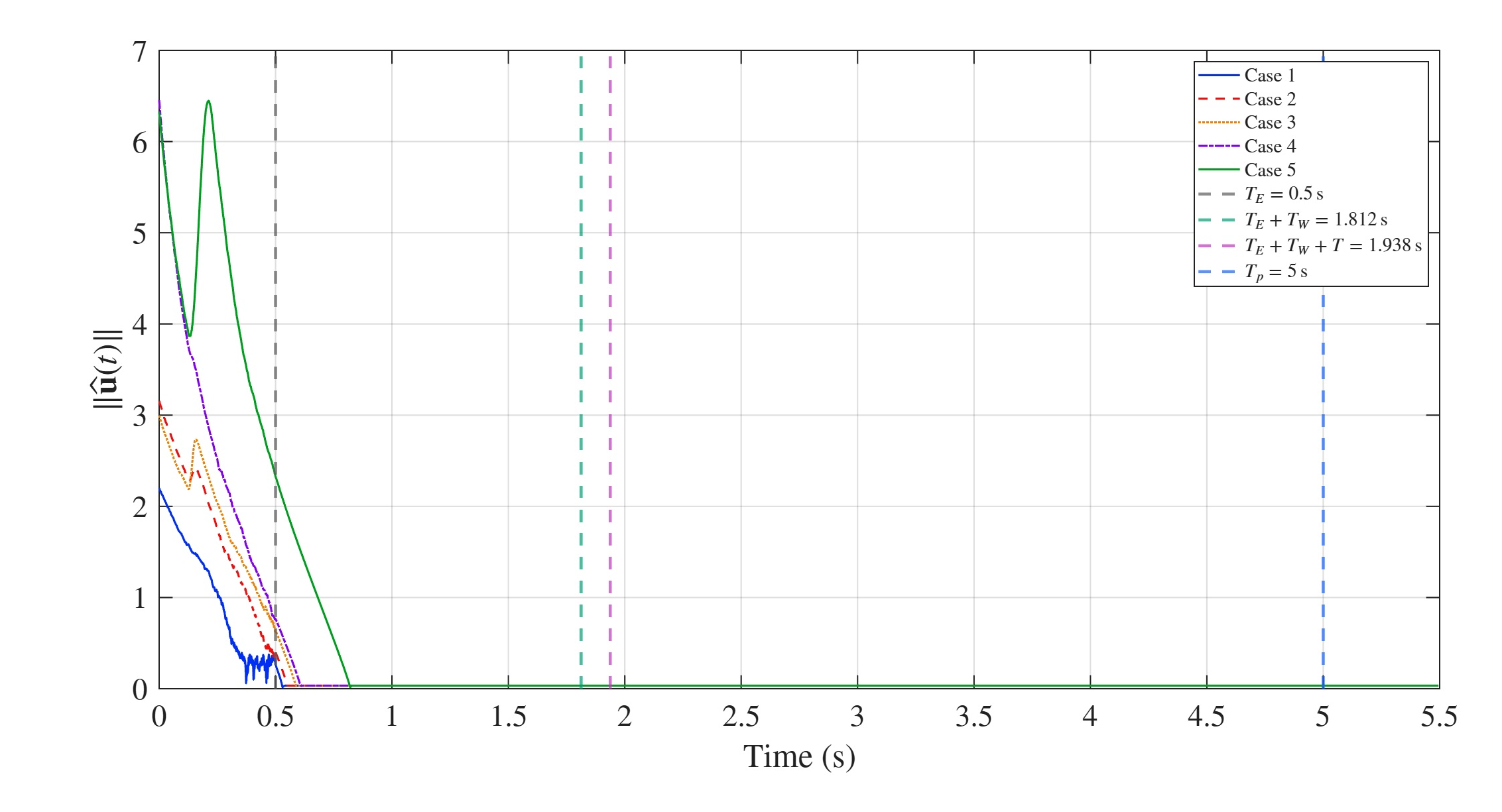}%
    }
    \caption{Evolution of the learned control policies for
    \(T_p=5\,\mathrm{s}\).}
    \label{fig:Tp5_learned_control}
\end{figure}

\begin{figure}[H]
    \centering
    \scalebox{1}[0.8]{%
        \includegraphics[width=\columnwidth]{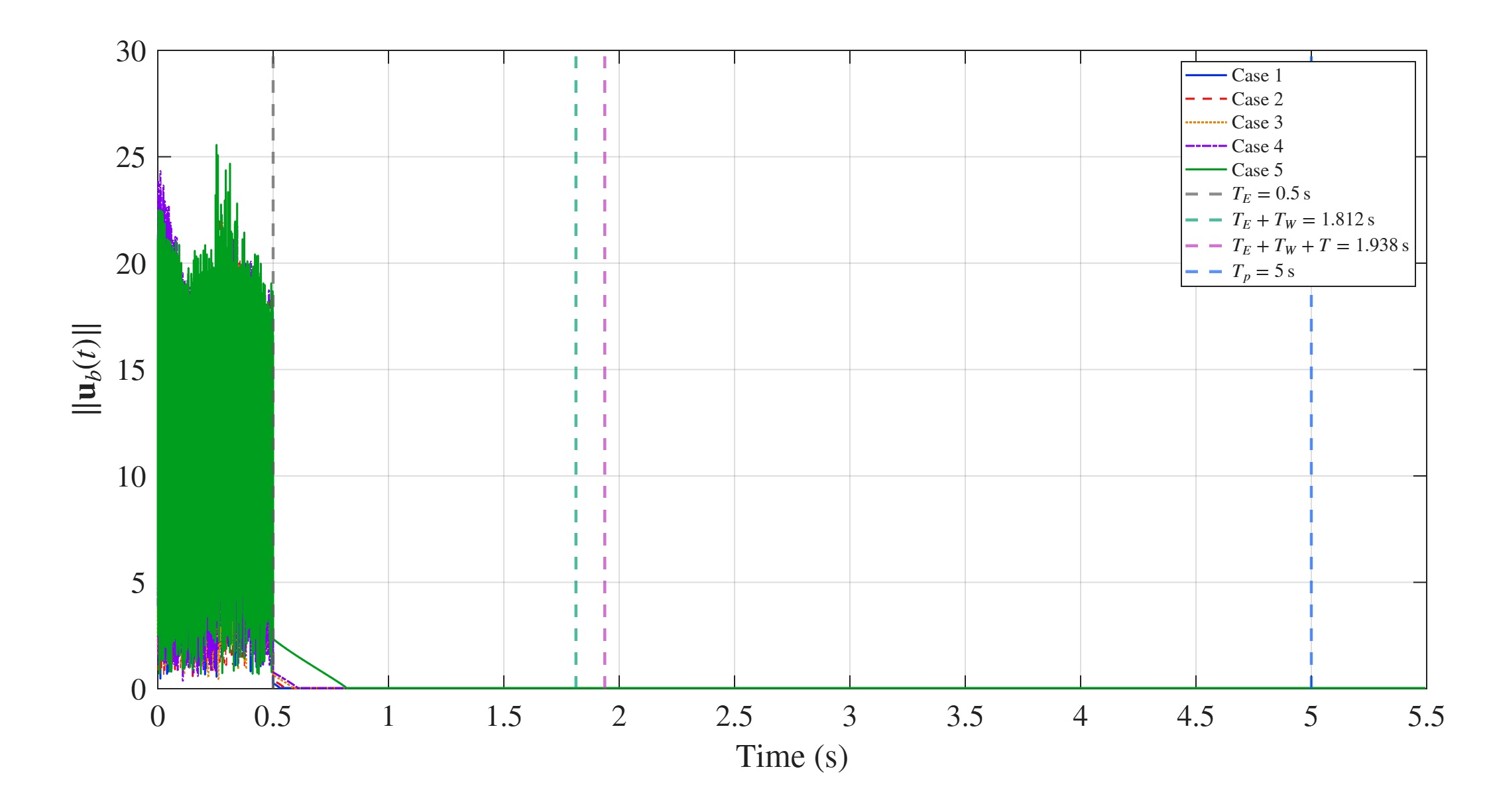}%
    }
    \caption{Behavior-control trajectories including the probing
    excitation during the replay-data acquisition interval for
    \(T_p=5\,\mathrm{s}\).}
    \label{fig:Tp5_behavior_control}
\end{figure}

\begin{figure}[H]
    \centering
    \scalebox{1}[0.8]{%
        \includegraphics[width=\columnwidth]{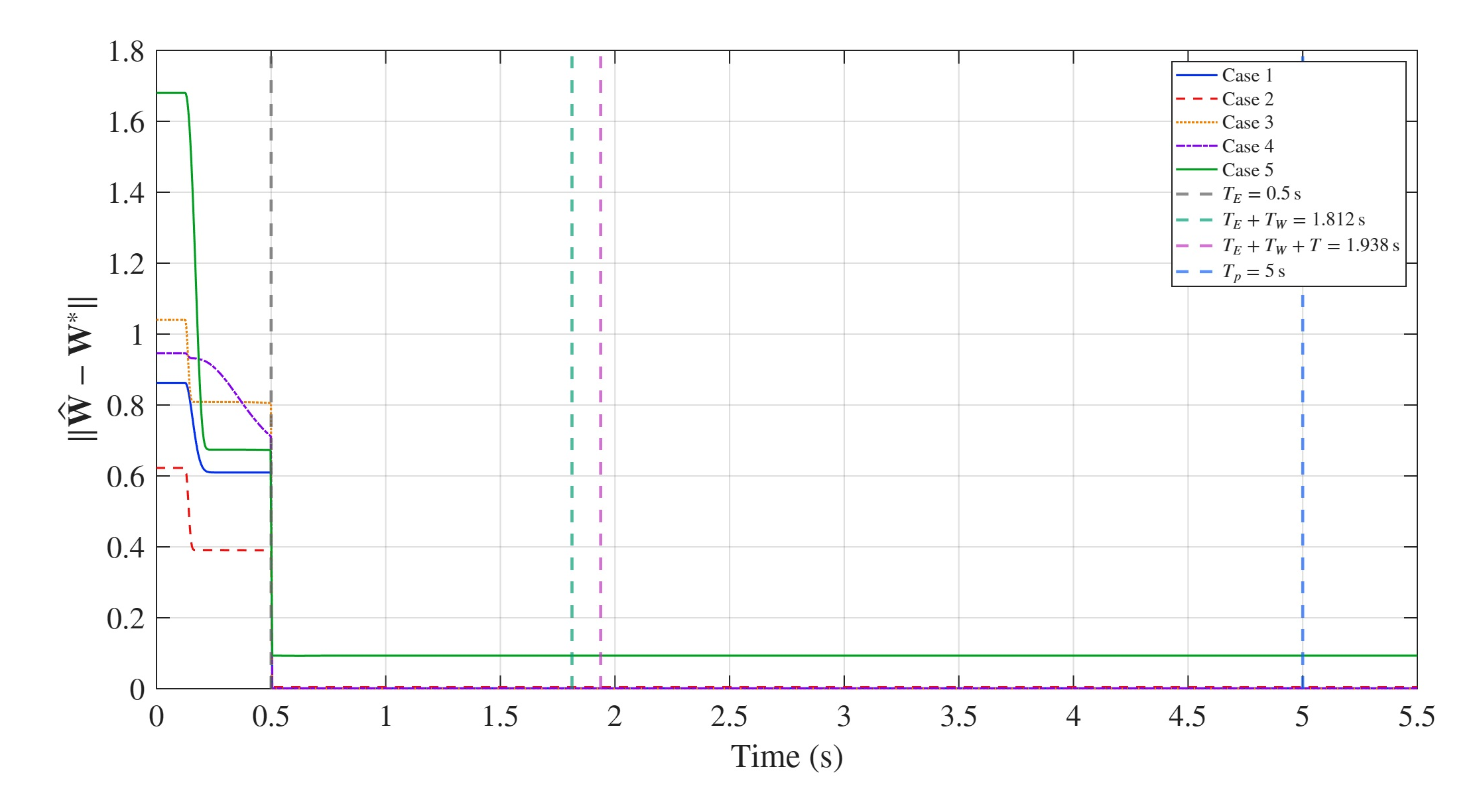}%
    }
    \caption{Evolution of the critic-weight estimation errors
    under five distinct initial conditions for
    \(T_p=5\,\mathrm{s}\).}
    \label{fig:Tp5_critic_error}
\end{figure}

\begin{figure}[H]
    \centering
    \scalebox{1}[0.8]{%
        \includegraphics[width=\columnwidth]{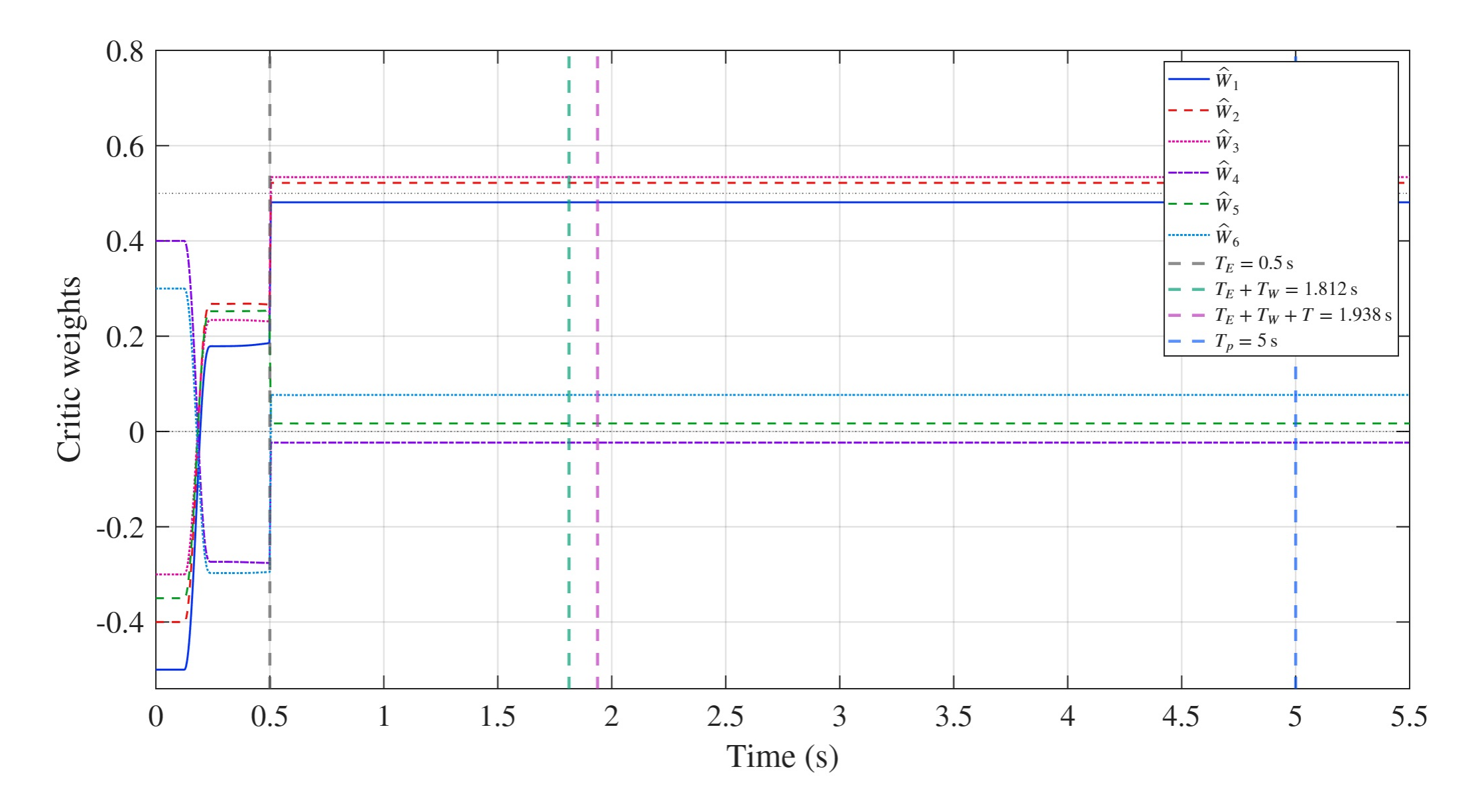}%
    }
    \caption{Trajectories of the six critic weights for
    Case 5 with \(T_p=5\,\mathrm{s}\).}
    \label{fig:Tp5_critic_weights}
\end{figure}

\begin{figure}[H]
    \centering
    \scalebox{1}[0.8]{%
        \includegraphics[width=\columnwidth]{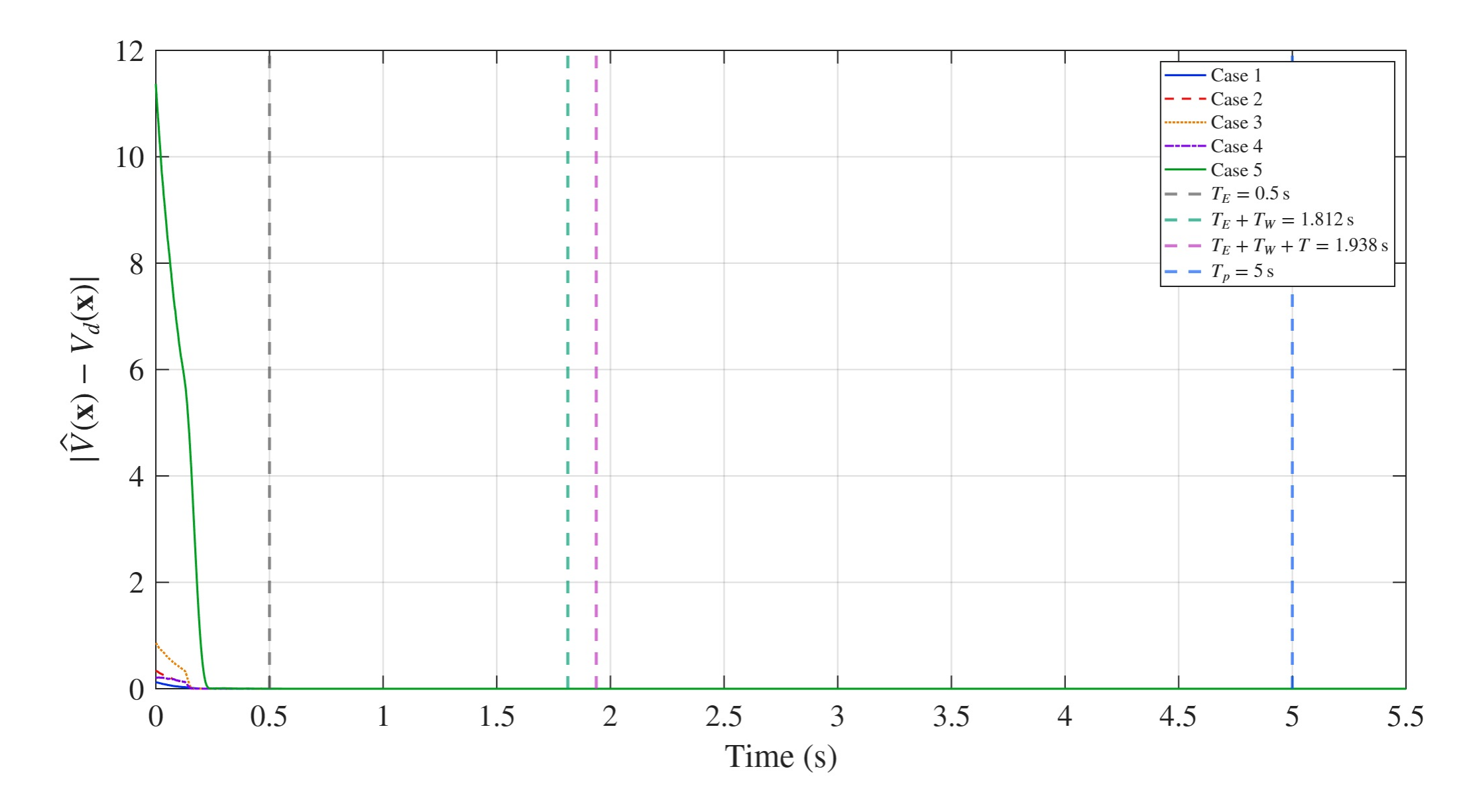}%
    }
    \caption{Value-function approximation errors under five
    distinct initial conditions for
    \(T_p=5\,\mathrm{s}\).}
    \label{fig:Tp5_value_error}
\end{figure}

\begin{figure}[H]
    \centering
    \scalebox{1}[0.8]{%
        \includegraphics[width=\columnwidth]{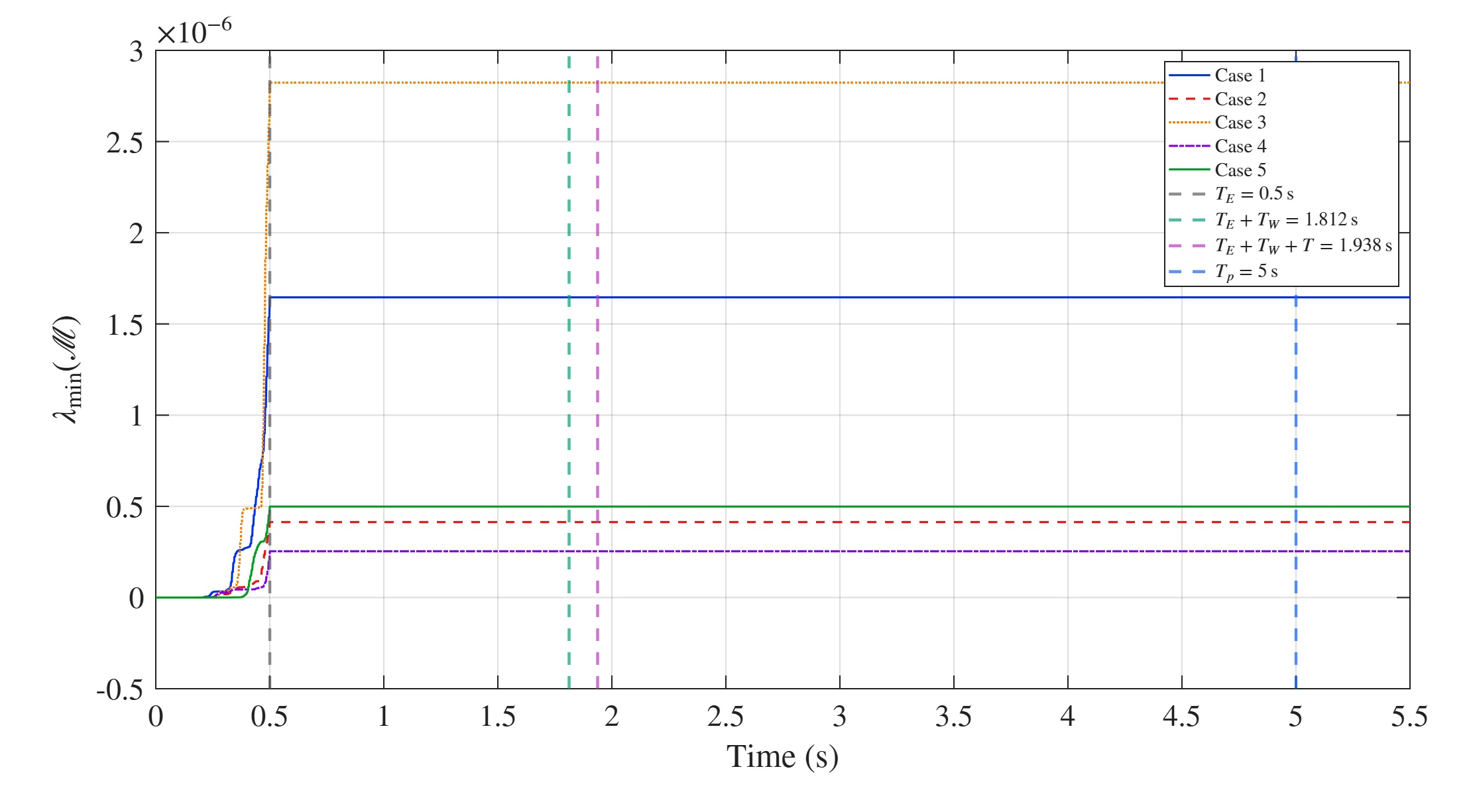}%
    }
    \caption{Evolution of the minimum eigenvalue of the
    normalized experience-replay Gramian during data acquisition
    and after replay-stack freezing for
    \(T_p=5\,\mathrm{s}\).}
    \label{fig:Tp5_replay_eigenvalue}
\end{figure}


\subsection{Case B: \(T_p=50\,\mathrm{s}\)}
\label{subsec:Tp50_results}

\begin{figure}[H]
    \centering
    \scalebox{1}[0.8]{%
        \includegraphics[width=\columnwidth]{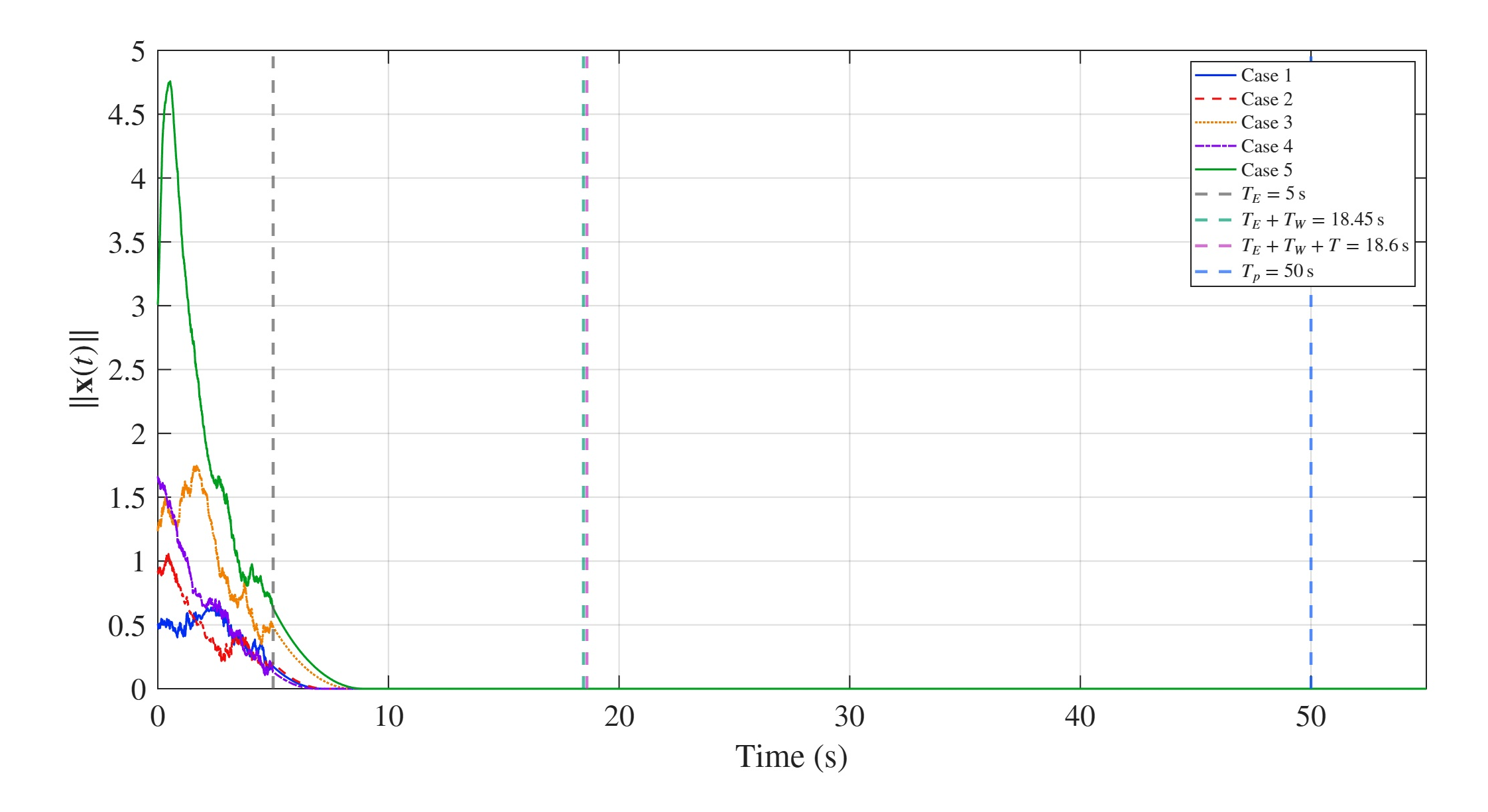}%
    }
    \caption{Predefined-time closed-loop state responses for
    \(T_p=50\,\mathrm{s}\) under five distinct initial conditions.}
    \label{fig:Tp50_state}
\end{figure}

\begin{figure}[H]
    \centering
    \scalebox{1}[0.8]{%
        \includegraphics[width=\columnwidth]{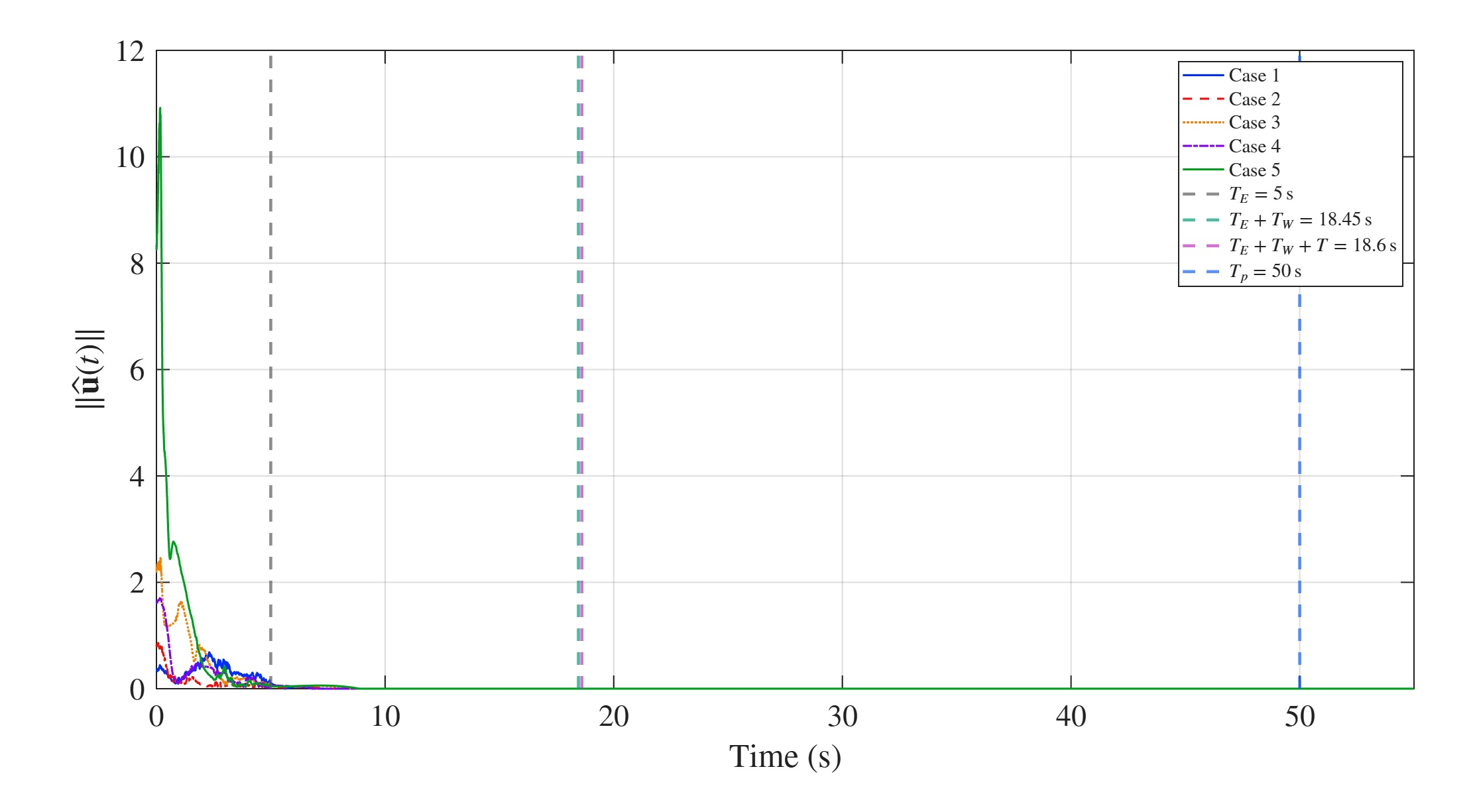}%
    }
    \caption{Evolution of the learned control policies for
    \(T_p=50\,\mathrm{s}\).}
    \label{fig:Tp50_learned_control}
\end{figure}

\begin{figure}[H]
    \centering
    \scalebox{1}[0.8]{%
        \includegraphics[width=\columnwidth]{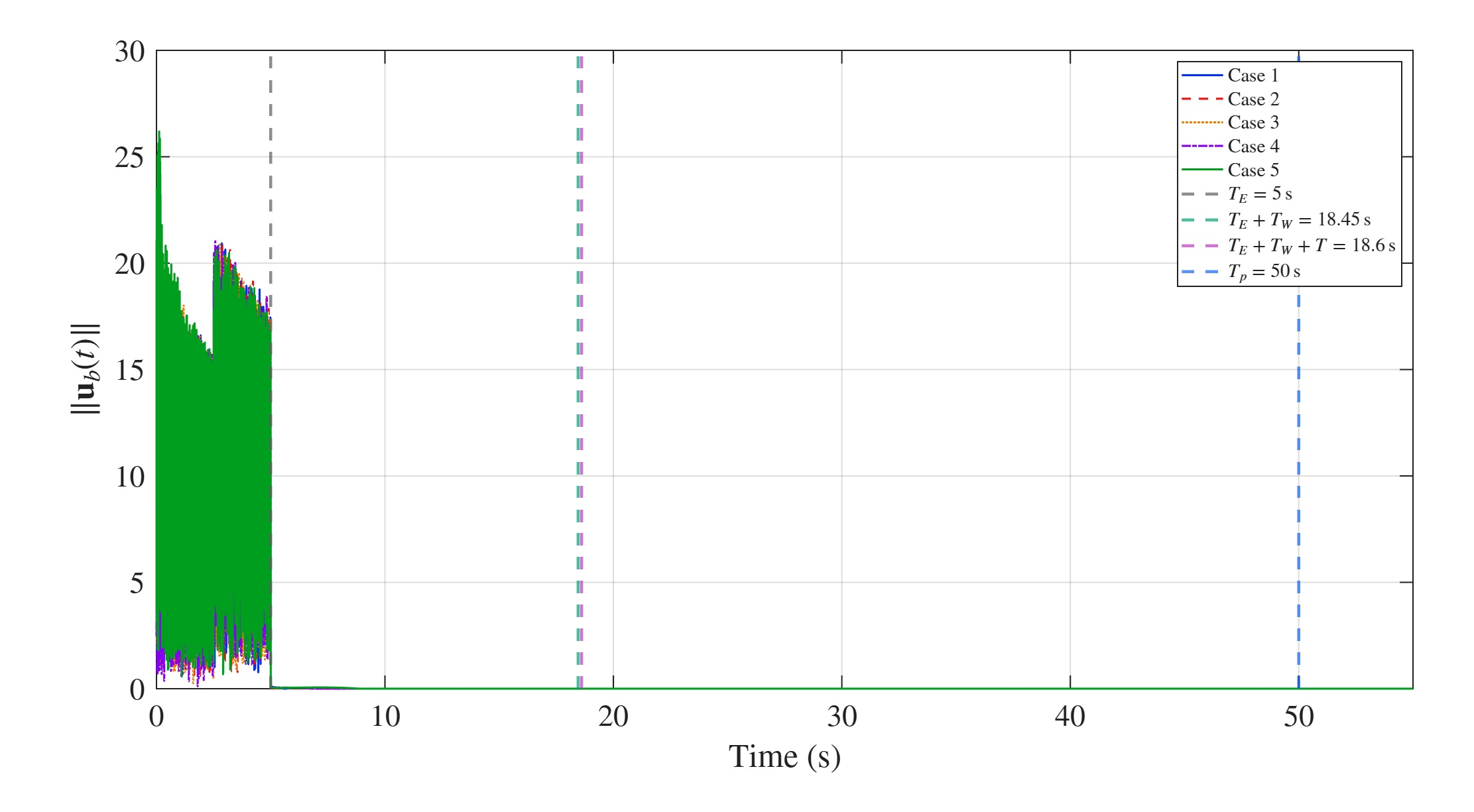}%
    }
    \caption{Behavior-control trajectories including the probing
    excitation during the replay-data acquisition interval for
    \(T_p=50\,\mathrm{s}\).}
    \label{fig:Tp50_behavior_control}
\end{figure}

\begin{figure}[H]
    \centering
    \scalebox{1}[1.0]{%
        \includegraphics[width=\columnwidth]{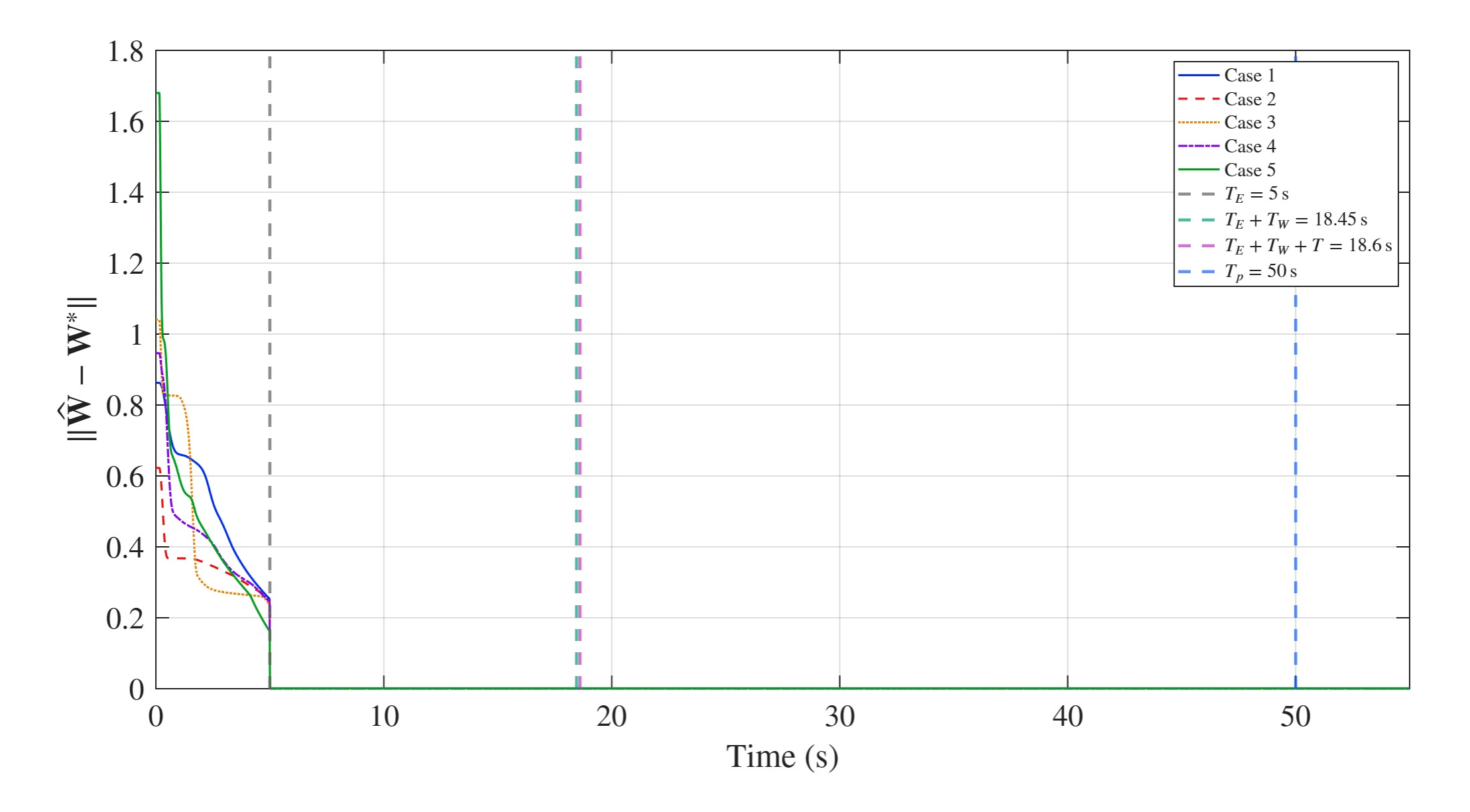}%
    }
    \caption{Evolution of the critic-weight estimation errors
    under five distinct initial conditions for
    \(T_p=50\,\mathrm{s}\).}
    \label{fig:Tp50_critic_error}
\end{figure}

\begin{figure}[H]
    \centering
    \scalebox{1}[0.85]{%
        \includegraphics[width=\columnwidth]{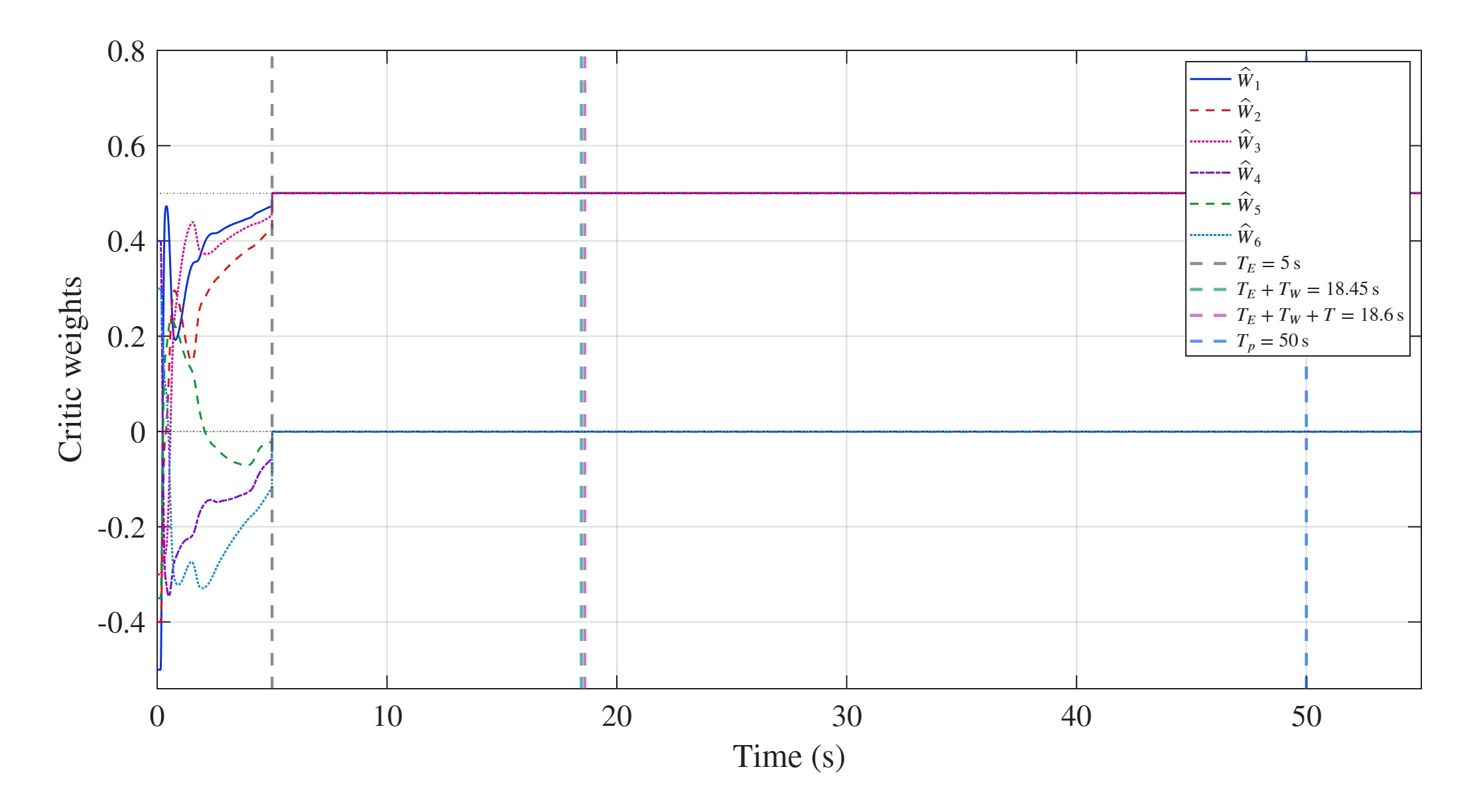}%
    }
    \caption{Trajectories of the six critic weights for
    Case 5 with \(T_p=50\,\mathrm{s}\).}
    \label{fig:Tp50_critic_weights}
\end{figure}

\begin{figure}[H]
    \centering
    \scalebox{1}[0.8]{%
        \includegraphics[width=\columnwidth]{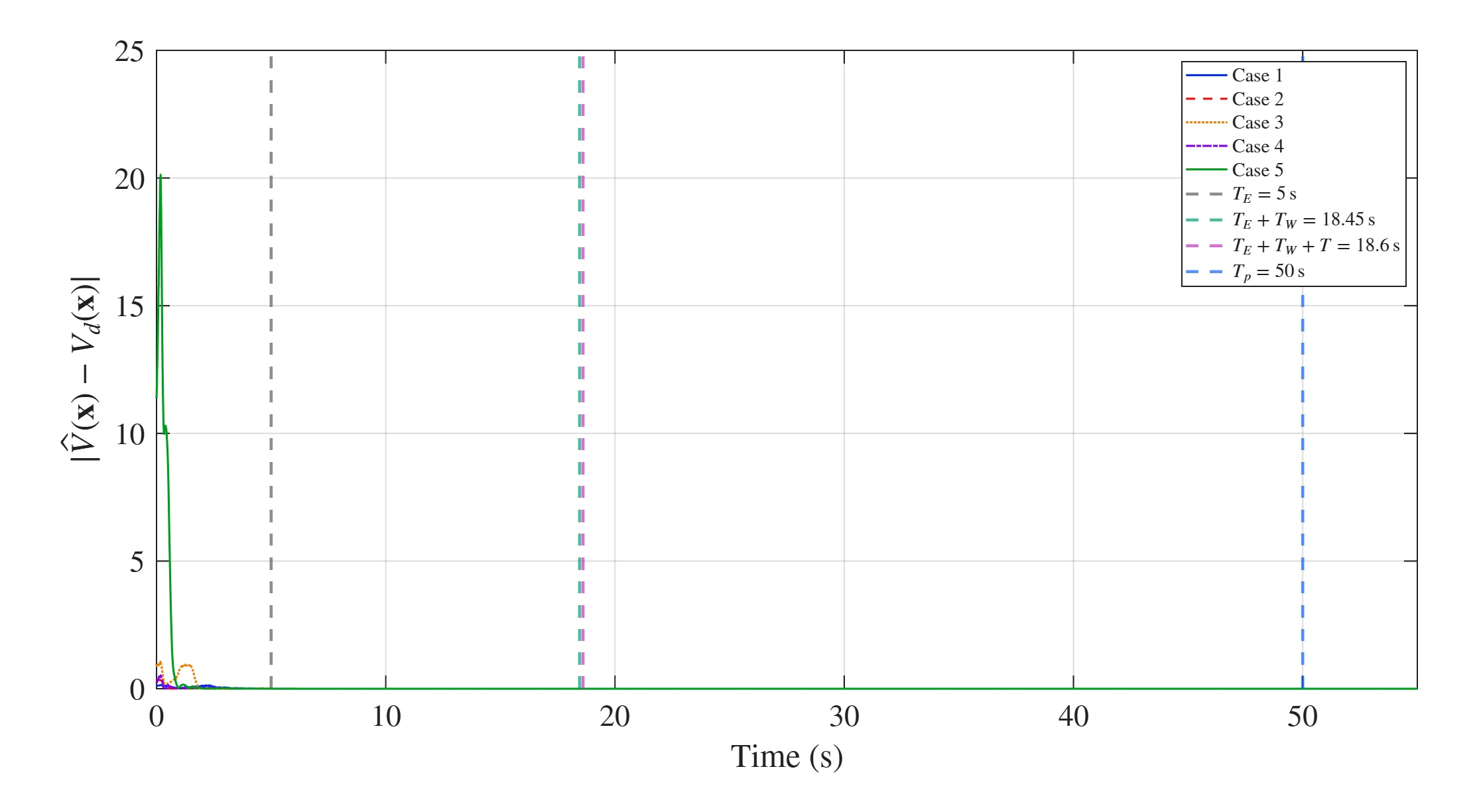}%
    }
    \caption{Value-function approximation errors under five
    distinct initial conditions for
    \(T_p=50\,\mathrm{s}\).}
    \label{fig:Tp50_value_error}
\end{figure}

\begin{figure}[H]
    \centering
    \scalebox{1}[0.8]{%
        \includegraphics[width=\columnwidth]{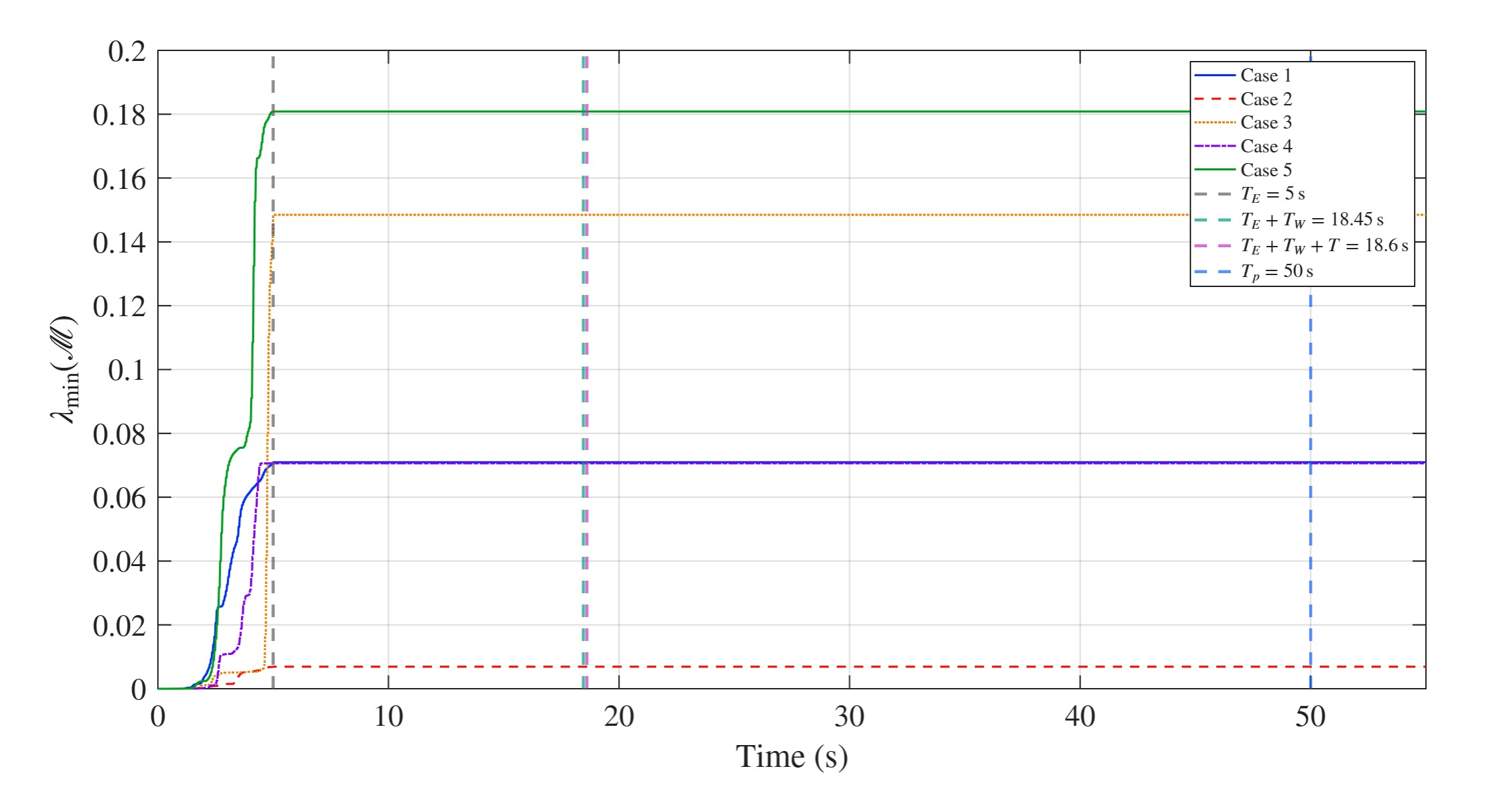}%
    }
    \caption{Evolution of the minimum eigenvalue of the
    normalized experience-replay Gramian during data acquisition
    and after replay-stack freezing for
    \(T_p=50\,\mathrm{s}\).}
    \label{fig:Tp50_replay_eigenvalue}
\end{figure}

\begin{table}[!t]
\centering
\caption{Replay informativity and data-dependent critic gains
for \(T_p=5\,\mathrm{s}\) and \(T_p=50\,\mathrm{s}\),
where \(i\in\{1,\ldots,5\}\) denotes the index of
the considered initial-condition pair.}
\label{tab:critic_gains_comparison}
\renewcommand{\arraystretch}{1.10}
\setlength{\tabcolsep}{3pt}
\footnotesize
\begin{tabular}{ccccc}
\hline
&
\multicolumn{2}{c}{\(T_p=5\,\mathrm{s}\)}
&
\multicolumn{2}{c}{\(T_p=50\,\mathrm{s}\)}
\\
\cline{2-3}\cline{4-5}
\textbf{Case}
&
\(\lambda_{\min}(\mathcal M_i)\)
&
\(\alpha_{W,i}\)
&
\(\lambda_{\min}(\mathcal M_i)\)
&
\(\alpha_{W,i}\)
\\
\hline
1
&
\(1.6460\!\times\!10^{-6}\)
&
\(6.1637\!\times\!10^{13}\)
&
\(7.0970\!\times\!10^{-2}\)
&
\(7.1616\!\times\!10^{5}\)
\\
2
&
\(4.1349\!\times\!10^{-7}\)
&
\(4.8954\!\times\!10^{14}\)
&
\(6.8915\!\times\!10^{-3}\)
&
\(2.3668\!\times\!10^{7}\)
\\
3
&
\(2.8232\!\times\!10^{-6}\)
&
\(2.7440\!\times\!10^{13}\)
&
\(1.4849\!\times\!10^{-1}\)
&
\(2.3663\!\times\!10^{5}\)
\\
4
&
\(2.5362\!\times\!10^{-7}\)
&
\(1.0191\!\times\!10^{15}\)
&
\(7.0616\!\times\!10^{-2}\)
&
\(7.2156\!\times\!10^{5}\)
\\
5
&
\(4.9858\!\times\!10^{-7}\)
&
\(3.6973\!\times\!10^{14}\)
&
\(1.8086\!\times\!10^{-1}\)
&
\(1.7604\!\times\!10^{5}\)
\\
\hline
\end{tabular}
\end{table}

Figures~\ref{fig:identification_error}
and~\ref{fig:identifier_weights} demonstrate rapid
drift identification and bounded identifier-weight
trajectories. For both prescribed deadlines,
Figs.~\ref{fig:Tp5_state} and~\ref{fig:Tp50_state}
show that all five state trajectories converge
well before $T_p$, accompanied by substantial
reductions in critic-weight and value-function
approximation errors, as illustrated in
Figs.~\ref{fig:Tp5_critic_error},
\ref{fig:Tp50_critic_error},
\ref{fig:Tp5_value_error}, and
\ref{fig:Tp50_value_error}.

During data acquisition, a fixed critic gain of
15 is employed before switching to the
data-dependent predefined-time gain after $T_E$.
As shown in Figs.~\ref{fig:Tp5_critic_weights}
and~\ref{fig:Tp50_critic_weights}, the critic
weights approach their target values following
this transition. More importantly,
Figs.~\ref{fig:Tp5_replay_eigenvalue}
and~\ref{fig:Tp50_replay_eigenvalue}, together
with Table~\ref{tab:critic_gains_comparison},
reveal the strong dependence of the sufficient
critic gain on replay-data informativeness.
For $T_p=5\,\mathrm{s}$, the minimum eigenvalues
range from $2.5362\times10^{-7}$ to
$2.8232\times10^{-6}$, requiring gains of
order $10^{13}$--$10^{15}$.
For $T_p=50\,\mathrm{s}$, the eigenvalues increase
to $6.8915\times10^{-3}$--$1.8086\times10^{-1}$,
while the required gains decrease to
order $10^5$--$10^7$.
This behavior is consistent with the theoretical
gain conditions. The shorter acquisition interval
($T_E=0.5\,\mathrm{s}$) and smaller replay stack
($M=7501$), compared with $T_E=5\,\mathrm{s}$
and $M=9701$, contribute to the observed
differences in data informativeness.
Nevertheless, these differences are specific
to the adopted sampling and excitation schemes;
a larger dataset does not necessarily yield
a better-conditioned or more informative Gramian.

From a practical perspective, predefined-time
control may require relatively large control
inputs during the initial transient, particularly
when a short convergence deadline is prescribed.
This characteristic is illustrated in
Figs.~\ref{fig:Tp5_learned_control}
and~\ref{fig:Tp50_learned_control}, where
pronounced initial control efforts are observed.
Although such responses are admissible in the
present simulations, they may be difficult to
realize in physical systems because of actuator
saturation and limited control authority.
Meanwhile, Figs.~\ref{fig:Tp5_behavior_control}
and~\ref{fig:Tp50_behavior_control} show
substantial probing-induced variations during
data acquisition. Multisine excitation with
sufficiently diverse frequencies and strongly
varying amplitudes can generate more distinguishable
regression data across sampling instants, potentially
improving the minimum eigenvalue of the replay
Gramian (\(\lambda_{\min}
(
\sum_{j=1}^{M}
\bar{\boldsymbol\Omega}_j
\bar{\boldsymbol\Omega}_j^\top
)
\geq
\alpha_{\mathrm{IE}}>0\)) and reducing the number of stored samples
needed to satisfy the finite-data rank condition.
However, increasing excitation amplitude alone
does not guarantee data informativeness, while
aggressive probing may violate actuator, state,
or safety constraints. This reveals a fundamental
trade-off among excitation richness, physical
limitations, and onboard memory capacity.

Importantly, explicit probing is not inherently
necessary if sufficiently informative operational
data are already available in the replay stack, this is the approach that Kokolakis used in his paper~\citep{Kokolakis2025Predefined}.
When historical exploration remains insufficient,
model-based simulated experience provides another
possible mechanism for enriching learning data
by evaluating Bellman errors at selected state-space
points, as investigated in
\cite{Kamalapurkar2016,Kamalapurkar2017}.
Such an extension requires accounting for model
uncertainty and verifying the corresponding
rank and stability conditions.

Overall, these results suggest that practical
implementation should prioritize informative
data acquisition over indiscriminately increasing
excitation amplitude or replay-buffer size.
A promising strategy is to combine
actuator-constrained excitation design,
online monitoring of replay-Gramian eigenvalues,
and information-based sample selection to maintain
finite-data informativity with limited memory.
The extremely large sufficient critic gains
observed for short deadlines further motivate
numerically robust gain selection and explicit
consideration of discretization effects.
These observations open several research
directions, including safe excitation design,
memory-efficient concurrent learning, and
model-assisted data enrichment, with the broader
objective of preserving prescribed-time learning
and closed-loop stability guarantees under
realistic sensing, computation, and actuator
constraints.


\section{Conclusion}

This paper has successfully developed a theoretical framework for predefined-time stabilization in optimal control of fully unknown nonlinear systems. The obtained results address limitations remaining in the work of~\citep{Kokolakis2025Predefined} and further extend the existing theory by formulating the problem in the integral reinforcement learning setting, which is more suitable for practical implementation and for systems with unknown dynamics. Although integral reinforcement learning was introduced by Modares~\citep{Modares2014}, the present study advances this direction by establishing an update law that guarantees prescribed-time convergence for a fully unknown nonlinear system, which is particularly relevant to cyber--physical systems with strict timing requirements. The paper also introduces a new proof framework for the coupled learning and closed-loop stability problem, and the theoretical analysis establishes predefined-time stability while the simulations confirm the main properties predicted by the theory. Two prescribed deadlines, \(T_p=5\,\mathrm{s}\) and \(T_p=50\,\mathrm{s}\), are used to examine the controller under aggressive and relaxed timing requirements and to provide guidance for practical tuning of the prescribed convergence time. At the same time, the \(T_p=50\,\mathrm{s}\) case shows that the actual convergence occurs considerably earlier than the assigned deadline, indicating that the current convergence-time bound remains conservative. Such conservatism may be undesirable when the convergence profile must closely match a prescribed time, and therefore motivates future work on less conservative analysis and design conditions that more tightly relate the assigned deadline to the actual closed-loop convergence time.

\appendix

\section{Proof of Lemma~\ref{lem:low_window}}
\label{app:proof_derivative_limited_window}

\begin{proof}
Fix any \(t\) and define
\(I_y:=\int_{t-T}^{t}y(\tau)\,d\tau\).
If \(I_y=0\), then \(y(\tau)=0\) almost everywhere on
\([t-T,t]\), and the claimed inequality follows immediately.
Hence, suppose \(I_y>0\).

Since \(y\) is absolutely continuous on the compact interval
\([t-T,t]\), it is continuous and therefore attains its maximum.
Let
\(M_y:=\max_{\tau\in[t-T,t]}y(\tau)\),
and choose
\(\tau_m\in[t-T,t]\)
such that
\(y(\tau_m)=M_y\).
Because \(0<\gamma_1<1\), one has
\(y^{\gamma_1}(\tau)
=
y(\tau)y^{\gamma_1-1}(\tau)
\geq
M_y^{\gamma_1-1}y(\tau)\)
for every \(\tau\in[t-T,t]\).
Consequently,
\(\int_{t-T}^{t}y^{\gamma_1}(\tau)\,d\tau
\geq
M_y^{\gamma_1-1}I_y\).

We now distinguish two cases. First, suppose that
\(M_y\leq 4I_y/T\).
Then
\(M_y^{\gamma_1-1}
\geq
(4I_y/T)^{\gamma_1-1}\),
because \(\gamma_1-1<0\). Hence,
\(\int_{t-T}^{t}y^{\gamma_1}(\tau)\,d\tau
\geq
4^{\gamma_1-1}T^{1-\gamma_1}I_y^{\gamma_1}\).
Since
\(\mu=(\gamma_1+1)/2\),
one has
\(\gamma_1-\mu=(\gamma_1-1)/2<0\).
Using \(I_y\leq\bar I\) therefore gives
\(I_y^{\gamma_1-\mu}
\geq
\bar I^{(\gamma_1-1)/2}\).
It follows that
\(\int_{t-T}^{t}y^{\gamma_1}(\tau)\,d\tau
\geq
4^{\gamma_1-1}T^{1-\gamma_1}
\bar I^{(\gamma_1-1)/2}
I_y^\mu\).

Consider now the complementary case
\(M_y>4I_y/T\).
We first show that
\(M_y\leq L_{\dot y}T\).
Indeed, if \(M_y>L_{\dot y}T\), then the derivative bound
\(|\dot y(\tau)|\leq L_{\dot y}\)
implies
\(y(\tau)\geq
M_y-L_{\dot y}|\tau-\tau_m|\)
for all \(\tau\in[t-T,t]\).
Since
\(|\tau-\tau_m|\leq T\),
this gives
\(y(\tau)\geq M_y-L_{\dot y}T>0\).
More precisely, integrating the Lipschitz lower envelope over the
interval yields
\(I_y\geq M_yT-\frac12L_{\dot y}T^2\).
Under \(M_y>L_{\dot y}T\), the latter satisfies
\(I_y>M_yT/2\), which contradicts
\(M_y>4I_y/T\).
Thus,
\(M_y\leq L_{\dot y}T\).

At least one of the intervals
\([t-T,\tau_m]\) and \([\tau_m,t]\)
has length no smaller than \(T/2\).
Since
\(M_y/(2L_{\dot y})\leq T/2\),
there exists a one-sided subinterval adjacent to \(\tau_m\)
of length
\(M_y/(2L_{\dot y})\)
that is contained in \([t-T,t]\).
For every \(\tau\) in this subinterval, the derivative bound gives
\(y(\tau)\geq
M_y-L_{\dot y}|\tau-\tau_m|
\geq M_y/2\).
Therefore,
\(I_y\geq
(M_y/2)(M_y/(2L_{\dot y}))
=
M_y^2/(4L_{\dot y})\),
and hence
\(M_y\leq 2\sqrt{L_{\dot y}I_y}\).

Using again
\(\int_{t-T}^{t}y^{\gamma_1}(\tau)\,d\tau
\geq
M_y^{\gamma_1-1}I_y\)
and \(\gamma_1-1<0\), one obtains
\(\int_{t-T}^{t}y^{\gamma_1}(\tau)\,d\tau
\geq
(2\sqrt{L_{\dot y}I_y})^{\gamma_1-1}I_y\).
Since
\(\mu=(\gamma_1+1)/2\),
this reduces to
\(\int_{t-T}^{t}y^{\gamma_1}(\tau)\,d\tau
\geq
2^{\gamma_1-1}
L_{\dot y}^{-(1-\gamma_1)/2}
I_y^\mu\).

Thus, in either case,
\(\int_{t-T}^{t}y^{\gamma_1}(\tau)\,d\tau\)
is bounded below by
\(\rho_{\gamma_1,T}I_y^\mu\),
where \(\rho_{\gamma_1,T}\) is the coefficient defined in the
statement of the lemma. Substituting
\(I_y=\int_{t-T}^{t}y(\tau)\,d\tau\)
completes the proof.
\end{proof}


\bibliographystyle{elsarticle-harv}
\bibliography{references}





\end{document}